\documentclass[11pt, a4paper]{article}
\usepackage{fullpage}

\usepackage[english]{babel}
\usepackage[utf8]{inputenc}

\usepackage[bookmarks]{hyperref}
\usepackage{amssymb,amsmath, amsthm, amsfonts}
\usepackage{xcolor}
\usepackage{hyperref}
\usepackage{stackengine}
\usepackage{graphicx,color,colordvi}
\usepackage{bbm}
\usepackage{varioref}
\usepackage{cleveref}
\usepackage{stmaryrd}
\usepackage{mathtools}
\usepackage{wrapfig}
\usepackage{tikz}
\usepackage{tcolorbox}
\usepackage{subcaption}
\usepackage[shortlabels]{enumitem}

\usepackage{physics}
\usepackage{braket}
\usetikzlibrary{shapes,calc,decorations.markings, decorations.pathmorphing, 3d, cd}

\tikzset{
  equation/.style={
    baseline={([yshift=-1.5ex]current bounding box.center)}
  },
  torus horizontal/.style = {
  decoration={
    markings,
    mark=at position 0.5 with {
            \draw (-2pt,-2pt) -- (2pt,2pt);
            \draw (2pt,-2pt) -- (-2pt,2pt);
    }}, decorate},
  torus vertical/.style = {
  decoration={
    markings,
    mark=at position 0.5 with {
            \draw (-2pt,-2pt) -- (2pt,2pt);
            \draw (-3pt,-2pt) -- (1pt,2pt);
    }}, decorate}
}

\newtheorem{Theo}{Theorem}
\newtheorem{Prop}[Theo]{Proposition}
\newtheorem{Coro}[Theo]{Corollary}
\newtheorem{Lemm}[Theo]{Lemma}

\theoremstyle{definition}
\newtheorem{Defi}[Theo]{Definition}
\newtheorem{Rema}[Theo]{Remark}

\newtheorem{Assumption}{Assumption}

\theoremstyle{remark}

\newcommand\unit{\mathbbm{1}}
\newcommand\dist {\operatorname{dist}}

\newcommand\complex{\mathbb{C}}

\newcommand\cB{\mathcal{B}}
\newcommand\cD{\mathcal{D}}
\newcommand{\cE}{{\mathcal{E}}}
\newcommand\cG{\mathcal{G}}
\newcommand\cH{\mathcal{H}}
\newcommand\cK{\mathcal{K}}
\newcommand\cL{\mathcal{L}}
\newcommand{\cO}{{\mathcal{O}}}
\newcommand{\cP}{{\mathcal{P}}}
\newcommand{\cR}{{\mathcal{R}}}
\newcommand\cS{\mathcal{S}}

\newcommand{\cV}{{\mathcal{V}}}
\newcommand\id{\operatorname{Id}}

\newcommand\gap{\operatorname{gap}}

\usepackage{xcolor, soul}

\usepackage{graphicx} % Required for inserting images

\usepackage[style=phys, sorting=nyt, eprint=true]{biblatex}
\usepackage[blocks]{authblk}

\title{Spectral Gap Bounds for Quantum Markov Semigroups\\ via Correlation Decay}
\author[1,2,3]{Angelo Lucia\thanks{ {\tt angelo.lucia@polimi.it} , \, \url{https://orcid.org/0000-0003-1709-1220}}}
\author[2,3]{David Pérez-García\thanks{ {\tt dperezga@ucm.es} , \, \url{https://orcid.org/0000-0003-2990-791X}}}
\author[4]{Antonio Pérez-Hernández\thanks{ {\tt antperez@ind.uned.es} , \, \url{https://orcid.org/0000-0001-8600-7083}  }   }

\affil[1]{
Dipartimento di Matematica, Politecnico di Milano, 20131 Milano, Italy
}

\affil[2]{
    Instituto de Ciencias Matemáticas,
    28049 Madrid, Spain
    }
    
\affil[3]{
    Departamento de Análisis Matemático y Matemática Aplicada, 
    Universidad Complutense de Madrid,\newline 28040 Madrid, Spain
}

\affil[4]{
    Departamento de Matemática Aplicada I,
    Universidad Nacional de Educación a Distancia,\newline 28040 Madrid, Spain
}

\begin{document}

\maketitle

{\small \noindent \textbf{Acknowledgements:} \footnotesize \\
The authors acknowledge financial support from grants PID2020-113523GB-I00, PID2023-146758NB-I00 and CEX2023-001347-S, funded by MICIU/AEI/10.13039/501100011033. D.\,P-G.~acknowledges support from grant  TEC-2024/COM-84-QUITEMAD-CM, funded by Comunidad de Madrid.
A.\,L.~acknowledges support from  grant RYC2019-026475-I funded by MICIU/AEI/10.13039/501100011033 and ``ESF Investing in your future'', and by ``Programma per Giovani Ricercatori Rita Levi Montalcini'' funded by the Italian Ministry of University and Research (MUR). This work has been financially supported by the Ministry for Digital Transformation and the Civil Service of the Spanish Government through the QUANTUM ENIA project call – Quantum Spain project, and by the European Union through the Recovery, Transformation and Resilience Plan – NextGenerationEU within the framework of the Digital Spain 2026 Agenda.\\
}

{\small\noindent \textbf{Abstract}:
Starting from an arbitrary full-rank state of a lattice quantum spin system, we define a \emph{canonical purified Hamiltonian} and characterize its spectral gap in terms of a spatial mixing condition (or correlation decay) of the state. When the state considered is a Gibbs state of a  local,  commuting  Hamiltonian at positive temperature, we show that the spectral gap of the canonical purified Hamiltonian provides a lower bound to the spectral gap of a large class of reversible generators of quantum Markov semigroup, including local and ergodic Davies generators.
As an application of our construction, 
we show that the mixing condition is always satisfied for any finite-range 1D model, as well as by Kitaev's quantum double models.

%\noindent \emph{2020 MSC: 81P73; 82C10; 81S22; 81V27; 37A25}.
}

\newpage
\tableofcontents

\section{Introduction}

In the analysis of the behavior of classical spin system, a key role has been played by the observation that certain ``dynamical'' features (i.e., properties related to the evolution, such as the convergence rate to equilibrium) are related to certain ``static'' properties of the invariant state (i.e., various kinds of decay of correlations). Uncovering and understanding these type of connections has been central in many seminal results \cite{Aizenman1987, Zegarlinski1990, Stroock1992, Martinelli1993, Martinelli1994A, Martinelli1994B, Martinelli1999, Cesi2001}.

It is therefore not surprising that the same approach has been attempted in order to study the behavior of dissipative quantum spin systems \cite{Kastoryano2013,KB_2016,Bardet2021,Bardet2023,Bardet2024}. While the intuition supporting this approach still indicates that short-range correlations in invariant states should be related to quick convergence of the evolution, finding the exact generalizations of the classical results to the quantum setting has been challenging.

In \cite{KB_2016}, the authors show that a specific class of generators of dissipative semigroups, called Davies generators, associated to the Gibbs state of a commuting, local, finite-range Hamiltonian, have a spectral gap independent of system size if and only if the Gibbs state satisfies a condition they call \emph{strong clustering}. Note that, by standard arguments, bounds on the spectral gap can be related to mixing time estimates on the semigroup.
While conceptually their result fulfills the program of relating static and dynamical properties of these models, the strong clustering condition is given in terms of local conditional expectations, and is therefore difficult to compute explicitly for specific models.

In this work, we take a different approach: starting from an arbitrary full-rank state $\sigma$ of a lattice quantum spin system on a finite volume $\Lambda$, we construct a self-adjoint operator $\mathbf{H}$ on a doubled Hilbert space, which we denote the \emph{canonical purified Hamiltonian} (canonical in the sense that it only depends on $\sigma$), which has a purification of $\sigma$ as unique ground state (see Section~\ref{sec:canonicalPurifiedHamiltonian}). We argue that this canonical purified Hamiltonian connects dynamical and static properties of a large class of dissipative evolutions having $\sigma$ as invariant state.
On the one hand, in Section~\ref{sec:OpenQuantumSystems} we show that it is possible to compare $\mathbf{H}$ with a large class of generators of $\sigma$-reversible semigroups, which notably include the case of locally ergodic Davies generators, obtaining spectral gap bounds (and consequently mixing time estimates) for these generators in terms of the gap of $\mathbf{H}$.
On the other hand, in Section~\ref{sec:Estimating_the_gap}, we show that a positive bound on the gap of $\mathbf{H}$ can be obtained, under mild technical assumptions, from a spatial mixing condition on the state $\sigma$. We also show a weak reverse statement: that when $\sigma$ is the Gibbs state of a local, finite-range, commuting Hamiltonian, then the decay of the spatial mixing condition is implied when $\mathbf{H}$ is \emph{locally gapped} (i.e., the spectral gap over any finite volume is uniformly bounded away from zero).

To be more concrete, for any partition of our set $\Lambda$ into four disjoint subsets $\Lambda =ABCD:=A \sqcup B \sqcup C \sqcup D$ , we consider the quantity
\begin{equation}\label{eq:spatial-mixing-condition}
    \Delta_\sigma(A : C | D) :=
    \sup_{R_{AD}, Q_{CD}} \abs{\Tr_{ACD}[(\sigma_{ACD} - \sigma_{AD} \sigma_{D}^{-1} \sigma_{DC}) Q_{CD}^{\dagger} R_{AD}]} \,,
\end{equation}
where the supremum is taken over $R_{AD}$ supported on $AD$, $Q_{CD}$ supported on $CD$, and normalized in such a way that 
$ \operatorname{Tr} (\sigma R_{AD}^\dag R_{AD}) = \operatorname{Tr} (\sigma  Q_{CD}^\dag Q_{CD})  = 1.$
We show that, under certain technical assumptions, $\mathbf{H}$ is (locally) gapped if and only if $\Delta_\sigma(A : C | D)$ decays sufficiently fast in the size of $B$ for certain type of decompositions. For instance, when considering as $\sigma$ the Gibbs state of a 2D local commuting Hamiltonian, we need to examine the three type of decompositions exhibited in Figure~\ref{fig:decompositionTorusQDM}. We also provide upper and lower bounds to $\Delta_\sigma(A:C|D)$, which are easier to compute, allowing us to provide numerous examples of cases in which $\mathbf{H}$ is gapped.

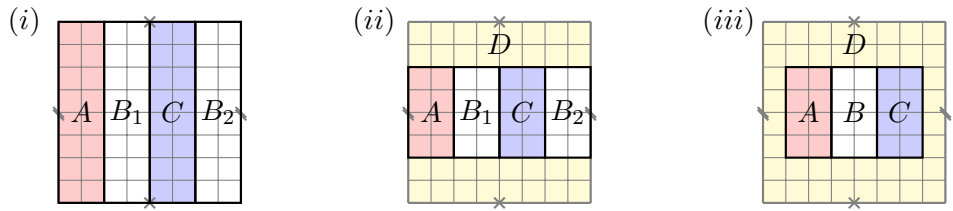
\begin{figure}[ht]
\centering
\begin{tikzpicture}[equation,scale=0.3]
    \node at (-1.5, 8) {$(i)$};
    \filldraw[red!20] (0,0) rectangle (2,8);
    \filldraw[blue!20] (4,0) rectangle (6,8);
    \draw[gray, thin] (0,0) grid (8,8);
    \draw[gray,thick, postaction=torus horizontal] (0,0) -- (8,0);
    \draw[gray,thick, postaction=torus horizontal] (0,8) -- (8,8);
    \draw[gray,thick, postaction=torus vertical] (0,0) -- (0,8);
    \draw[gray,thick, postaction=torus vertical] (8,0) -- (8,8);
    \draw[black, thick] (0,0) rectangle (2,8);
    \node at (1, 4) {$A$};
    \draw[black, thick] (2,0) rectangle (4,8);
    \node at (3, 4) {$B_{1}$};
    \draw[black, thick] (4,0) rectangle (6,8);
    \node at (5, 4) {$C$};
    \draw[black, thick] (6,0) rectangle (8,8);
    \node at (7, 4) {$B_{2}$};
    \end{tikzpicture}
    %%%%%%%%%%%%%%%%%%%%%%%%%%%
    \hspace{1cm} %%%%%%%%%%%%%%%%%%%%%%%%%%%%%%
\begin{tikzpicture}[equation,scale=0.3]
    \node at (-1.5, 8) {$(ii)$};
    \filldraw[yellow!20] (0,0) rectangle (8,2);
    \filldraw[yellow!20] (0,6) rectangle (8,8);
    \filldraw[red!20] (0,2) rectangle (2,6);
    \filldraw[blue!20] (4,2) rectangle (6,6);
    \draw[gray, thin] (0,0) grid (8,8);
    \draw[gray,thick, postaction=torus horizontal] (0,0) -- (8,0);
    \draw[gray,thick, postaction=torus horizontal] (0,8) -- (8,8);
    \draw[gray,thick, postaction=torus vertical] (0,0) -- (0,8);
    \draw[gray,thick, postaction=torus vertical] (8,0) -- (8,8);
    \draw[black, thick] (0,2) rectangle (2,6);
    \node at (1, 4) {$A$};
    \draw[black, thick] (2,2) rectangle (4,6);
    \node at (3, 4) {$B_{1}$};
    \draw[black, thick] (4,2) rectangle (6,6);
    \node at (5, 4) {$C$};
    \draw[black, thick] (6,2) rectangle (8,6);
    \node at (7, 4) {$B_{2}$};
    \node at (4, 7) {$D$};
    \end{tikzpicture}
    %%%%%%%%%%%%%%%%%%%%%%%%%%%
    \hspace{1cm} %%%%%%%%%%%%%%%%%%%%%%%%%%%%%%
    \begin{tikzpicture}[equation,scale=0.3]
    \node at (-1.5, 8) {$(iii)$};
    \filldraw[yellow!20] (0,0) rectangle (8,8);
    \filldraw[white] (1,2) rectangle (7,6);
    \filldraw[red!20] (1,2) rectangle (3,6);
    \filldraw[blue!20] (5,2) rectangle (7,6);
    \draw[gray, thin] (0,0) grid (8,8);
    \draw[gray,thick, postaction=torus horizontal] (0,0) -- (8,0);
    \draw[gray,thick, postaction=torus horizontal] (0,8) -- (8,8);
    \draw[gray,thick, postaction=torus vertical] (0,0) -- (0,8);
    \draw[gray,thick, postaction=torus vertical] (8,0) -- (8,8);
    \draw[black, thick] (1,2) rectangle (3,6);
    \node at (2, 4) {$A$};
    \draw[black, thick] (3,2) rectangle (5,6);
    \node at (4, 4) {$B$};
    \draw[black, thick] (5,2) rectangle (7,6);
    \node at (6, 4) {$C$};
    \node at (4, 7) {$D$};
    \end{tikzpicture}
    \caption{Possible decompositions of the torus into four subregions $\Lambda = ABCD$. In $(i)$ and $(ii)$ the subset $B$ is formed by the union of $B_{1}$ and $B_{2}$, whereas in $(i)$ the set $D$ is empty.}
    \label{fig:decompositionTorusQDM}
\end{figure}

As we will see, in cases where $D=\emptyset$, the quantity $\Delta(A:C|\emptyset)$ serves as an upper bound of the usual operator correlation function $\operatorname{Corr}_{\sigma}(A:C)$, so the exponential decay of the former yields the absence of a thermal phase transition provided that $\sigma$ is a Gibbs state.

Note that the quantity $\Delta_\sigma(A:C|D)$ is zero if, and only if, 
\begin{equation}
    \sigma_{ADC} = \sigma_{AD} \sigma_{D}^{-1} \sigma_{DC}\,.
\end{equation}
This expression bears a clear resemblance to the Petz recovery condition 
\begin{equation}
\sigma_{ADC} = \sigma_{DC}^{1/2} \sigma_{D}^{-1/2}\sigma_{AD} \sigma_{D}^{-1/2} \sigma_{DC}^{1/2}\,,
\end{equation}
that characterizes \emph{quantum Markov chains} between $A-D-C$, i.e., states for which the conditional mutual information $I(A:C|D)$ is zero \cite{hayden2004}. Actually, states satisfying $\sigma_{ADC} = \sigma_{AD} \sigma_{D}^{-1} \sigma_{DC}$ are called \emph{BS-quantum Markov chains} between $A-D-C$. These states are characterized by the vanishing of several recently introduced measures of conditional mutual information based on the Belavkin-Staszewski (BS) relative entropy \cite{bluhm25} \cite{Gondolf24}. It is known that every BS-quantum Markov chain is a quantum Markov chain, but the converse does not hold \cite{bluhm25}. Thus, we can interpret $\Delta_\sigma(A:C|D)$ as a measure of correlations between regions $A$ and $C$ conditioned on $D$. However, we will not explore the precise connection with other measures of conditional information in this work.

As an application of our theory, we are able to verify a fast decay of $\Delta_\sigma(A:C|D)$ when $\sigma$ is the Gibbs state of any local, finite-range Hamiltonian in 1D (not necessarily commuting), or for the case of 2D quantum double models introduced by Kitaev~\cite{Kitaev2003}. These results are contained in Section~\ref{sec:1D-chains} and Section~\ref{sec:quantum-double}. Together with the results from Section~\ref{sec:OpenQuantumSystems}, this implies a spectral gap and mixing time estimate for any locally ergodic Davies generator for these models, with the extra restriction that we only know how to construct local Davies generators in the case of Gibbs states of commuting Hamiltonians.

Notwithstanding this restriction on the class of states for which  local Davies generators can be defined, the construction of the canonical purified Hamiltonian $\mathbf{H}$ and its spectral gap analysis work in a more general setting, i.e., even for Gibbs state of long range and/or non-commuting interactions in higher dimensions (of course, obtaining bounds on $\Delta_\sigma(A:C|D)$ might be extremely challenging in these settings). We believe that our results might be useful also to study these states, for which it has been hard to define generators having them as invariant states. 

\subsection{Notation}

Let $\mathcal{H}$ be a finite-dimensional Hilbert space and $\cB(\cH)$ the space of bounded linear operators on $\cH$. The trace of $Q \in \cB(\cH)$ is denoted by $\Tr(Q)$, while its adjoint is denoted by $Q^\dag$. The operator norm of an operator $Q \in \cB(\cH)$ is denoted by $\norm{Q}_{\infty}$, or simply $\|Q\|$ when no confusion arises, whereas for every $p \in [1, \infty)$ the Schatten $p$-norm is defined by
\[ \norm{Q}_{p} = (\operatorname{Tr}|Q|^{p})^{1/p} \quad \mbox{ where } \quad |Q|:= (Q^{\dagger} Q)^{1/2}\,. \]
In particular, for $p=2$ the norm is the norm associated to the Hilbert-Schmidt scalar product on $\cB(\cH)$ given by
\begin{equation}
\braket{Q, S}_{HS} = \operatorname{Tr}(Q^{\dagger} S) \quad , \quad Q,S \in \cB(\cH)\,. 
\end{equation}
Each state $\omega$ of $\cB(\cH)$ can be represented as $\omega(Q) = \Tr(\rho\, Q)$, for some density operator $\rho \in \cS(\cH)$, where
\[ \cS(\cH) = \{ \rho \in \cB(\cH) \colon \rho^{\dagger} = \rho, \rho \geq 0, \Tr(\rho)=1 \}\,. \]
We will also work with another scalar product on $\cB(\cH)$. Let $\sigma \in \cS(\cH)$ be a full-rank density operator, i.e. $\sigma > 0$. We then define the \emph{GNS (or Liouville) scalar product} associated to $\sigma$ as
\begin{equation}
    \braket{Q, S}_\sigma = \operatorname{Tr}(\sigma\, Q^\dag S) \quad ,\quad Q,S \in \cB(\cH)\,, 
\end{equation}
whose corresponding norm we denote by $\norm{\cdot}_\sigma$.

For each operator $T:\cB(\cH) \to \cB(\cH)$, we denote by $T^{\ast}$ the dual or adjoint operator with respect to the Hilbert-Schmidt scalar product on $\cB(\cH)$, that is, $T^{\ast}$ is the only operator satisfying
\[ \braket{A, T(B)}_{HS} = \braket{T^{\ast}(A), B}_{HS} \quad \text{ for every } A,B \in \cB(\cH)\,. \]
Moreover, $\norm{T}$ will denote the operator norm of $T$ when $\cB(\cH)$ is equipped with the Hilbert-Schmidt norm.

We are going to work with Hilbert spaces representing states of a quantum spin system on a finite volume. We will denote by $\Lambda$ a finite set: as concrete examples to keep in mind, consider $\Lambda = \mathbbm{Z}_N^D$ or $\Lambda = [-N,N]^{\times D} \subset \mathbbm{Z}^D$, for some given integers $D$ and $N$. If $X \subset \Lambda$, we will denote $X^{c} = \Lambda \setminus X$. If $A_{1}, \ldots, A_{n}$ are pairwise disjoint subsets of $\Lambda$, we will write $A_{1}A_{2}\ldots A_{n}$ to denote their disjoint union $A_{1} \sqcup A_{2} \sqcup \ldots \sqcup A_{n}$. In particular, if these sets form a partition of $\Lambda$, we will write $\Lambda = A_{1} \ldots A_{n}$

To each $x \in \Lambda$ we associate a finite-dimensional Hilbert space $\cH_{x}$, that we will refer to as the \emph{local} Hilbert space associated to the site $x$. For simplicity, we will assume that they all have the same \emph{local} dimension $d$, that is $\cH_{x} \equiv \complex^{d}$ for every site $x$. The space $\cB(\cH_{x})$ will be refereed as the local space of observables. For every subset $X \subset \Lambda$ we then define a local Hilbert space $\cH_{X} \equiv \otimes_{x \in X} \cH_{x} \equiv (\complex^d)^{\otimes |X|}$, as well as a space of observables $\cB(\cH_{X}) \equiv \otimes_{x \in X} \cB(\cH_{x})$. If $X \subset Y$ are subsets of $\Lambda$, then we can embed $\cB(\cH_X)$ into $\cB(\cH_Y)$ by identifying it with $\cB(\cH_X) \otimes \unit_{Y\setminus X}$, where $\unit_{Y\setminus X}$ is the unit of $\cB(\cH_{Y\setminus X})$, i.e., $\unit_{Y\setminus X} \equiv \otimes_{x\in Y\setminus X} \unit_{\cH_x}$. This is a unital and $\norm{\cdot}_{\infty}$-preserving inclusion. Given $Q \in \cB(\cH_\Lambda)$ we will say that $Q$ is supported in $X \subset \Lambda$ if $Q \in \cB(\cH_X)$. The \emph{support} of $Q$ is the smallest subset with this property. If $\sigma$ is a full-rank state on $\cH_{\Lambda}$, and $X$ a subset of $\Lambda$, we will denote by $\sigma_{X^c}$ the partial trace over $X$ of $\sigma$, i.e., $\sigma_{X^c} = \Tr_X (\sigma) \in \cB(\cH_{X^{c}})$.

We will sometimes, but not always, consider states $\sigma$ which are Gibbs state of local Hamiltonians. In these cases, we suppose to have an assignment $\Phi$ from any $X\subset \Lambda$ to $\Phi_X\in \cB(\cH_X)$ such that $\Phi_X$ is self-adjoint. We then define for any region $R\subset \Lambda$:
\begin{equation}
    H_R := \sum_{X \colon X \subset R} \Phi_X \quad , \quad H_{R}^{\partial} := H_{\Lambda} -H_{\Lambda \setminus R} = \sum_{X \colon X \cap R \neq \emptyset} \Phi_{X}.
\end{equation}
We call $\Phi$ an \emph{interaction}, and we say that $\Phi$ has \emph{finite range} with range $r > 0$ if $\Phi_X = 0$ whenever $\operatorname{diam}(X) > r$. In this case, we denote
\[\overline{R} = \bigcup \{ X \colon X \cap R \neq \emptyset, \Phi_{X} \neq 0\} \,\, , \quad\partial R = \overline{R} \setminus R\,, \]
and
\[ \| \Phi\| = \sup_{x \in \Lambda} \sum_{X \ni x}\| \Phi_{X}\|_\infty .\]
We will say that $\Phi$ is \emph{commuting} if $\Phi_X$ and $\Phi_{X'}$ commute for every pair of regions $X$, $X'$. 
The \emph{Gibbs (or thermal) state} of the system at (inverse) temperature $\beta >0$ is then defined as $\sigma_{\beta} = e^{-\beta H_{\Lambda}}/Z_{\beta}$ where $Z_{\beta} = \operatorname{Tr}(e^{-\beta H_{\Lambda}})$.

\subsection{Canonical purified Hamiltonian}
\label{sec:canonicalPurifiedHamiltonian}

Let $\sigma$  be a positive semidefinite operator on $\cH_{\Lambda}$. For every $X \subset \Lambda$,  let us define
\begin{equation}\label{equa:definingWX}
W_{X}:=  \mathcal{B}(\mathcal{H}_{X^{c}})\, \sigma^{1/2} = \{ (\mathbbm{1}_{X} \otimes O) \sigma^{1/2} \colon O \in \mathcal{B}(\mathcal{H}_{X^{c}}) \} \subset \cB(\cH_\Lambda) \,.\end{equation}
We denote by $\Pi_{X}: \cB(\cH_\Lambda) \to \cB(\cH_\Lambda)$ the orthogonal projection (with respect to the Hilbert-Schmidt scalar product) onto $W_{X}$. As usual, we write $\Pi^{\perp}=\operatorname{Id}-\Pi$. If $X= \{ x\}$ is unipunctual, we will simply write $\Pi_{x} = \Pi_{\{ x\}}$. Observe that if $X \subset Y$, then $W_{Y} \subset W_{X}$, and thus $\Pi_{Y} \leq \Pi_{X}$ and $\Pi_{X}^{\perp} \leq \Pi_{Y}^{\perp}$.

\begin{Defi}\label{Defi:DaviesHamiltonian}
With the previous notation, we define the \emph{canonical purified Hamiltonian} associated to $\sigma$ as
\begin{equation}
\mathbf{H}_{X}: \cB(\cH_\Lambda) \to \cB(\cH_\Lambda) \quad , \quad \mathbf{H}_{X} = \sum_{x \in X} \Pi_{x}^{\perp} \quad , \quad X \subset \Lambda\,.     
\end{equation}
When $X= \Lambda$, we will simply write $\mathbf{H} = \mathbf{H}_{\Lambda}$. 
\end{Defi}

For every $X \subset \Lambda$, note that $\mathbf{H}_{X}$ is positive-semidefinite with respect to the Hilbert-Schmidt scalar product, and it is frustration free in the sense that its ground state space is given by its kernel, being moreover $\ker(\mathbf{H}_{X}) = W_{X}$. In particular, the ground state space of $\mathbf{H}_{\Lambda}$ is $W_\Lambda = \mathbb{C} \sigma^{1/2}$.

One significant advantage of the canonical purified Hamiltonian is that it admits an explicit expression for the orthogonal projection $\Pi_{X}$ onto the ground state space $W_{X}$ of $\mathbf{H}_{X}$, as stated in the following proposition.

\begin{Prop}\label{Prop:explicitprojectionCanonicalPurified}
If $\sigma$ is full rank, then the orthogonal projection $\Pi_{X}$ onto $W_{X}$ is given by
\begin{equation}\label{eq:explicitprojectionCanonicalPurified}
\Pi_{X}(Q)= \Tr_{X}(Q \sigma^{1/2}) \sigma_{X^{c}}^{-1} \sigma^{1/2} \,, \quad Q \in \cB(\cH_\Lambda)\,.  
\end{equation} 
\end{Prop}

\begin{proof}
By the well-known Hilbert projection theorem, there exists a unique orthogonal projection $\Pi_{X}$ onto $W_{X}$ that associates each $Q \in \cB(\cH_\Lambda)$ with the unique element $\Pi_{X}(Q) \in W_{X}$ such that $Q - \Pi_{X}(Q)$ is orthogonal to $W_{X}$. Note that  we can express the orthogonality condition $Q - \Pi_{X}(Q) \perp W_{X}$ explicitly as follows: for every $O \in \cB(\cH_{X^{c}})$
\[
0=\langle O \sigma^{1/2}, Q - \Pi_{X}(Q) \rangle_{HS}  =  \operatorname{Tr}(\sigma^{1/2}O^{\dagger}Q) - \operatorname{Tr}(\sigma^{1/2} O^{\dagger}\Pi_{X}(Q)) \,, 
\]
or equivalently
\[ \operatorname{Tr}(O^{\dagger}Q\sigma^{1/2}) = \operatorname{Tr}(O^{\dagger}\Pi_{X}(Q)\sigma^{1/2})  \,.\]
Next, we utilize the property 
\begin{equation}\label{equa:auxOrthogonalprojection1}
\Tr(O_{1}O_{2}) = \Tr_{X^{c}}(O_{1}\Tr_{X}(O_{2})) \qc \mbox{for every } O_{1} \in \cB(\cH_{X^{c}}), \, O_{2} \in \cB(\cH_\Lambda)\,,
\end{equation}
to rewrite the previous expression as follows
\[ \Tr_{X^{c}}(O^{\dagger}\Tr_{X}(Q\sigma^{1/2})) = \Tr_{X^{c}}(O^{\dagger}\Tr_{X}(\Pi_{X}(Q)\sigma^{1/2}))\,. \]
But since $O \in \cB(\cH_{X^{c}})$ is arbitrary, we conclude that $Q - \Pi_{X}(Q) \perp W_{X}$ if and only if 
\[ \Tr_{X}(Q\sigma^{1/2}) =  \Tr_{X}(\Pi_{X}(Q)\sigma^{1/2})\,.\]
Now, we observe that the condition $\Pi_{X}(Q) \in W_{X}$ implies that $\Pi_{X}(Q) \sigma^{-1/2}$ belongs to $\cB(\cH_{X^{c}})$, which allows us to rewrite the last expression as
\[ \Tr_{X}(Q\sigma^{1/2}) = 
\Tr_{X}(\Pi_{X}(Q)\sigma^{1/2}) = \Tr_{X}(\Pi_{X}(Q)\sigma^{-1/2}\sigma) = \Pi_{X}(Q)\sigma^{-1/2} \sigma_{X^c}.
\]
From here, we solve for $\Pi_{X}(Q)$ to get the claimed expression for $\Pi_{X}(Q)$.
\end{proof}
\begin{Defi}
    We denote by $\gap(\mathbf{H}_{X})$ the \emph{spectral gap} of the operator $\mathbf{H}_{X}$, namely the difference between the two smallest eigenvalues of $\mathbf{H}_{X}$ (not counting multiplicities), or in other words the smallest non-zero eigenvalue of $\mathbf{H}_{X}$.
\end{Defi}

\section{Open systems and dynamics}
\label{sec:OpenQuantumSystems}

In this section, we recall some basics on quantum Markov semigroups, and then we will show that, under certain assumptions, the spectral gap of the generator of a quantum Markov semigroup can be bounded in terms of the spectral gap of the canonical purified Hamiltonian $\mathbf{H}$. For details on the definitions and basic properties of quantum Markov semigroups, we refer to \cite{Qchannels, Alicki_L_07}.

\subsection{Quantum Markov semigroups}
Let us consider a quantum system with associated space of states $\cH$ (finite-dimensional Hilbert space).

\begin{Defi}
A \emph{quantum Markov semigroup} (QMS for short) on $\cH$ is a one-parameter family $(T_{t})_{t \geq 0}$ of linear operators $T_{t}:\cB(\cH) \to \cB(\cH)$  satisfying the following properies: (i) $T_{t+s} = T_{t} \circ T_{s}$ for every $t,s \geq 0$, (ii) $T_{0} = \id$, (iii) $t \to T_{t}$ is continuous (in the operator norm), (iv) $T_{t}$ is completely positive and  $T_{t}(\mathbbm{1}) = \mathbbm{1}$ (unit-preserving) for every $t \geq 0$.
\end{Defi}

This type of semigroups model the evolution of observables $t \mapsto A(t) := T_{t}(A)$, whereas the dual semigroup $(T_{t}^{\ast})_{t \geq 0}$, whose elements are completely positive and trace-preserving, describes the evolution of states $t \mapsto \rho(t) = T_{t}^{\ast}(\rho)$. Both pictures are equivalent via the duality relation:
\[ \Tr(\rho A(t)) = \langle \rho, T_{t}(A) \rangle_{HS} = \langle T_{t}^{\ast}(\rho), A \rangle_{HS}  = \Tr(\rho(t) A)\,. \]

The continuous semigroup structure ensures that the evolution is memoryless over time (Markovian) and differentiable. As a consequence, it can be described as $T_{t} = e^{t\cL}$ for a (unique) superoperator $\cL: \cB(\cH) \longrightarrow \cB(\cH)$ called the \emph{generator}, that is characterized by the quantum Markovian \emph{master equation}
\[ \frac{d}{dt} T_{t} = \cL \circ T_{t}\,. \]
In the particular case of quantum Markov semigroups, the generator is called the \emph{Liouvillian} and admits the following well-known characterization.

\begin{Theo}[Structure of QMS generators]\label{Theo:structureQMSgenerators}
A superoperator $\cL:\cB(\cH) \to \cB(\cH)$ is a generator of a quantum Markov semigroup if and only if there exist a Hermitian operator $H$ and a set of operators $\{ L_{j}\}_{j=1}^{D}$ where $D = \dim(\cH)^{2}$ such that
\begin{equation}\label{eq:generator-derivation-dissipation}
\cL(Q)  = i \delta(Q) + \cD(Q)
\end{equation}
where $\delta(Q) = \comm{H}{Q}$ is a derivation, i.e. $\delta(Q_1 Q_2) = \delta(Q_1) Q_2 + Q_1\delta(Q_2)$, and
\begin{equation}
\cD(Q) = \sum_{j=1}^{D} L_{j}^{\dagger} Q L_{j} - \frac{1}{2}\{ L_{j}^{\dagger}L_{j}, Q\}_{+}=\frac{1}{2}\sum_{j=1}^{D}\qty(L_{j}^{\dagger}[Q,L_{j}] + [L_{j}^{\dagger},Q]L_{j})\,, \end{equation}
being $\{ a,b\}_{+} := ab+ba$ the anticommutator. The previous expression is called the Lindbladian form of the generator, and the superoperator $\cD$ is called the \emph{dissipative} term.
\end{Theo}

We say that a state $\sigma$ is a \emph{steady state} of the quantum Markov semigroup, or of its generator $\cL$, if $e^{t \cL^{\ast}}(\sigma) = \sigma$ for every $t$, or equivalently, if $\cL^{\ast}(\sigma) = 0$. If $\sigma$ is the only steady state of the semigroup, that is, if $\ker{(\cL^{\ast})} = \mathbb{C} \sigma$, then the semigroup, or its generator, is said to be \emph{primitive}. In this case, it can be shown that $\rho(t)$ converges to $\sigma$ as $t \rightarrow \infty$ for every initial state $\rho$.

\subsection{Detailed-balance and spectral gap}
We first recall the definition of quantum detailed balance, or reversibility \cite{kossakowski_quantum_1977,agarwal_open_1973,Alicki1976,carmichael_detailed_1976,spohn_stationary_1977}, in particular in the formulation of \cite{Temme2010}.

Let $\sigma \in \cB(\cH)$ be positive semi-definite. For each $s \in [0,1]$ let us define $\Gamma_{s} : \cB(\cH) \longrightarrow \cB(\cH)$ by
\[ \Gamma_{s}(Q) := \sigma^{1-s} Q \sigma^{s}\quad , \quad Q \in \cB(\cH). \]
Observe that $\Gamma_{s}^{\ast} := (\Gamma_{s})^{\ast} = \Gamma_{s}$, since
\[ \langle A, \Gamma_{s}(B) \rangle_{HS} = \operatorname{Tr}(A^{\dagger} \sigma^{1-s}B \sigma^{s}) = \operatorname{Tr}(\sigma^{s}A^{\dagger} \sigma^{1-s}B) = \operatorname{Tr}((\sigma^{1-s}A \sigma^{s})^\dagger B) = \langle \Gamma_{s}(A),B \rangle_{HS}\,. \] 

\begin{Defi}\cite{Temme2010}
An operator $\cL: \cB(\cH) \to \cB(\cH)$ is said to satisfy the \emph{$s$-detailed balance condition}  with respect to a full-rank state $\sigma \in \cB(\cH)$ for some $s \in [0,1]$ (or that it is ($s$,$\sigma$)-\emph{reversible}) if $\Gamma_{s} \circ \cL  =  \cL^{\ast} \circ \Gamma_{s}$.
\end{Defi}

Observe that in this case, $\cL^{\ast}(\sigma) = \cL^{\ast} \Gamma_{s}(\mathbbm{1}) = \Gamma_{s} \cL (\mathbbm{1}) = \Gamma_{s}(0)=0$. In particular, if $\cL$ is the generator of a quantum Markov semigroup, then $\sigma$ is a steady state of the semigroup. The following result is well known and can be easily proved.

\begin{Prop}\label{Prop:characterizationDetailedBalance}
Given a full-rank state $\sigma \in \cS(\mathcal{H})$ and $s \in [0,1]$, the following conditions are equivalent:
\begin{enumerate}[(i)]
\item $\mathcal{L}$ satisfies the $s$-detailed balance condition with respect to $\sigma$ and $s$;
\item $\Gamma_{s}^{-1/2} \circ \mathcal{L} \circ \Gamma_{s}^{1/2}$ is self-adjoint, i.e., $\Gamma_{s}^{1/2} \circ \mathcal{L} \circ \Gamma_{s}^{-1/2} = \Gamma_{s}^{-1/2} \circ \mathcal{L}^{\ast} \circ \Gamma_{s}^{1/2}$; 
\item $\mathcal{L}$ is self-adjoint with respect to the inner product defined by
\[ \langle A , B \rangle_{\sigma, s}:= \operatorname{Tr}(A^{\dagger} \Gamma_{s}(B)) = \operatorname{Tr}(A^{\dagger} \sigma^{1-s} B \sigma^{s})\,. \]
\end{enumerate}
\end{Prop}

\noindent It is known that if $\cL$ satisfies the detailed-balance condition w.r.t. $\sigma$ and some $s \in [0,1/2) \cup (1/2,1]$, then it satisfies the condition for every $s \in [0,1]$, see \cite{Carlen_2017}, in which case it is self-adjoint with respect to the GNS inner product $\langle \cdot, \cdot \rangle_{\sigma,s}$. From now on we will restrict to the case $s=1$, so we will omit any reference to $s$.

\begin{Rema}
Suppose we have a generator $\cL$ decomposed as in \eqref{eq:generator-derivation-dissipation}, namely $i\delta + \cD$, and assume that $\cD$ is $\sigma$-reversible. Since $\delta = \delta^{\ast}$ is self-adjoint, $\sigma$ is a steady state for $\cL$ if and only if $\delta(\sigma) = 0$, i.e., $\comm{H}{\sigma}=0$. In this case, $\delta$ is also $\sigma$-reversible. In fact
\[
\langle Q, \delta(Q) \rangle_\sigma = \operatorname{Tr}( \sigma Q^\dagger \delta (Q) ) = \operatorname{Tr}( \delta(Q \sigma )^{\dagger} Q ) = 
\operatorname{Tr}(\sigma \delta(Q)^\dagger Q) + \operatorname{Tr}(\delta(\sigma)^{\dagger} Q^*Q) = \langle\delta(Q),Q \rangle_\sigma\,.
\]

Note that this implies that $i\delta$, and consequently $\cL$, \emph{is not} $\sigma$-reversible unless $\delta = 0$. Moreover, if $\delta(\sigma) = 0$ and $\cD$ is $\sigma$-reversible, then $\cL$ is primitive if and only if $\cD$ is primitive, which can be reformulated as $\ker(\cD) = \mathbb{C} \mathbf{1}$ or $\ker(\cD^\ast) = \mathbb{C} \sigma$. Under these conditions, it is possible to estimate the speed of convergence towards the steady state for every initial state in terms of the spectral gap of its dissipative term $\cD$, as the following proposition shows.
\end{Rema}

%\textcolor{blue}{Antonio: I think we need to add the hypothesis of primitivity (which actually follows from the convergence result). I think they forgot about this condition, which do appear in their previous Theorem 9 for the discrete case. On the other hand, I waws taking a look at the definition of spectral gap of $\mathcal{L}$ that is given in (QUANTUM LOGARITHMIC SOBOLEV INEQUALITIES AND RAPID MIXING by Temme and Kastoryano), and I do not see why the Dirichlet form takes real values, especially when looking at Proposition 8 equation (25), since they are not assuming that $\mathcal{L}$ is reversible (¿?).}
\begin{Prop}{{\cite[Lemma 12]{Temme2010}}}\label{Prop:speedConvergenceGap}
    Let $\cL = i\delta + \cD$ be the generator of a quantum Markov semigroup, decomposed as in Theorem~\ref{Theo:structureQMSgenerators}.
    Suppose that there exists a full-rank state $\sigma$ such that $\delta(\sigma) = 0 $, $\cD$ is $\sigma$-reversible and $\ker \cD =  \mathbb{C} \mathbf{1}$. Then, we have that
    \begin{equation}
        \norm{\rho(t) - \sigma}_1 \le \sigma_{\min}^{-1/2} e^{-t \gap(\cD)},
    \end{equation}
    where $\rho(t)  = T_t^* (\rho)$ is the evolution of an arbitrary state $\rho$, $\sigma_{\min}$ is the smallest eigenvalue of $\sigma$, and $\gap(\cD)$ is the spectral gap of $\cD$, defined as
    \begin{equation}
    \gap(\cD) = - \min \{\lambda : \lambda \in spec(\cD), \lambda \neq 0 \}.
    \end{equation}
\end{Prop}
Note that in the proof of \cite[Lemma 12]{Temme2010} it is assumed that $\cD$ satisfies the $1/2$-detailed balance condition (with respecto to $\sigma$), but as we have discussed this is a strictly weaker condition than the $s$-detailed balance condition for $s\neq 1/2$.

In order to verify that a given generator satisfies the $\sigma$-reversibility condition, we will use the following result, based on \cite[Theorem 3.1]{Carlen_2017}, which extends upon a prior characterization by Alicki \cite{Alicki1976}. 

\begin{Prop}\label{Prop:detailBalanceConseq}
Let $\cD$ be the dissipative term of a generator of a QMS on $\cB(\cH)$ and let $\sigma \in \cS(\cH)$be a full-rank state. Then, $\cD$ satisfies the detailed balance condition with respect to $\sigma$ if and only if $\cD$ can be decomposed as
\begin{equation}\label{eq:carlen-maas-1}
    \cD(Q) = \frac{1}{2} \sum_{j \in J} \cD_{j}(Q)\,,
\end{equation}
where
\begin{equation}\label{eq:carlen-maas-2}
    \cD_{j}(Q):= e^{-\omega_{j}/2} \left(V_{j}^{\dagger}[Q, V_{j}] + [V_{j}^{\dagger}, Q]V_{j} \right) 
+ e^{\omega_{j}/2} \left(V_{j}[Q, V_{j}^{\dagger}] + [V_{j}, Q]V_{j}^{\dagger} \right),
\end{equation}
$\omega_{j} \in \mathbb{R}$ for every $j \in J$, and $\{V_{j} \}_{j \in J}$ is a collection of elements in $\cB(\cH)$ satisfying:
\begin{enumerate}[(i)]
\item $\{ V_{j}\}_{j \in J} = \{ V_{j}^{\dagger}\}_{j \in J}$.
\item Each $V_{j}$ is an eigenvector of the modular operator $\Delta_{\sigma}(Q) = \sigma Q \sigma^{-1}$ with eigenvalue $e^{-\omega_{j}}$.
\end{enumerate}
Moreover, in this case, each $\cD_j$ is negative semidefinite with respect to the GNS inner product, that is, 
\begin{equation}\label{equa:positivesemidefinite}
-\langle Q, \cD_{j}(Q) \rangle_{\sigma} \ge 0 \, \qc \forall Q \in \cB(\cH),
\end{equation}
and its kernel is given by
\begin{equation}\label{eq:kernel-commutant}
    \ker(\cD_{j}) = \{ V_{j}, V_{j}^{\dagger }\}'.
\end{equation}
\end{Prop}
Note that, in the case in which $\sigma= e^{-\beta H_\Lambda}/Z_{\beta}$ is the Gibbs state at inverse temperature $\beta >0$ associated to a Hamiltonian $H_\Lambda$, then the eigenvalues $\omega_j$ are exactly equal to the \emph{Bohr frequencies} of $H_\Lambda$, i.e., differences of eigenvalues of $H_\Lambda$. 
\begin{proof}
The equivalence between equations \eqref{eq:carlen-maas-1} and \eqref{eq:carlen-maas-2} and the fact that $\cD$ satisfies the detailed balance condition (with respect to $\sigma$) is the content of \cite[Theorem 3.1]{Carlen_2017}. To see that $-\cD_j$ is positive semidefinite with respect to the GNS scalar product, let us start writing
\begin{align*} 
\langle Q, \cD_{j}(Q) \rangle_{\sigma} 
& = e^{-\omega_{j}/2} \langle Q, V_{j}^{\dagger}[Q, V_{j}] \rangle_{\sigma} + e^{-\omega_{j}/2} \langle Q, [V_{j}^{\dagger},Q] V_{j} \rangle_{\sigma}  \\
& +e^{\omega_{j}/2} \langle Q, V_{j}[Q, V_{j}^{\dagger}] \rangle_{\sigma} +e^{\omega_{j}/2} \langle Q, [V_{j}, Q] V_{j}^{\dagger} \rangle_{\sigma}.
\end{align*}
In the previous expression, we can replace the first and third summands using that
\begin{align*}
\langle Q, V_{j}^{\dagger} [Q,V_{j}] \rangle_{\sigma} 
& = \operatorname{Tr}(\sigma Q^{\dagger} V_{j}^{\dagger} [Q,V_{j}]) = \langle V_{j}Q, [Q,V_{j}] \rangle_{\sigma}\,,\\
\langle Q, V_{j}^{\dagger} [Q,V_{j}] \rangle_{\sigma} 
& = \operatorname{Tr}(\sigma Q^{\dagger} V_{j} [Q,V_{j}^{\dagger}]) = \langle V_{j}^{\dagger}Q, [Q,V_{j}^{\dagger}] \rangle_{\sigma}\,.
\end{align*}
On the other hand, by means of the equality
$\Delta_{\sigma}(V_{j}) = \sigma V_{j}\sigma^{-1} =e^{-\omega_{j}} V_{j}$, we can replace the other two summands with
\begin{align*}
\langle Q, [V_{j}^{\dagger}, Q]V_{j} \rangle_{\sigma}
& =\operatorname{Tr}(\sigma Q^{\dagger} [V_{j}^{\dagger},Q]V_{j}) 
= e^{\omega_{j}}\operatorname{Tr}(\sigma V_{j}Q^{\dagger} [V_{j}^{\dagger},Q])  
=e^{\omega_{j}} \langle Q V_{j}^{\dagger}, [V_{j}^{\dagger}, Q] \rangle_{\sigma}\,,\\
\langle Q, [V_{j}, Q]V_{j}^{\dagger}\rangle_{\sigma} 
& =\operatorname{Tr}(\sigma Q^{\dagger}[V_{j}, Q]V_{j}^{\dagger}) 
= e^{-\omega_{j}} \operatorname{Tr}(\sigma V_{j}^{\dagger}Q^{\dagger}[V_{j}, Q])
= e^{-\omega_{j}}\langle Q V_{j}, [V_{j}, Q]\rangle_{\sigma}\,.
\end{align*}
 Thus, we have the alternative expression
\begin{align*} 
\langle Q, \cD_{j}(Q) \rangle_{\sigma} 
& = e^{-\omega_{j}/2} \langle V_{j}Q, [Q, V_{j}] \rangle_{\sigma} + e^{-\omega_{j}/2} e^{\omega_{j}}\langle Q V_{j}^{\dagger}, [V_{j}^{\dagger},Q]  \rangle_{\sigma}  \\
& \quad \quad +e^{\omega_{j}/2} \langle V_{j}^{\dagger}Q, [Q, V_{j}^{\dagger}] \rangle_{\sigma} +e^{\omega_{j}/2}  e^{-\omega_{j}} \langle QV_{j}, [V_{j}, Q]  \rangle_{\sigma}\\[2mm]
& = e^{-\omega_{j}/2} \langle V_{j}Q, [Q, V_{j}] \rangle_{\sigma} -  e^{\omega_{j}/2}\langle Q V_{j}^{\dagger}, [Q,V_{j}^{\dagger}]  \rangle_{\sigma}  \\
& \quad \quad +e^{\omega_{j}/2} \langle V_{j}^{\dagger}Q, [Q, V_{j}^{\dagger}] \rangle_{\sigma} -e^{-\omega_{j}/2}   \langle QV_{j}, [Q,V_{j}]  \rangle_{\sigma}\\[2mm]
& = -e^{-\omega_{j}/2} \langle [Q, V_{j}], [Q, V_{j}] \rangle_{\sigma} - e^{\omega_{j}/2} \langle [Q, V_{j}^{\dagger}], [Q, V_{j}^{\dagger}] \rangle_\sigma \\
& =   -e^{-\omega_{j}/2} \norm{ [Q,V_j]}_{\sigma}^2 - e^{\omega_{j}/2} \norm*{[Q,V_j^\dag]}_\sigma^2.
\end{align*}
As the last expression is $\le 0$, and vanishes only when $Q$ commutes with $V_j$ and $V_j^\dag$, this shows both \eqref{equa:positivesemidefinite} and \eqref{eq:kernel-commutant}.
\end{proof}

\subsection{Local primitivity}\label{sec:local-primitivity}

We will now consider the case of a QMS defined on a quantum spin system over $\Lambda$, i.e., when $\cH = \cH_\Lambda$, in which case it will be natural to consider generators which are \emph{locally primitive}, meaning they can be written as a sum of terms whose kernels contain only operators supported on the complement of a single site.
\begin{Defi}\label{def:local-primitivity}
Let $\cL:\cB(\cH_{\Lambda}) \to \cB(\cH_{\Lambda})$ be a Lindbladian generator of a QMS and let $\sigma \in \cB(\cH_{\Lambda})$ be a full-rank state. Assume that we can decompose $\cL$ as
\[ \cL = i \delta + \cD = i \delta + \textstyle \sum_{x \in \Lambda} \cD_{x}, \]
where $\delta(\cdot) = i\comm{H_\Lambda}{\cdot}$ is a derivation with respect to a self-adjoint operator $H_\Lambda$ which commutes with $\sigma$. We will say that this decomposition is
\begin{enumerate}[($i$)]
\item locally $\sigma$-reversible, if  $\cD_{x}$ is $\sigma$-reversible for every $x \in \Lambda$,
\item locally primitive, if $\ker(\cD_{x}) \subset \mathbbm{1}_{x} \otimes \cB(\cH_{\Lambda \setminus \{ x\}})$ for every $x \in \Lambda$.
\end{enumerate}
\end{Defi}

Note that these generators satisfy the conditions of Theorem~\ref{Prop:speedConvergenceGap}. To investigate the gap of the dissipative part, denoted as $\text{gap}(\mathcal{D})$, we can leverage the local reversibility condition of $\mathcal{D} = \sum_{x \in \Lambda} \mathcal{D}_{x}$ and treat $\mathcal{D}$ as a local Hamiltonian with respect to the scalar product $\langle \cdot,\cdot\rangle_{\sigma}$. Several tools have emerged in recent years for studying the spectral gap of local Hamiltonians \cite{Kastoryano2013,Lucia2023}. However, a potential issue arises due to the scalar product $\langle \cdot,\cdot\rangle_{\sigma}$, which is not inherently ``local''. To address this, following the ideas of \cite{Alicki2009}, we relate each local term $\cD_x$ to an operator $D_{x}$ on $\cB(\cH_\Lambda)$ which is self-adjoint with respect to the Hilbert-Schmidt scalar product:
\begin{equation}\label{eq:purification}
    D_{x}:\cB(\cH_\Lambda) \to \cB(\cH_\Lambda) \qc D_{x} = -\Gamma^{1/2} \circ \cD_{x} \circ \Gamma^{-1/2}\,,
\end{equation}
or explicitly
\begin{equation}
     D_{x}(Q) = -\cD_{x}(Q \sigma^{-1/2}) \sigma^{1/2}\, , \quad Q \in \mathcal{B}(\mathcal{H}_\Lambda) .
\end{equation}
Note that, by Proposition~\ref{Prop:detailBalanceConseq}, each $D_{x}$ is positive semi-definite (i.e. $D_{x} \geq 0$) for every $x \in \Lambda$. More generally, for every $X \subset \Lambda$ define $D_{X}:=\sum_{x \in X}D_{x}$ and $\cD_{X}=\sum_{x \in X}\cD_{x}$. Then, $D_{X}$ and $-\cD_{X}$ have the same (nonnegative) eigenvalues, and their eigenvectors can be identified via the map $Q \mapsto Q \sigma^{1/2}$.

In the following result, we will show that in the case of a local QMS generator which also satisfies the local primitivity condition, it is always possible to estimate its spectral gap in terms of the spectral gap of the \emph{canonical purified Hamiltonian} $\mathbf{H}$ defined in Section~\ref{sec:canonicalPurifiedHamiltonian}.

\begin{Theo}\label{thm:local-primitivity-gap}
Let $\sigma \in \cB(\cH_{\Lambda})$ be a full-rank state and let $\mathbf{H}_{\Lambda}$ be its canonical purified Hamiltonian. Then, for every $\cL$ satisfying the conditions in Definition~\ref{def:local-primitivity}, it holds that
\[ \gap(\cD) \geq \min_{x \in \Lambda}\{ \gap(\cD_{x})\} \gap\left( \mathbf{H}_{\Lambda}\right). \]
\end{Theo}

\begin{proof}
By construction, $\gap(\cD)$ is the same as the gap of the Hamiltonian $D_{\Lambda} = \sum_{x\in \Lambda} D_x$, where $D_x$ is defined in \eqref{eq:purification}. 
 In particular, it holds that $\gap(\cD_{X}) =  \gap(D_{X})$, and the ground state space (i.e. kernel) of $H_{X}$ satisfies, 
\[ \textstyle \ker{(D_{X})} = \bigcap_{x \in X} \ker{(D_{x})} = \bigcap_{x \in X} \ker{(\cD_{x})} \sigma^{1/2} \subset W_{X} \qc \ker(H_{\Lambda}) = W_{\Lambda} =  \mathbb{C} \sigma^{1/2}\,, \]
where we have used the property of local primitivity. Let $P_{X}$ be the orthogonal projection onto $\ker{( \sum_{x \in X} D_{x})}$, and let us write $P_{x}:= P_{\{ x\}}$ for each $x \in \Lambda$ to simplify. Then, the \emph{local} gap of each $D_{x}$ satisfies
\[\textstyle  D_{x} \geq \gap(D_{x}) P_{x}^{\perp} = \gap(\cD_{x}) P_{x}^{\perp} \geq \gap(\cD_{x}) \Pi_{x}^{\perp}\,. \]
Summing over all $x \in \Lambda$ we get that
\[ D_{\Lambda}=\sum_{x \in \Lambda} D_{x} \geq \min_{x \in \Lambda}\{ \gap(\cD_{x}) \} \sum_{x \in \Lambda} \Pi_{x}^{\perp} =  \min_{x \in \Lambda}\{ \gap(\cD_{x}) \} \, \mathbf{H}_{\Lambda}\,. \]
Observe now that $D_{\Lambda}$ and $\mathbf{H}_{\Lambda}$, being frustration-free Hamiltonians with lowest eigenvalue equal to zero, have the same ground space $W_{\Lambda} = \mathbb{C} \sigma^{1/2}$. Using this property on $\mathbf{H}_\Lambda$, we can now add the following to the previous chain of inequalities
\[ D_{\Lambda}=\sum_{x \in \Lambda} D_{x} \geq \min_{x \in \Lambda}\{ \gap(\cD_{x}) \} \gap\left(\mathbf{H}_\Lambda \right)  \Pi_{\Lambda}^{\perp} \,. \]
Finally, using the definition of gap and noting that $W_{\Lambda}$ is also the ground space of $H_{\Lambda}$, we conclude that
\[ \gap(\cD) = \gap(D_\Lambda) \geq \min_{x \in \Lambda}\{ \gap(\cD_{x}) \} \gap\left(\mathbf{H}_\Lambda \right) \,. \] 

\end{proof}

The preceding theorem enables us to divide the spectral gap problem of $\cD$ into two distinct steps. Firstly, we examine the gap of the local dissipative terms $\mathcal{D}_{x}$ and establish a uniform lower bound for this gap. Subsequently, we apply the results from Section~\ref{sec:canonicalPurifiedHamiltonian} to derive a lower bound for the spectral gap of the canonical purified Hamiltonian $\mathbf{H}_{\Lambda}$, by verifying the corresponding mixing condition on $\sigma$. Remarkably, the latter gap relies solely on the structure of $\sigma$ but not on $\cD$, so it is a static property of the system, and it is independent of the specific choice of generator $\cD$ (as long as it satisfied local primitivity and local reversibility).

\subsection{Davies semigroup}

A particularly important class of generators of QMS is the one considered by Davies~\cite{Davies1974}. They describe, under certain assumptions, the weak-coupling limit of the evolution given by the coupling of a quantum spin system with Hamiltonian $H_\Lambda$ to a thermal bath at inverse temperature $\beta$.

In order to construct a locally primitive generator, we consider as coupling operators for each $x \in \Lambda$ a family $(S_{x,\alpha})_{\alpha}$ of elements in $\cB(\cH_{\Lambda})$.
Then the Davies generator is given by
\begin{equation}
    \cL(Q) = i [H_{\Lambda},Q] + \sum_{x \in \Lambda} \sum_{\alpha, \omega} \hat g_{x,\alpha}(\omega)\qty( S_{x, \alpha}(\omega)^{\dagger} Q S_{x, \alpha}(\omega) - \frac{1}{2}\{ S_{x, \alpha}(\omega)^{\dagger}S_{x, \alpha}(\omega), Q\}_{+})\, 
\end{equation}
where the variable $\omega$ runs over the Bohr frequencies of $H_\Lambda$, $\hat g_{x,\alpha}(\omega)$ are positive constants satisfying $\hat g_{x,\alpha}(-\omega) = e^{-\beta \omega}\hat g_{x,\alpha}(\omega)$, and the operators $S_{x,\alpha}(\omega)$ are related to the couplings $S_{x,\alpha}$ by
\[
e^{it H_\Lambda} S_{x,\alpha} e^{-it H_\Lambda} = \sum_{\omega} S_{x,\alpha}(\omega) e^{-i\omega t} \qc \forall t \in \mathbb{R},
\]
implying that $ S_{x,\alpha}(\omega)^\dag = S_{x,\alpha}(-\omega)$.
We can see that the dissipative terms of $\cL$ can be put into the form of \eqref{eq:carlen-maas-1} by a simple re-grouping:
\begin{equation}
    \cL(Q) = i [H_{\Lambda},Q] + \sum_{x \in \Lambda} \sum_{\alpha, \omega} 
    \qty(e^{-\omega/2}V_{x,\alpha,\omega}^{\dag}\comm{Q}{V_{x,\alpha,\omega}} + e^{\omega/2} \comm{V_{x,\alpha,\omega}}{Q}V_{x,\alpha,\omega}^\dag),
\end{equation}
where $V_{x,\alpha,\omega} = g_{x,\alpha}^{1/2}(\omega)S_{x, \alpha}(\omega)$. From Proposition~\ref{Prop:detailBalanceConseq} it follows that
\begin{multline}\label{eq:davies-1}
\cD_{x,\alpha,\omega}(Q) := e^{-\omega/2} \left(V_{x,\alpha,\omega}^{\dagger}[Q, V_{x,\alpha,\omega}] + [V_{x,\alpha,\omega}^{\dagger}, Q]V_{x,\alpha,\omega} \right)\\ 
+ e^{\omega/2} \left(V_{x,\alpha,\omega}[Q, V_{x,\alpha,\omega}^{\dagger}] + [V_{x,\alpha,\omega}, Q]V_{x,\alpha,\omega}^{\dagger} \right)
\end{multline}
is semidefinite negative with respect to the GNS scalar product, and its kernel is given by $\{V_{x,\alpha,\omega},V_{x,\alpha,\omega}^\dag\}'=\{S_{x, \alpha}(\omega), S_{x, \alpha}(\omega)^\dag\}'$.

By further grouping the terms $\cD_{x,\alpha,\omega}$, we obtain a decomposition of the type of Definition~\ref{def:local-primitivity}: defining
\begin{equation}\label{eq:davies-2}
    \cD_{x} := \sum_{\alpha,\omega} \cD_{x,\alpha,\omega} \,,
\end{equation}
it follows that $\cD_x$ is $\sigma$-reversible for every $x \in \Lambda$, where $\sigma = \frac{1}{Z_\beta}\exp(-\beta H_\Lambda)$ is the thermal equilibrium state of the system. Moreover, we have that~\cite[Proposition 5.5]{Lucia2023}
\[
\ker(\cD_{x}) = \{ S_{x, \alpha}(\omega) \colon \forall \alpha,\omega\}'.
\]
Therefore, local primitivity holds as long as the right-hand side of the last equation is contained in $\cB(\cH_{\Lambda\setminus\{x\}})$. A sufficient condition for this to happen is the following:
\begin{Assumption}\label{assumption:davies-1}
The coupling operators $\{ S_{x,\alpha}\}$ satisfy
\begin{equation}
    \{ S_{x, \alpha} \colon \forall \alpha \}' = \cB(\cH_{\Lambda\setminus\{x\}}).
\end{equation}
\end{Assumption}

We can summarize this discussion in the following Proposition:
\begin{Prop}\label{prop:davies-local-primitivity}
Under Assumption~\ref{assumption:davies-1}, $\cL$ satisfies the conditions of Definition~\ref{def:local-primitivity}, and as a consequence of Theorem~\ref{thm:local-primitivity-gap}
\[
\gap(\cD) \ge \min_{x\in \Lambda}{\gap(\cD_x)} \gap(\mathbf{H}),
\]
where $\cD_x$ is defined in equations \eqref{eq:davies-1}and \eqref{eq:davies-2}. 
\end{Prop}

In order to estimate $\gap(\cD)$, we now need to compute a lower bound to $\gap(\cD_x)$. In \cite{Lucia2023}, this was done under the following assumptions:
\begin{Assumption}\label{assumption:davies-2}
The system Hamiltonian $H_\Lambda = \sum_{X\subset \Lambda} \Phi_X$ is commuting and has finite range $r>0$.
\end{Assumption}

\begin{Assumption}\label{assumption:davies-3}
The coupling operators $\{S_{x,\alpha}\}_\alpha$ are \emph{local}, namely they belong to $\cB(\cH_x)$ for every $x \in \Lambda$.
\end{Assumption}

Under Assumptions~\ref{assumption:davies-1}-\ref{assumption:davies-3}, the generators $\cD_x$ are also finite range and supported on neighborhoods $R_x \subset \Lambda$ of $x$, where $R_{x}$ is contained in a ball of radius $4r$ centered at $x$. This allows to obtain a simple estimate on the gap of $\cD_x$.
\begin{Prop}[{\cite[Proposition 5.11]{Lucia2023}}]\label{prop:local-davies-gap-commuting}
Under Assumptions~\ref{assumption:davies-1},~\ref{assumption:davies-2} and~\ref{assumption:davies-3}, there exist positive constants $C_1$ and $C_2$ such that
    \begin{equation}
        \gap(\cD_x) \ge \frac{C_2}{\abs{\Omega_x}} \hat{g}_{\min}(x) e^{-C_1\beta \abs{R_x}},
    \end{equation}
    where $\Omega_x$ is the set of Bohr frequencies of the Hamiltonian $H_{\{x\}}^{\partial} = H_\Lambda - H_{\Lambda\setminus \{x\}}$, and $\hat{g}_{\min}(x) = \min_\alpha \min_{\omega \in \Omega_x} \hat{g}_{x,\alpha}(\omega)$.
\end{Prop}
Note that, in the case of \emph{translation invariance}, the estimate on $\gap(\cD_x)$ is independent on $x$, and scales with respect to beta as $\order{\hat{g}_{\min} e^{-c\beta}}$. For certain choices of thermal bath, $\hat{g}_{\min}$ is also independent of $\beta$, and so the bound reduces to an exponential in $\beta$.

\section{Spectral gap estimates for the canonical purified Hamiltonian}
\label{sec:Estimating_the_gap}

The aim of this section is to obtain estimates on $\gap(\mathbf{H}_{\Lambda})$ where $\mathbf{H}_{\Lambda}$ is the canonica purified Hamiltonian of a full-rank state $\sigma \in \cB(\cH_\Lambda)$.
As shown in \cite{KL_2018}, the gap of every local and frustration-free Hamiltonian $H_{\Lambda}$ is closely related to a property of its orthogonal projections $P_{X}$ onto the ground spaces of $H_{X}$, where $X \subset \Lambda$. This property, known as the \emph{martingale condition}, describes the decay of $\| P_{A} P_{B} - P_{A \cup B} \|$ in relation to the size of $A \cap B$. In the case of the canonical purified Hamiltonian, the availability of explicit and tractable descriptions for each ground state subspace $W_{X}$ and its orthogonal projection $\Pi_{X}$ allows for an explicit computation of these differences in terms of a  correlation measure of the given state $\sigma$ defining the Hamiltonian.

There is, however, a subtle difference between the definition of $\mathbf{H}_{X}$ in our context, and the Hamiltonian $H_{X}$ used in \cite{KL_2018}. In that work,  $H_{X}$ is defined as a sum $H_{X} = \sum_{Z \subset X} \Phi_{Z}$ for some local interaction $\Phi$ where each $\Phi_{Z}$ is positive semidefinite and supported in $Z$. In contrast, we have defined $\mathbf{H}_{X} = \sum_{x \in X} \Pi_{x}^{\perp}$, where each $\Pi_{x}^{\perp}$ is positive semidefinite but may not be supported in a proper subset of $\Lambda$ for a general $\sigma$. Nevertheless, this is not impediment to prove the following lemma since its argument only relies on the fact that the summands $\Pi_{x}^{\perp}$ are positive semidefinite, and thus $\mathbf{H}_{L \cup R} \leq \mathbf{H}_{L} + \mathbf{H}_{R}$ for every pair of subsets $L,R \subset \Lambda$.

\begin{Lemm}[{\cite{KL_2018,Lucia2023}}]\label{Lemm:DivideAndConquer}
Let $X \subset \Lambda$. Assume there exist $0< \delta < 1$ and $L_{i}, R_{i} \subset \Lambda$ for $i=1, \ldots, s$ satisfying the following properties:
\begin{enumerate}[(i)]
\item $X=L_{i} \cup R_{i}$ for $i=1, \ldots, s$,
\item $(L_{i} \cap R_{i}) \cap (L_{j} \cap R_{j})=\emptyset$ whenever $i \neq j$,
\item $\| \Pi_{L_{i}} \Pi_{R_{i}} - \Pi_{X} \| \leq \delta$.
\end{enumerate}
Then,
\begin{equation} \gap(\mathbf{H}_{X}) \geq \frac{1-\delta}{1+ \frac{1}{s}} \min\{ \gap(\mathbf{H}_{Y}) \colon Y=L_{1}, \ldots, L_{s}, R_{1}, \ldots, R_{s} \}\,.
\end{equation}
\end{Lemm}

This result allows us to recursively reduce the spectral gap problem to sets of smaller size. It is then important to have a method for estimating $\|\Pi_{L} \Pi_{R} - \Pi_{L \cup R}\|$, as well as explicit estimates of the gap on ``small'' sets. We deal with the first issue in  Section~\ref{subsec:MixingCondition}, where we relate the martingale property with a spatial mixing condition  for $\sigma$; and we deal with the second issue in Section~\ref{sec:rough_estimate}, where we establish a lower bound of the gap of $\mathbf{H}_{X}$, which we only expect to be useful when $X$ is ``small'' (i.e., at the end of the recursion for the gap estimate of Lemma~\ref{Lemm:DivideAndConquer}). These results are later used to derive explicit lower bounds on $\gap(\mathbf{H}_{\Lambda})$ for a 1D ring (Section~\ref{sec:gap-1D}), a 2D torus (Section~\ref{sec:gap-2D}) and higher-dimensional systems (Section~\ref{sec:gap-higher-dim}).

\subsection{Mixing condition}
\label{subsec:MixingCondition}

The martingale condition for the canonical purified Hamiltonian $\mathbf{H}$ is going to be reformulated in terms of a measure of correlations of the original state $\sigma$, defined in the following definition.

\begin{Defi}\label{def:spatial-mixing}
    Let $\sigma \in \cB(\cH_\Lambda)$ be a full-rank state, and let  $A,C,D$ be three disjoint subsets of $\Lambda$. We define
    \begin{equation}\label{eq:delta-definition}
        \Delta_\sigma(A:C|D) := 
\sup_{R_{AD}, Q_{CD}} \left|\Tr_{ACD}[(\sigma_{ACD} - \sigma_{AD} \sigma_{D}^{-1} \sigma_{DC}) Q_{CD}^{\dagger} R_{AD}]\right| \,,
\end{equation}
where the supremum is taken over $R_{AD} \in \cB(\cH_{AD})$ and $Q_{CD} \in \cB(\cH_{CD})$ with 
$\norm{R_{AD}}_\sigma = \norm{Q_{CD}}_\sigma  = 1$.
\end{Defi}

The following result shows several upper and lower bounds for \eqref{eq:delta-definition}, that will be helpful to study its decay.
 
\begin{Prop}
\label{prop:estimatesMartinagleCondition}
Let $\Lambda =ABCD$ be a partition into four disjoint subsets. Then, we have
\begin{equation}\label{equa:estimatesMartinagleCondition1}
\Delta_\sigma(A:C|D) \leq \frac{1}{2}\left(\left\|\mathbbm{1} - \sigma_{AD} \sigma_{D}^{-1} \sigma_{DC} \sigma_{ACD}^{-1} \right\|_{\infty} + \left\|\mathbbm{1} - \sigma_{ACD}^{-1}\sigma_{AD} \sigma_{D}^{-1} \sigma_{DC}\right\|_{\infty} \right)
\end{equation}
Moreover, if $D=\emptyset$, the following bound holds
\begin{align}
\label{equa:estimatesMartinagleCondition2} \Delta_\sigma(A:C|\emptyset) & \leq \| \mathbbm{1} - \sigma_{A} \sigma_{C} \sigma^{-1}_{AC}\|_{\infty}\,,\\[2mm]
\label{equa:estimatesMartinagleCondition3} \Delta_\sigma(A:C|\emptyset) & \geq  \sup_{\| Q_{A}\|_{\infty}, \| Q_{C}\|_{\infty} \leq 1} \left|\Tr(\sigma Q_{C}Q_{A}) - \Tr(\sigma Q_{C}) \Tr(\sigma 
Q_{A}) \right| \,,
\end{align}
\end{Prop}

\begin{proof}
Let us first prove \eqref{equa:estimatesMartinagleCondition1}. We start from  the definition of $\Delta_\sigma$ in \eqref{eq:delta-definition}. Using that 
\[ \Tr_{ADC}(\sigma_{ADC}Q_{CD}^{\dagger}Q_{DC}) = \Tr_{ADC}(\sigma_{ADC}R_{AD}^{\dagger}R_{AD}) =1,\]
we can estimate
\begin{align*}
|\Tr_{ACD}((\sigma_{ACD} & - \sigma_{AD} \sigma_{D}^{-1} \sigma_{DC})  Q_{CD}^{\dagger} R_{AD})|=\\[1mm]
& = |\Tr_{ACD}((\mathbbm{1} - \sigma_{ACD}^{-1/2}(\sigma_{AD} \sigma_{D}^{-1} \sigma_{DC}) \sigma_{ACD}^{-1/2}) (Q_{CD} \sigma_{ACD}^{1/2})^{\dagger} R_{AD} \sigma_{ACD}^{1/2})|\\[1mm]
& \leq \| \mathbbm{1} - \sigma_{ACD}^{-1/2}(\sigma_{AD} \sigma_{D}^{-1} \sigma_{DC}) \sigma_{ACD}^{-1/2} \|_{\infty} \,.
\end{align*}
We then apply \cite[Theorem IX.4.5]{Bhatia97}, which for $S := \sigma_{ADC}^{-1/2}$ and $X := \sigma_{AD} \sigma_{D}^{-1} \sigma_{DC} - \sigma_{ADC}$ yields that
\[ \| S^{\dagger} X S\|_{\infty} \leq \frac{1}{2}\| SS^{\dagger}X + X S S^{\dagger}\|_{\infty} \leq \frac{1}{2}(\|SS^{\dagger}X\|_{\infty} + \| XSS^{\dagger}\|_{\infty})\,. \]
This leads to the desired inequality \eqref{equa:estimatesMartinagleCondition1}.

To check the second part, let us now assume that $D=\emptyset$. In this case, the formula we just proved reads
\begin{align*} 
\Delta_\sigma(A:C|\emptyset) & \leq \frac{1}{2}\left(\left\|\mathbbm{1} - \sigma_{A}  \sigma_{C} \sigma_{AC}^{-1} \right\|_{\infty} + \left\|\mathbbm{1} - \sigma_{AC}^{-1}\sigma_{A} \sigma_{C}\right\|_{\infty} \right) = \left\|\mathbbm{1} - \sigma_{AC}^{-1}\sigma_{A} \sigma_{C}\right\|_{\infty}\,, 
\end{align*}
as the norm is invariant under the adjoint operation. This provides \eqref{equa:estimatesMartinagleCondition2}. Finally, since every $Q$ with $\| Q\|_{\infty} = 1$ satisfies $\operatorname{Tr}(Q_{C}^{\dagger}Q_{C} \sigma) \leq 1$, we can estimate from below
\[ \Delta_\sigma(A:C|\emptyset) \geq \sup_{\| Q_{A}\|_{\infty}, \| Q_{C}\|_{\infty} \leq 1} |\Tr(\sigma Q_{C}Q_{A}) - \Tr(\sigma Q_{C}) \Tr(\sigma 
Q_{A}) | \,,  \]
which is the usual operator correlation function, as it appears in \eqref{equa:estimatesMartinagleCondition3} . 
\end{proof}

We are now ready to prove the main theorem of this section, showing that $\Delta_\sigma(A:C|D)$ precisely encodes the martingale condition for the canonical purified Hamiltonian $\mathbf{H}$.

\begin{Theo}\label{Theo:martingaleVScorrelations}
Let $\Lambda =ABCD$ be a partition into four disjoint subsets. Then, we have that
\begin{equation}\label{equa:martingaleVScorrelationsMainFormula}
\| \Pi_{AB}\Pi_{BC} - \Pi_{ABC}  \| = \Delta_\sigma(A:C|D).
\end{equation}
\end{Theo}

\begin{proof}
Let us start by noticing that
\begin{equation}\label{equa:martingaleVScorrelationsAux1}
\begin{split}
\| \Pi_{AB}\Pi_{BC} - \Pi_{ABC}  \| & = \| (\Pi_{AB} - \Pi_{ABC})(\Pi_{BC} - \Pi_{ABC}) \|\\[1mm]
& = \sup\{|\langle \tilde{Q} - \Pi_{ABC}(\tilde{Q}), \tilde{R} - \Pi_{ABC}(\tilde{R})\rangle_{HS}| \} \,,
\end{split}
\end{equation}
where the supremum is taken over $\tilde{Q} \in W_{AB}$ and $\tilde{R} \in W_{BC}$ such that $\| \tilde{Q}\|_{2} = \| \tilde{R}\|_{2} = 1$.
Since $\Pi_{ABC}$ is a self-adjoint projection, we can expand the above scalar product as 
\begin{equation}\label{equa:equa:martingaleVScorrelationsAux3}
\langle \tilde{Q} - \Pi_{ABC}(\tilde{Q}), \tilde{R} - \Pi_{ABC}(\tilde{R})\rangle_{HS} =
 \langle \tilde{Q} , \tilde{R}\rangle_{HS} - \langle \tilde{Q},  \Pi_{ABC}(\tilde{R})\rangle_{HS}\,.
\end{equation}
We can next represent every $\tilde{Q} \in W_{AB}$  as $\tilde{Q} = Q_{CD} \sigma^{1/2}$ for some $Q_{CD} = \mathbbm{1}_{AB} \otimes Q_{CD} \in \cB(\cH_{CD})$, and analogously every $\widetilde{R} \in W_{BC}$ as $\tilde{R} = R_{AD} \sigma^{1/2}$ for some $R_{AD} = \mathbbm{1}_{BC} \otimes R_{AD} \in \cB(\cH_{AD})$. In particular, 
\[
  \Pi_{ABC}(\tilde{R}) = \Tr_{ABC}(R_{AD} \sigma) \sigma_{D}^{-1} \sigma^{1/2} = \Tr_{A}(R_{AD} \sigma_{AD})\sigma_{D}^{-1} \sigma^{1/2}\,.
\]
We also have that $\norm{Q_{CD}}_\sigma = \norm*{\tilde{Q}}_2 = 1$ and $\norm{R_{AD}}_\sigma = \norm*{\tilde{R}}_2 = 1$.
With these new expressions, we can rewrite the scalar products appearing in \eqref{equa:equa:martingaleVScorrelationsAux3} as
\begin{align*}
\langle \widetilde{Q} , \widetilde{R} \rangle_{HS} & = \operatorname{Tr}(\sigma Q_{CD}^{\dagger} R_{AD}) = \operatorname{Tr}_{ADC}(\sigma_{ADC} Q_{CD}^{\dagger} R_{AD})
\end{align*}
and
\begin{align*}
\langle \widetilde{Q} , \Pi_{ABC}(\widetilde{R}) \rangle_{HS} & = \operatorname{Tr}(\sigma^{1/2} Q_{CD}^{\dagger} \Tr_{A}(R_{AD} \sigma_{AD}) \sigma_{D}^{-1} \sigma^{1/2})\\[2mm]
& =\operatorname{Tr}(\sigma_{D}^{-1} \sigma   Q_{CD}^{\dagger} \Tr_{A}(R_{AD} \sigma_{AD})  )\\[2mm]
& = \operatorname{Tr}_{CD}( \sigma_{D}^{-1} \sigma_{CD} Q_{CD}^{\dagger} \Tr_{A}(R_{AD} \sigma_{AD}) )\\[2mm]
& = \operatorname{Tr}_{ADC}(\sigma_{AD}\sigma_{D}^{-1} \sigma_{CD}Q_{CD}^{\dagger}R_{AD}).
\end{align*}
Finally, inserting this formulas in \eqref{equa:equa:martingaleVScorrelationsAux3}, we conclude the result.
\end{proof}

\begin{Coro}\label{coro:martingaleVScorrelations}
    Let $\Lambda =ABCD$ be a partition into four disjoint subsets. Then, the quantity $\Delta_\sigma(A:C|D)$ can be written as
    \begin{equation}
        \Delta_\sigma(A:C|D) = \sup_{Q_{CD}, R_{AD}} \abs{ \langle Q_{CD} , R_{AD} \rangle_\sigma},
    \end{equation}
    where the supremum is taken over $R_{AD} \in \cB(\cH_{AD})$ and $Q_{CD} \in \cB(\cH_{CD})$ satisfying $\norm*{R_{AD}}_\sigma  = \norm*{Q_{CD}}_\sigma = 1$
    and
    \begin{equation}
    \Tr_C (Q_{CD} \sigma_{CD}) = \Tr_A(R_{AD}\sigma_{AD}) = 0.
    \end{equation}
\end{Coro}
\begin{proof}
In the proof of Theorem~\ref{Theo:martingaleVScorrelations}, we have seen that we can write
\begin{equation}\label{equa:commutingMartingaleUpperBound1}
\Delta_\sigma(A:C|D)
= \sup_{\widetilde{Q}, \widetilde{R}} \abs{ \langle \widetilde{Q}, \widetilde{R} \rangle_{HS}}
\end{equation}
where the supremum is taken over all $\widetilde{Q} \in W_{AB} \cap W_{ABC}^{\perp}$ and $\widetilde{R} \in W_{BC} \cap W_{ABC}^{\perp}$ with $\| \widetilde{Q}\|_{2} = \|\widetilde{R}\|_{2} = 1$. 
The condition $\widetilde{Q} \in W_{AB}$  means that $\widetilde{Q} = Q_{CD}\sigma^{1/2}$ for some $Q_{CD} \in \cB(\cH_{CD})$. With this notation, $\| \widetilde{Q}\|_{2} = 1$ is equivalent to $\norm*{Q_{CD}}_\sigma=1$ whereas $\widetilde{Q} \in W_{ABC}^{\perp}$ is equivalent to
\[ 
0 = \Pi_{ABC}(\widetilde{Q}) = \Tr_{ABC}( Q_{CD}\sigma) \sigma_{D}^{-1}\sigma^{1/2} = \Tr_{C}( Q_{CD}\sigma_{CD})\sigma_{D}^{-1}\sigma^{1/2}
\]
or $\Tr_{C}(Q_{CD} \sigma_{DC})=0$. 
Analogously,  $\widetilde{R} \in W_{BC} \cap W_{ABC}^{\perp}$ and $\|\widetilde{R}\|_{2} = 1$ is equivalent to $\widetilde{R} = R_{AD}\sigma^{1/2}$ for some $R_{AD} \in \cB(\cH_{AD})$, $\Tr_{A}(R_{AD}\sigma_{AD}) = 0$ and $\norm*{\sigma_{AD}}_\sigma = 1$. To conclude, we observe that $\langle \widetilde{Q}, \widetilde{R} \rangle_{HS} = \langle Q_{CD} , R_{AD} \rangle_\sigma$.
\end{proof}

In the particular case in which the marginals of the state $\sigma$ over different regions commute with each other, we can improve the bound of Proposition~\ref{prop:estimatesMartinagleCondition} as follows.
\begin{Prop}\label{prop:commutingMartingaleUpperBound}
Let $\Lambda =ABCD$ be a partition into four disjoint subsets.  If  the marginals $\sigma_{AD}, \sigma_{D}, \sigma_{DC}$ commute with each other, then we can upper estimate
\[ \Delta_\sigma(A:C|D) \leq  1-  \|\sigma_{AD} \sigma_{D}^{-1} \sigma_{DC} \sigma_{ADC}^{-1}\|_{\infty}^{-1} \,. \]
\end{Prop}

\begin{Rema}
Using that $\|\sigma_{AD} \sigma_{D}^{-1} \sigma_{DC} \sigma_{ADC}^{-1}\|_{\infty} \leq 1+ \|\sigma_{AD} \sigma_{D}^{-1} \sigma_{DC} \sigma_{ADC}^{-1} - \mathbbm{1}\|_{\infty}$ we get an improved version of \eqref{equa:estimatesMartinagleCondition1}, namely
\[ \Delta_\sigma(A:C|D)  \leq \frac{ \| \mathbbm{1} -\sigma_{AD} \sigma_{D}^{-1} \sigma_{DC} \sigma_{ADC}^{-1}\|_{\infty}}{1+\|  \mathbbm{1} -\sigma_{AD} \sigma_{D}^{-1} \sigma_{DC} \sigma_{ADC}^{-1}\|_{\infty}}\,. \]
\end{Rema}

\begin{proof}[Proof of Proposition~\ref{prop:commutingMartingaleUpperBound}] Recall that, by Corollary~\ref{coro:martingaleVScorrelations},
\[ 
\Delta_\sigma(A:C|D)
=  \sup_{Q_{CD}, R_{AD}} \abs{\langle Q_{CD}, R_{AD} \rangle_{\sigma}} =
\sup_{Q_{CD}, R_{AD}} \Re\langle Q_{CD}, R_{AD} \rangle_{\sigma},
\]
where the supremum is taken over $R_{AD} \in \cB(\cH_{AD})$ and $Q_{CD} \in \cB(\cH_{CD})$ satisfying $\norm*{R_{AD}}_\sigma  = \norm*{Q_{CD}}_\sigma = 1$
and $\Tr_C (Q_{CD} \sigma_{CD}) = \Tr_A(R_{AD}\sigma_{AD}) = 0$.

Using the polarization identity and the normalization of $Q_{CD}$ and $R_{AD}$, we can write
\[ \Re\langle Q_{CD}, R_{AD} \rangle_{\sigma} = 1-\frac{1}{2} \norm*{Q_{CD}-R_{AD}}_\sigma^2  = 1 - \frac{1}{2} \norm*{S_{ADC}}^2_\sigma ,\]
where $S_{ADC}:=Q_{CD}-R_{AD}$, so that the supremum turns into an infimum
\begin{equation}\label{equa:commutingMartingaleUpperBound2}
\Delta_\sigma(A:C|D)
= 1-\frac{1}{2}\textstyle \inf \norm*{S_{ADC}}^2_\sigma.
\end{equation}

Let us denote $\rho_{ADC} = \sigma_{DC} \sigma_{D}^{-1} \sigma_{AD}$, which a positive and invertible operator by the commutativity hypothesis. Then, we can write
\begin{align*}
\norm*{S_{ADC}}^2_\sigma  = \Tr_{}(S_{ADC} \, \sigma \, S_{ADC}^{\dagger}) & = \Tr_{ADC}(S_{ADC} \, \sigma_{ADC} \,S_{ADC}^{\dagger})\\ & = \Tr_{ADC}(S_{ADC}\rho_{ADC}^{1/2} \rho_{ADC}^{-1/2} \sigma_{ADC}  \rho_{ADC}^{-1/2}\rho_{ADC}^{1/2}  S_{ADC}^{\dagger}) )\,. 
\end{align*}
Using the inequalities $\mathbbm{1} \| O^{-1}\|_{\infty}^{-1} \leq O \leq \| O\|_{\infty} \mathbbm{1}$ valid for every positive and invertible observable $O$, we can lower estimate the previous expression to get
\begin{equation} \label{equa:commutingMartingaleUpperBound3}
\norm*{S_{ADC}}^2_\sigma \geq \Tr_{ADC}(S_{ADC} \rho_{ADC} S_{ADC}^{\dagger}) \|\rho_{ADC}^{1/2} \sigma_{ADC}^{-1} \rho_{ADC}^{1/2} \|_{\infty}^{-1}  
\,.
\end{equation}
Notice that the first factor can be simplified
\begin{align*}
\Tr_{ADC}(S_{ADC} \rho_{ADC} S_{ADC}^{\dagger}) 
& = \Tr_{ADC}(S_{ADC}^{\dagger}S_{ADC} \rho_{ADC})\\  
& = \Tr_{ADC}(Q_{CD}^{\dagger} Q_{CD} \sigma_{AD} \sigma_{D}^{-1} \sigma_{DC}) -2 \Re \Tr_{ADC}(Q_{CD}^{\dagger}R_{AD}\sigma_{AD} \sigma_{D}^{-1}\sigma_{DC})\\
& \quad \quad +\Tr_{ADC}(R_{AD}^{\dagger} R_{AD} \sigma_{AD} \sigma_{D}^{-1} \sigma_{DC})\\
&=  \Tr_{CD}(Q_{CD}^{\dagger}Q_{CD} \sigma_{D}\sigma_{D}^{-1} \sigma_{CD}) - 2 \Re \Tr_{DC}(Q_{CD}^{\dagger} \Tr_{A}(R_{AD}\sigma_{AD}) \sigma_{D}^{-1}\sigma_{DC}\\
& \quad \quad + \Tr_{AD}(R_{AD}^{\dagger}R_{AD} \sigma_{AD} \sigma_{D}^{-1} \sigma_{D})\\
& = \Tr_{CD}(Q_{CD}^{\dagger}Q_{CD} \sigma_{D}) - 0 + \Tr_{AD}(R_{AD}^{\dagger}R_{AD} \sigma_{AD} ) =2\,.
\end{align*}
Replacing this value in \eqref{equa:commutingMartingaleUpperBound3} and applying this estimate to \eqref{equa:commutingMartingaleUpperBound2} we get that
\[  \Delta_\sigma(A:C|D) \leq 1-\| \rho_{ADC}^{1/2} \sigma_{ADC}^{-1} \rho_{ADC}^{1/2}\|_{\infty}^{-1}\,. \]
Finally, using that $\| \rho_{ADC}^{1/2} \sigma_{ADC}^{-1} \rho_{ADC}^{1/2}\|_{\infty} \leq \| \rho_{ADC} \sigma_{ADC}^{-1}\|_{\infty}$ by \cite[Proposition IX.1.1]{Bhatia97}, we conclude the result.
\end{proof}

\subsection{Small regions estimates}
\label{sec:rough_estimate}

The aim of this section will be to prove the following estimate on the spectral gap of $\mathbf{H}$, which will be useful to bound the gap of the canonical Hamiltonian on small subsets of $\Lambda$.

\begin{Theo}\label{Theo:roughBoundSpectralGap}
Let $\sigma \in \cB(\cH_\Lambda)$ be a full-rank state, and let $\mathbf{H}$ be the associated canonical purified Hamiltonian. For each $X \subset \Lambda$, let us denote 
\begin{equation}\label{equa:roughBoundSpectralGap}
\eta_{X}(\sigma) := \inf\{ \|Q \sigma^{-1/2} \|_{\infty} \cdot \|\sigma^{1/2} Q^{-1} \|_{\infty} \colon Q \in \mathcal{B}_{X^{c}} \text{ invertible} \}\,.    
\end{equation}
Then, we have that
\begin{equation}
\operatorname{gap}(\mathbf{H}_{X}) \geq \frac{1}{\eta_{X}(\sigma)^{4}}\,.
\end{equation}
\end{Theo}
\begin{Rema}\label{Rema:etaProperties}
Before proving the theorem, let us comment on some useful properties of $\eta_X$.
\begin{itemize}
\item[$(i)$] Since $\| \cdot \|_{\infty}$ is sub-multiplicative, we have the trivial bounds 
\[ 1 \leq \eta_{X}(\sigma) \leq \|\sigma^{-1/2}\|_{\infty} \| \sigma^{1/2}\|_{\infty}\le (\sigma_{\max}/\sigma_{\min})^{1/2},\]
where $\sigma_{\max}$ (resp. $\sigma_{\min}$) is the largest (resp. smallest) eigenvalue of $\sigma$.
\item[$(ii)$] We can interpret $\eta_{X}(\sigma)$ as a measure of the distance from $\sigma$ to $\cB_{X^{c}}$, since its minimal possible value,$\eta_{X}(\sigma) = 1$, is attained  if and only if $\sigma \in \mathcal{B}_{X^{c}}$. Indeed, if $\eta_{X}(\sigma) = 1$, an easy compactness argument in the finite-dimensional space $\cB_{X^{c}}$ shows that there exists $Q \in \cB(\cH_{X^c})$ such that $1=\|Q\sigma^{-1/2}\|_{\infty} \|\sigma^{1/2}Q^{-1}\|_{\infty}$. For a general operator $O$,  the condition $\| O\|_{\infty} \| O^{-1}\|_{\infty} = 1$ implies that $O^{\dagger} O =\mathbbm{1}$, so $\sigma = Q^{\dagger}Q \in \cB_{X^{c}}$. The converse implication is trivial.
\item[$(iii)$] Using \cite[Theorem IX.2.1]{Bhatia97}, we can omit the square roots at the expenses of minimizing over invertible positive elements
\[ \eta_{X}(\sigma) \leq \inf\{ \|Q \sigma^{-1} \|_{\infty} \cdot \|\sigma Q^{-1} \|_{\infty} \colon Q \in \mathcal{B}_{X^{c}} \text{ positive and invertible} \}\,.  \]
\item[$(iv)$] When $\sigma$ is the Gibbs state $\sigma=e^{-\beta H_{\Lambda}}/\Tr(e^{-\beta H_{\Lambda}})$ of a finite-range and commuting Hamiltonian, then we can simply upper bound $\eta_X(\sigma)$ by choosing $Q = e^{-\frac{\beta}{2} H_{X^{c}}}$ in \eqref{equa:roughBoundSpectralGap}:
\begin{equation} \eta_{X}(\sigma) \leq  e^{\beta \|H_{\Lambda} - H_{X^{c}}\|_{\infty}} \leq e^{ \beta |X| \| \Phi\|} \label{eq:roughBoundFiniteRange} \end{equation}
where $\Phi$ is the local interaction generating $H_{\Lambda}$.
\end{itemize}
\end{Rema}

In the particularly simple situation in which $\sigma$ is the maximally mixed state $\frac{1}{d_{\Lambda}} \mathbbm{1}_{\Lambda}$, we can see that the subspaces $W_X$ are exactly equal to  $\cB(\cH_{X^{c}})$, the orthogonal projection onto $W_{X}$ is then given by $\tau_{X}(\cdot) := \frac{1}{d_{X}} \Tr_{X}(\cdot)$, the canonical purified Hamiltonian is commuting and its spectral gap is one. We will use this fact to relate the gap of the canonical purified Hamiltonian of an arbitrary full-rank state $\sigma$ to that of the maximally mixed state, in order to prove the above Theorem~\ref{Theo:roughBoundSpectralGap}.

For any $Q \in \cB(\cH_\Lambda)$, note that $Q \in W_{X}$ if and only if $Q \sigma^{-1/2} \in \cB(\cH_{X^{c}})$, by the very definition of the subspaces $W_{X}$, see \eqref{equa:definingWX}. The following result gives a quantitative version of this relation.

\begin{Lemm}\label{lemm:comparingLocalTermsCPH}
Let $X \subset \Lambda$  and let $R_{X^{c}} \in \cB(\cH_{X^{c}})$ be any invertible element. Then, for every $Q \in \cB(\cH_\Lambda)$, denoting $\widetilde{Q} := Q \sigma^{-1/2}R_{X^{c}}$ we have
\begin{equation}\label{equa:} 
\frac{1}{c^2} \cdot \langle \widetilde{Q}, \tau_{X}^{\perp}(\widetilde{Q}) \rangle_{HS} \leq \langle Q, \Pi_{X}^{\perp}(Q) \rangle_{HS} \leq  C^2 \cdot \langle \widetilde{Q}, \tau_{X}^{\perp}(\widetilde{Q}) \rangle_{HS}\,. 
\end{equation}
where $c=\| R_{X^{c}}^{-1} \sigma^{1/2}\|_{\infty} \|\sigma^{-1/2} R_{X^{c}}\|_{\infty}^{2}$ and $C=\| R_{X^{c}}^{-1}\sigma^{1/2}\|_{\infty}$.
\end{Lemm}

\begin{proof}
First, let us define for each $X \subset \Lambda$ the bounded and linear operator  $T_{X}: \cB(\cH_{\Lambda}) \to \cB(\cH_{\Lambda})$ given by
\[ T_{X}(Q) = \tau_{X}(Q\sigma^{-1/2})\sigma^{1/2} \quad , \quad Q \in \cB(\cH_{\Lambda})\,.\]
It satisfies that $T_{X}(Q) \in W_{X}$ for every $Q \in \cB(\cH_{\Lambda})$, and moreover $T_{X}(Q) = Q$ whenever $Q \in W_{X}$. In other words, $T_{X}$ is a linear projection onto $W_{X}$, but different from $\Pi_{X}$ as it is not orthogonal. Note that for the prefixed invertible element $R_{X^{c}}  \in \cB(\cH_{X^{c}})$, we could rewrite $T_{X}(Q) = \tau_{X}(Q \sigma^{-1/2} R_{X^{c}}) R_{X^{c}}^{-1} \sigma^{1/2}$. Therefore, we can estimate
\begin{equation}\label{equa:comparingLocalTermsCPHaux1} 
\| Q - T_{X}(Q)\|_{2} = \| \tau_{X}^{\perp}(Q \sigma^{-1/2} R_{X^{c}}) R_{X^{c}}^{-1} \sigma^{1/2} \|_{2} \leq \| Q\|_{2} \|\sigma^{-1/2} R_{X^{c}} \|_{\infty}  \| R_{X^{c}}^{-1} \sigma^{1/2}\|_{\infty}\,. 
\end{equation}
With this operator and its properties in mind, let us prove \eqref{lemm:comparingLocalTermsCPH}. Note first, that we can rewrite
\[ \langle Q, \Pi_{X}^{\perp}(Q) \rangle_{HS} = \| \Pi_{X}^{\perp}(Q)\|_{2}^{2} = \|Q - \Pi_{X}(Q)\|_{2}^{2}\quad , \quad \langle \widetilde{Q}, \tau_{X}^{\perp}(\widetilde{Q}) \rangle_{HS} = \| \tau_{X}^{\perp}(\widetilde{Q})\|_{2}^{2}\,.  \]
Using that $\Pi_{X}^{\perp} T_{X} = 0$, we can estimate
\begin{align*} 
\| \Pi_{X}^{\perp}(Q)\|_{2} 
 =  \|  \Pi_{X}^{\perp}(Q - T_{X}(Q))\|_{2} 
\leq \| Q - T_{X}(Q) \|_{2} 
 & =  \|  \tau_{X}^{\perp}(\widetilde{Q}) R_{X^{c}}^{-1}\sigma^{1/2} \|_{2}\\[1.5mm]  
 & \leq \| R_{X^{c}}^{-1} \sigma^{1/2}\|_{\infty} \| \tau_{X}^{\perp}(\widetilde{Q}) \|_{2}  \,.
\end{align*}
On the other hand, using \eqref{equa:comparingLocalTermsCPHaux1}, 
\begin{align*}
\|\tau_{X}^{\perp}(\widetilde{Q})\|_{2} 
= \|(Q  - T_{X}(Q)) \sigma^{-1/2} R_{X^{c}}\|_{2} 
& \leq \| Q - T_{X}(Q)\|_{2} \cdot \| \sigma^{-1/2} R_{X^{c}}\|_{\infty}\\
& = \| Q - \Pi_{X}(Q) - T_{X}(Q - \Pi_{X}(Q))\|_{2} \cdot \| \sigma^{-1/2}R_{X^{c}}\|_{\infty}\\
& = \| \Pi_{X}^{\perp}(Q) - T_{X}(\Pi_{X}^{\perp}(Q))\|_{2} \cdot \| \sigma^{-1/2}R_{X^{c}}\|_{\infty}\\
& \leq \| \Pi_{X}^{\perp}(Q)\|_{2}  \|\sigma^{-1/2} R_{X^{c}} \|_{\infty}  \| R_{X^{c}}^{-1} \sigma^{1/2}\|_{\infty} \|\sigma^{-1/2} R_{X^{c}}\|_{\infty} \,.
\end{align*}
This concludes the proof.
\end{proof}

\begin{proof}[Proof of Theorem~\ref{Theo:roughBoundSpectralGap}]
Let us fix arbitrary $R_{X^{c}} \in \cB(\cH_{X^{c}})$ and $Q \in \mathcal{B}(\cH_\Lambda)$, denote $\widetilde{Q}:=Q\sigma^{-1/2} R_{X^{c}}$. Note that $R_{X^{c}} \in \mathcal{B}(\cH_{\{x\}^{c}})$ for every $x \in X$, which let us apply Lemma~\ref{lemm:comparingLocalTermsCPH} for each $x \in X$ to obtain
\[ 
\langle Q, \mathbf{H}_{X}(Q) \rangle_{HS}
=  \sum_{x \in X} \langle Q, \Pi_{x}^{\perp}(Q) \rangle_{HS}
\geq \frac{1}{c^2} \sum_{x \in X} \langle \widetilde{Q}, \tau_{x}^{\perp}(\widetilde{Q}) \rangle_{HS} \,.
\]
Using next that for the canonical purity Hamiltonian associated to the maximally mixed state, the gap of any $\sum_{x \in X} \tau_{x}^{\perp}$ is always one, and applying again Lemma~\ref{lemm:comparingLocalTermsCPH} we get
\begin{multline*}
\langle Q, \mathbf{H}_{X}(Q) \rangle_{HS} \geq \frac{1}{c^2}\langle \widetilde{Q}, \tau_{X}^{\perp}(\widetilde{Q}) \rangle_{HS}\\
\geq \frac{1}{c^2\cdot C^2} \langle Q, \Pi_{X}^{\perp}(Q) \rangle_{HS} = \frac{1}{\|R_{X^{c}}^{-1} \sigma^{1/2}\|_{\infty}^{4} \cdot \|\sigma^{-1/2} R_{X^{c}}\|_{\infty}^{4}} \cdot \langle Q, \Pi_{X}^{\perp}(Q) \rangle_{HS} \,.
\end{multline*}
Since $R_{X^{c}} \in \cB(\cH_{X^{c}})$ is an arbitrary invertible operator, we can take supremums to get the constant $\eta_{X}(\sigma)$, finishing the proof.
\end{proof}

\subsection{Gap in 1D}\label{sec:gap-1D}

To describe the periodic boundary conditions, we have to introduce further notation. For each natural $N$, let us denote by $\mathbb{S}_{N}$ the quotient $\mathbb{R}/\sim$ where we relate $x \sim x+N$ for every $x \in \mathbb{R}$. Note that we can identify $\mathbb{S}_{N} \equiv [0, N)$. Let us recall the notion of a (closed) \emph{interval} in $\mathbb{S}_{N}$. Given $x,y \in \mathbb{S}_{N}$ we denote by $d_{+}(x,y)$ the unique $0 \leq c < N$ such that $x+c \sim y$.  Then, we define the \emph{interval} $[a,b]$ as the set $\{ x \in \mathbb{S}_{N} \colon d_{+}(a, x) \leq d_{+}(a,b) \}$. 

Let us now consider a quantum spin system over the lattice ring $\Lambda_{N} \equiv \mathbb{Z}_{N} \subset \mathbb{S}_{N}$. In an abuse of notation, we will identify each interval $[a,b]$ with $[a,b] \cap \Lambda_{N}$. 

\begin{Assumption}\label{Assumption:1Dgap}
Let us assume we have a full-rank state $\sigma \in \cB(\cH_{\Lambda_N})$, as well as a positive non-increasing function $\delta:\mathbb{R} \to (0,1)$ satisfying the following condition:  Consider any partition $\Lambda_N = A I_{1} C I_{2}$ into four disjoint intervals as in Figure~\ref{fig:1DringSplitting}, where $I_{1}$ and $I_{2}$ shield $A$ from $C$, and $|I_{1}|, |I_{2}| \geq \ell$ for some $\ell \geq 1$. Then, for any $D \in \{\emptyset, I_{1},I_{2}\}$, it holds that
\begin{equation}\label{equa:decayMartingaleCondition1D}
\Delta_{\sigma}(A:C|D) \leq \delta(\ell) \,. 
\end{equation}
\end{Assumption}

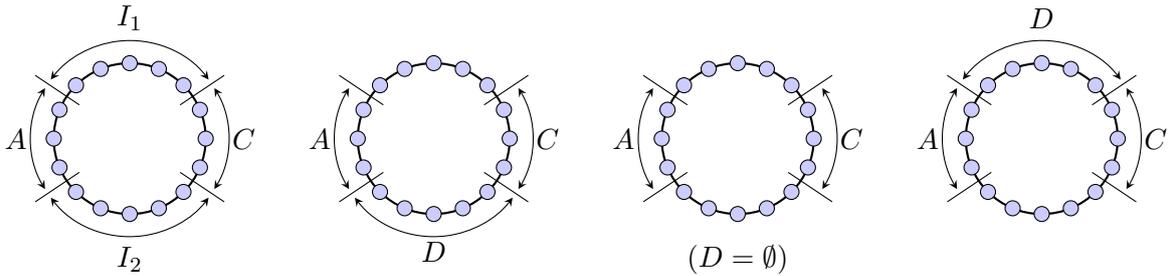
\begin{figure}[bht]
\centering
\begin{tikzpicture}
% Definimos el número de nodos
    \def\nodes{16}
    % Definimos el radio del círculo
    \def\radius{1}
    % Definimos el ángulo para la división en regiones
    \def\angreg{34}

    %Primer dibujo

    \begin{scope} 
    % Dibujamos la circunferencia
    \draw[thick] (0,0) circle(\radius);

    % Dibujamos los nodos del lattice circular
    \foreach \i in {1,...,\nodes} {
        % Calculamos la posición de los nodos en el círculo
        \pgfmathsetmacro{\angle}{360/\nodes * (\i - 1)}
        \node[draw, circle, fill=blue!20, inner sep=2pt] at (\angle:\radius) {};
    }

    % Marcamos las cuatro regiones cortando perpendicularmente la circunferencia con segmentos finos

    \draw[thin] (-\angreg:{\radius-0.2}) -- (-\angreg:{\radius+0.5});

     \draw[thin] (\angreg:{\radius-0.2}) -- (\angreg:{\radius+0.5});

    \draw[thin] ({180-\angreg}:{\radius-0.2}) -- ({180-\angreg}:{\radius+0.5});

     \draw[thin] ({180+\angreg}:{\radius-0.2}) -- ({180+\angreg}:{\radius+0.5});

    %Determinamos las regiones

    \draw[ stealth-stealth] (-\angreg+3:{\radius+0.3}) arc(-\angreg+3:\angreg-3:{\radius+0.3});

    \node at ({\radius+0.5}, 0) {$C$};

     \draw[ stealth-stealth] (\angreg+3:{\radius+0.3}) arc(\angreg+3:{180-\angreg-3}:{\radius+0.3});

     \node at (0, {\radius+0.6}) {$I_{1}$};
     
    \draw[ stealth-stealth] ({180-\angreg+3}:{\radius+0.3}) arc({180-\angreg+3}:{180+\angreg-3}:{\radius+0.3});

    \node at ({-\radius-0.5}, 0) {$A$};

    \draw[ stealth-stealth] ({180+\angreg+3}:{\radius+0.3}) arc({180+\angreg+3}:{360-\angreg-3}:{\radius+0.3});

    \node at (0,{-\radius-0.6}) {$I_{2}$};
    \end{scope}

    %Segundo dibujo

    \begin{scope}[xshift=4cm]  
    % Dibujamos la circunferencia
    \draw[thick] (0,0) circle(\radius);

    % Dibujamos los nodos del lattice circular
    \foreach \i in {1,...,\nodes} {
        % Calculamos la posición de los nodos en el círculo
        \pgfmathsetmacro{\angle}{360/\nodes * (\i - 1)}
        \node[draw, circle, fill=blue!20, inner sep=2pt] at (\angle:\radius) {};
    }

    % Marcamos las cuatro regiones cortando perpendicularmente la circunferencia con segmentos finos

    \draw[thin] (-\angreg:{\radius-0.2}) -- (-\angreg:{\radius+0.5});

     \draw[thin] (\angreg:{\radius-0.2}) -- (\angreg:{\radius+0.5});

    \draw[thin] ({180-\angreg}:{\radius-0.2}) -- ({180-\angreg}:{\radius+0.5});

     \draw[thin] ({180+\angreg}:{\radius-0.2}) -- ({180+\angreg}:{\radius+0.5});

    %Determinamos las regiones

    \draw[ stealth-stealth] (-\angreg+3:{\radius+0.3}) arc(-\angreg+3:\angreg-3:{\radius+0.3});

    \node at ({\radius+0.5}, 0) {$C$};

    \draw[ stealth-stealth] ({180-\angreg+3}:{\radius+0.3}) arc({180-\angreg+3}:{180+\angreg-3}:{\radius+0.3});

    \node at ({-\radius-0.5}, 0) {$A$};

    \draw[ stealth-stealth] ({180+\angreg+3}:{\radius+0.3}) arc({180+\angreg+3}:{360-\angreg-3}:{\radius+0.3});

    \node at (0,{-\radius-0.5}) {$D$};
    \end{scope}

    %Tercer dibujo

    \begin{scope}[xshift=8cm] 
    % Dibujamos la circunferencia
    \draw[thick] (0,0) circle(\radius);

    % Dibujamos los nodos del lattice circular
    \foreach \i in {1,...,\nodes} {
        % Calculamos la posición de los nodos en el círculo
        \pgfmathsetmacro{\angle}{360/\nodes * (\i - 1)}
        \node[draw, circle, fill=blue!20, inner sep=2pt] at (\angle:\radius) {};
    }

    % Marcamos las cuatro regiones cortando perpendicularmente la circunferencia con segmentos finos

    \draw[thin] (-\angreg:{\radius-0.2}) -- (-\angreg:{\radius+0.5});

     \draw[thin] (\angreg:{\radius-0.2}) -- (\angreg:{\radius+0.5});

    \draw[thin] ({180-\angreg}:{\radius-0.2}) -- ({180-\angreg}:{\radius+0.5});

     \draw[thin] ({180+\angreg}:{\radius-0.2}) -- ({180+\angreg}:{\radius+0.5});

    %Determinamos las regiones

    \draw[ stealth-stealth] (-\angreg+3:{\radius+0.3}) arc(-\angreg+3:\angreg-3:{\radius+0.3});

    \node at ({\radius+0.5}, 0) {$C$};

    \draw[ stealth-stealth] ({180-\angreg+3}:{\radius+0.3}) arc({180-\angreg+3}:{180+\angreg-3}:{\radius+0.3});

    \node at ({-\radius-0.5}, 0) {$A$};

    \node at (0,{-\radius-0.6}) {$(D=\emptyset)$};

    \end{scope}

    %Cuarto dibujo

    \begin{scope}[xshift=12cm] 
    % Dibujamos la circunferencia
    \draw[thick] (0,0) circle(\radius);

    % Dibujamos los nodos del lattice circular
    \foreach \i in {1,...,\nodes} {
        % Calculamos la posición de los nodos en el círculo
        \pgfmathsetmacro{\angle}{360/\nodes * (\i - 1)}
        \node[draw, circle, fill=blue!20, inner sep=2pt] at (\angle:\radius) {};
    }

    % Marcamos las cuatro regiones cortando perpendicularmente la circunferencia con segmentos finos

    \draw[thin] (-\angreg:{\radius-0.2}) -- (-\angreg:{\radius+0.5});

     \draw[thin] (\angreg:{\radius-0.2}) -- (\angreg:{\radius+0.5});

    \draw[thin] ({180-\angreg}:{\radius-0.2}) -- ({180-\angreg}:{\radius+0.5});

     \draw[thin] ({180+\angreg}:{\radius-0.2}) -- ({180+\angreg}:{\radius+0.5});

    %Determinamos las regiones

    \draw[ stealth-stealth] (\angreg+3:{\radius+0.3}) arc(\angreg+3:{180-\angreg-3}:{\radius+0.3});

     \node at (0, {\radius+0.6}) {$D$};
    
    \draw[ stealth-stealth] (-\angreg+3:{\radius+0.3}) arc(-\angreg+3:\angreg-3:{\radius+0.3});

    \node at ({\radius+0.5}, 0) {$C$};

    \draw[ stealth-stealth] ({180-\angreg+3}:{\radius+0.3}) arc({180-\angreg+3}:{180+\angreg-3}:{\radius+0.3});

    \node at ({-\radius-0.5}, 0) {$A$};

    \end{scope}
    
    \end{tikzpicture}
    \caption{The first picture represent a splitting of the ring as in Assumption~\ref{Assumption:1Dgap}, whereas the other  pictures correspond to the three possible choices for $D$ when considering $\Delta_{\sigma}(A:C|D)$.}
    \label{fig:1DringSplitting}
\end{figure}

Under the previous condition, we can prove the following result. 
%We first use the recursive Lemma~\ref{Lemm:DivideAndConquer} to establish a lower bound of the spectral gap of $\mathbf{H}_{\Lambda_{N}}$ in terms of the gap of smaller intervals.

\begin{Theo}[1D models]\label{Theo:gapIntervals1D}
Let $\mathbf{H}_{\Lambda_N}$ be the canonical purified Hamiltonian of $\sigma$. For fixed integers $N \geq \mu \geq 9$, let us denote  $\delta_{k} = \delta(\lfloor \frac{\mu}{9}(9/8)^{k}\rfloor)$ for each $k \geq 0$. Then,
\begin{equation}
    \label{eq:gapIntervals1D}
\operatorname{gap}(\mathbf{H}_{\Lambda_{N}}) \geq e^{-5} \cdot \qty(\prod_{k=0}^{\infty} (1-\delta_{k}) ) \cdot \min\qty{ \frac{1}{\eta_I(\sigma)^4} \colon I \subset \Lambda_{N} \text{ interval with }|I| \leq \mu} .
\end{equation}
\end{Theo}

\begin{Rema}\label{Rema:infiniteProductConvergence}
The constant $\prod_k (1-\delta_k)$ can be seen to be strictly positive
whenever $(\delta_k)_{k}$ is a summable sequence, since
\[
\prod_{k=0}^{\infty} (1-\delta_{k}) \ge \prod_{k=0}^{k_0-1}(1-\delta_k) \prod_{k=k_0}^{\infty} \frac{1}{1+2\delta_k} \ge \prod_{k=0}^{k_0-1}(1-\delta_k) \,\,  \exp(-2\sum_{k\ge k_0} \delta_k),
\]
where $k_0$ is the smallest value of $k$ such that $\delta_k < 1/2$ for every $k\ge k_0$.
\end{Rema}

\begin{proof}
We will use the floor $\lfloor \cdot \rfloor$ and ceiling $\lceil \cdot \rceil$ functions. We begin by considering a partition $\Lambda = A I_{1} C I_{2}$ as in Figure~\ref{fig:1DringSplitting} with sizes $|A|=|C|=1$, $|B_{1}| = \lfloor N/2 \rfloor -1$ and $|B_{2}| = \lceil N/2 \rceil-1$. Note that the identity $\lceil N/2 \rceil+ \lfloor N/2 \rfloor = N$ guarantees that such a partition is indeed possible. If we define the intervals $L = I_{2}AI_{1}$ and $R = I_{1}CI_{2}$, then Assumption~\ref{Assumption:1Dgap} yields that 
\[ \| \Pi_{L} \Pi_{R} -\Pi_{\Lambda_{N}} \| = \Delta_{\sigma}(A:C|\emptyset) \leq \delta(\lfloor N/2 \rfloor-1) \leq \delta_{0}.\]
Applying Lemma~\ref{Lemm:DivideAndConquer}, we obtain 
\begin{equation}\label{equa:gapIntervals1Daux1} 
\gap(\mathbf{H}_{\Lambda_{N}}) \geq \frac{1-\delta_{0}}{2} \min\{\gap(\mathbf{H}_{L}), \gap(\mathbf{H}_{R})\} \geq  \frac{1-\delta_{0}}{2} \min_{I \in \mathcal{F}}\gap(\mathbf{H}_{I})\,, 
\end{equation}
where $\mathcal{F}$ is the set of all intervals $I \subset \Lambda_{N}$ such that $|I| \leq \lceil N/2 \rceil + 1$. To estimate the last factor from below, we will adopt a similar strategy, recursively estimating the gap on a given interval in terms of the gaps of smaller subintervals. We start defining for each $k \geq 0$
\[
\mathcal{F}_{k} = \{ I \in \mathcal{F} \colon |I| \leq \mu (3/2)^{k}  \} \,.
\]
Observe that this is an increasing sequence $\mathcal{F}_{k} \subset \mathcal{F}_{k+1} \subset \mathcal{F}$ which eventually stabilizes, i.e. $\mathcal{F}_{k} = \mathcal{F}$ for some $k$. We claim that for every $k \geq 1$, denoting $s_{k}:=\lfloor (4/3)^{k} \rfloor$,
\begin{equation}\label{equa:gapIntervals1Daux2} 
\min_{I \in \mathcal{F}_{k}} \gap(\mathbf{H}_{I}) \geq \frac{1-\delta_{k}}{1+\frac{1}{s_{k}}} \min_{I' \in \mathcal{F}_{k-1}} \gap(\mathbf{H}_{I'})\,.
\end{equation}
To prove this, let us fix
$k\geq 1$ and an arbitrary $I \in \mathcal{F}_{k}$. If $I \in \mathcal{F}_{k-1}$, the inequality holds trivially, since the constant $\frac{1-\delta_{k}}{1+1/s_{k}}$ is smaller than one. Thus, we assume that $I \in \mathcal{F}_{k} \setminus \mathcal{F}_{k-1}$, implying
\[ \mu (3/2)^{k-1} < |I| \leq \mu (3/2)^{k}\,. \]
If we denote $|I| = m$, we can identify $I$ with the interval $[1,m]$. We will need the following easy estimate at several points of the argument
\begin{equation}\label{equa:gapIntervals1DAux} 
\frac{m}{6s_{k}} \geq \frac{\mu(3/2)^{k-1}}{6(4/3)^{k}} =  \frac{\mu}{9} \left( \frac{9}{8}\right)^{k} \geq 1\,. 
\end{equation}
Next, we define for each $j=1, \ldots, s_{k}$
\[ \textstyle
L_{j} = \left[1 \, , \,  \frac{m}{3} + (2j) \frac{m}{6 s_{k}}\right]  \quad , \quad R_{j} = \left[\frac{m}{3} + (2j-1) \frac{m}{6 s_{k}} \, , \, m\right]\,.
\]
If the endpoints are non-integers, we interpret the intervals as the set of integers (i.e. lattice points) they contain. We then have
\[ |L_{j}| \leq \frac{m}{3} + j \frac{m}{3 s_{k}} \leq m-\frac{m}{3}   \leq  \mu \left( \frac{3}{2}\right)^{k-1}\,, \]
and similarly
\[ |R_{j}| \leq \frac{2m}{3} - (2j-1)\frac{m}{6s_{k}}+1 \leq \frac{2m}{3} - \frac{m}{6s_{k}}+1  \leq \frac{2m}{3} \leq  \mu \left( \frac{3}{2}\right)^{k-1} \,, \]
where the third inequality follows from \eqref{equa:gapIntervals1DAux}. This shows that $L_{j}, R_{j} \in \mathcal{F}_{k-1}$ for every $j=1, \ldots, s_{k}$. On the other hand, the intersection of each pair $L_{j}, R_{j}$ is a new interval
\[ \textstyle L_{j} \cap R_{j} = \left[\frac{m}{3} + (2j-1) \frac{m}{6 s_{k}}, \frac{m}{3} + (2j) \frac{m}{6 s_{k}}\right]\,.  \]
Note that $(L_{i} \cap R_{i}) \cap (L_{j} \cap R_{j}) = \emptyset$ if $i \neq j$. Moreover, applying again \eqref{equa:gapIntervals1DAux} we can lower estimate
\[ |L_{j} \cap R_{j}| \geq \lfloor \tfrac{m}{6 s_{k}}\rfloor  \geq \lfloor \tfrac{\mu}{9}  \left( \tfrac{9}{8}\right)^{k}\rfloor \geq 1\,. \]
Each pair $(L_{j}, R_{j})$ yields a partition $\Lambda_{N} = A_{j}B_{j}C_{j}D_{j}$ into four disjoint consecutive intervals such that $I=A_{j}B_{j}C_{j}$, $L_{j} = A_{j}B_{j}$, $R_{j} = B_{j}C_{j}$ and so $B_{j} = L_{j} \cap R_{j}$. Moreover, $|B_{j}| \geq \lfloor \tfrac{m}{6 s_{k}}\rfloor $ and $|D_{j}| \geq N-m \geq \tfrac{m}{6} \geq \tfrac{m}{6s_{k}}$ since $m \leq \lceil N/2 \rceil + 1$ and $N \geq 9$. Therefore, by Assumption~\ref{Assumption:1Dgap}, we can estimate
\[ \| \Pi_{L_{j}} \Pi_{R_{j}} - \Pi_{I}\| = \Delta_{\sigma}(A_{j}:C_{j}|D_{j}) \leq \delta(\lfloor \tfrac{m}{6s_{k}} \rfloor) \leq \delta_{k}\,. \]
Applying Lemma~\ref{Lemm:DivideAndConquer}, and the fact that $L_{j}, R_{j} \in \mathcal{F}_{k-1}$ for every $j$, we get that
\begin{equation}\label{equa:gapIntervals1DAux3} 
\operatorname{gap}(\mathbf{H}_{I}) \geq \frac{1-\delta_{k}}{1 + \frac{1}{s_{k}}} \min_{I' \in \mathcal{F}_{k-1}}\operatorname{gap}(\mathbf{H}_{I'})\,. 
\end{equation}
This proves the claim. Thus, combining \eqref{equa:gapIntervals1Daux1} and \eqref{equa:gapIntervals1Daux2}, together with the fact that $\mathcal{F}_{k} = \mathcal{F}$ if $k$ is large enough, we deduce that
\[ 
\operatorname{gap}(\mathbf{H}_{\Lambda_{N}}) \geq  \qty(\prod_{k=0}^{\infty} \frac{1-\delta_{k}}{1+\frac{1}{s_{k}}} ) \cdot \min_{I \in \mathcal{F}_{0}} \operatorname{gap}(\mathbf{H}_{I}) .
\]
We can finish the proof by applying Theorem~\ref{Theo:roughBoundSpectralGap}, and noticing that
\[
\prod_{k=0}^{\infty} \frac{1}{1+\frac{1}{s_{k}}} \geq \exp(-\sum_{k=0}^{\infty} \frac{1}{s_{k}}) \ge \exp( - \sum_{k=0}^6 \frac{1}{s_k} - \sum_{k=7}^\infty \qty(\frac{3}{4})^{k-1} ) \ge \exp(-5),
\]
where we used the fact that $s_k \ge (4/3)^{k-1}$ when $k \ge 7$.
\end{proof}

\subsection{Gap in 2D}\label{sec:gap-2D}

We will take $\Lambda_{N}$ (the set where the \emph{spins} of the system are located) as the set of midpoints of the edges $\cE_{N}$ of the square lattice $\mathbb{Z}_{N} \times \mathbb{Z}_{N}$ on the torus $\mathbb{S}_{N} \times \mathbb{S}_{N}$, since this is the setting in which the Quantum Double Models are defined (see Figure~\ref{Fig:squareLatticeFirstRep}). We will identify each point of $\Lambda_{N}$ with the corresponding edge from $\cE_{N}$ in the forthcoming figures.

\begin{figure}[ht] 
\centering
  \begin{tikzpicture}[equation,scale=0.5]
    \draw[gray, thin] (0,0) grid (6,6);
    \draw[gray,thick, postaction=torus horizontal] (0,0) -- (6,0);
    \draw[gray,thick, postaction=torus horizontal] (0,6) -- (6,6);
    \draw[gray,thick, postaction=torus vertical] (0,0) -- (0,6);
    \draw[gray,thick, postaction=torus vertical] (6,0) -- (6,6);
    \end{tikzpicture}
    \hspace{2cm}
    \begin{tikzpicture}[equation,scale=0.5]
    \draw[gray, thin] (0,0) grid (6,6);
    \draw[gray,thick, postaction=torus horizontal] (0,0) -- (6,0);
    \draw[gray,thick, postaction=torus horizontal] (0,6) -- (6,6);
    \draw[gray,thick, postaction=torus vertical] (0,0) -- (0,6);
    \draw[gray,thick, postaction=torus vertical] (6,0) -- (6,6);
    \foreach \x in {0,...,5}
        \foreach \y in {0,...,6}
            \node[draw, circle, fill=blue!20, inner sep=1.5pt] at (0.5+\x,\y) {};
             \foreach \x in {0,...,6}
        \foreach \y in {0,...,5}
            \node[draw, circle, fill=blue!20, inner sep=1.5pt] at (\x,0.5+\y) {};
    \end{tikzpicture}\\[2mm]
\caption{The square lattice on a torus (left) and a quantum spin system with spins located at the midpoints of the edges (right). The markings along the borders of both squares indicate the pairwise identification of edges, following a standard topological representation of the torus. A similar convention will be used for the cylinder, where only one pair of opposite edges is identified.}
\label{Fig:squareLatticeFirstRep}
\end{figure}
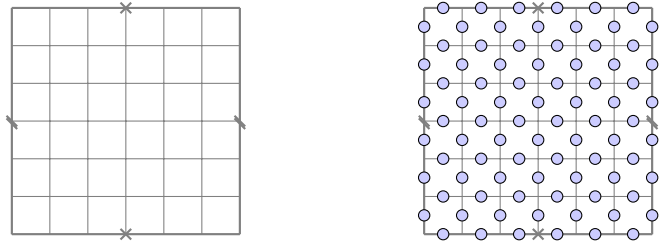

A \emph{proper rectangle} $\cR$ in $\mathbb{S}_{N} \times \mathbb{S}_{N}$ is a Cartesian product of intervals $\cR=[a_{1}, b_{1}] \times [a_{2}, b_{2}]$ where $a_{1}, b_{1}, a_{2}, c_{2} \in \mathbb{Z}_{N}$ (we will only work with intervals having endpoints in $\mathbb{Z}_{N}$). In this case, its number of plaquettes per row is then $d_{+}(a_{1},b_1)$ and per column is $d_{+}(a_{2},b_2)$. Shortly, we will say that $\cR$ has \emph{dimensions} $d_{+}(a_{1},b_1)$ and $d_{+}(a_{2},b_2)$. Note that both dimensions are less than or equal to $N-1$. A \emph{cylinder} is a Cartesian product of the form $\mathbb{S}_{N} \times [a,b]$ or $[a,b] \times \mathbb{S}_{N}$. It has dimensions $d_{+}(a,b) \leq N-1$ and $N$. We will use the term \emph{rectangle} to refer to proper rectangles, cylinders or the whole torus $\mathbb{S}_{N} \times \mathbb{S}_{N}$, which has both dimensions equal to $N$. In an abuse of notation, we will identify $\cR$ with $\cR \cap \Lambda_N$ and write $\cR \subset \Lambda_N$ (see Figure~\ref{Figure:rectanglesTorus}). 

\begin{figure}[ht]
\centering
\begin{tikzpicture}[equation,scale=0.4]
    \draw[gray, thin] (0,0) grid (6,6);
    \draw[gray,thick, postaction=torus horizontal] (0,0) -- (6,0);
    \draw[gray,thick, postaction=torus horizontal] (0,6) -- (6,6);
    \draw[gray,thick, postaction=torus vertical] (0,0) -- (0,6);
    \draw[gray,thick, postaction=torus vertical] (6,0) -- (6,6);
    \draw[black, ultra thick] (2,2) grid (5,4);
    \end{tikzpicture}
    \hspace{1cm} %%%%%%%%%%%%%%%%%%%%%%%%%%%%%%
    \begin{tikzpicture}[equation,scale=0.4]
    \draw[gray, thin] (0,0) grid (6,6);
    \draw[gray,thick, postaction=torus horizontal] (0,0) -- (6,0);
    \draw[gray,thick, postaction=torus horizontal] (0,6) -- (6,6);
    \draw[gray,thick, postaction=torus vertical] (0,0) -- (0,6);
    \draw[gray,thick, postaction=torus vertical] (6,0) -- (6,6);
     \foreach \x in {2,...,5}
        \foreach \y in {2,...,3}
            \node[draw, circle, fill=blue!20, inner sep=1.5pt] at (\x,\y+0.5) {};
    \foreach \x in {2,...,4}
        \foreach \y in {2,...,4}
        \node[draw, circle, fill=blue!20, inner sep=1.5pt] at (\x+0.5,\y) {};
    \end{tikzpicture}
    \hspace{1.5cm} %%%%%%%%%%%%%%%%%%%%%%%%%%%%%%
    \begin{tikzpicture}[equation,scale=0.4]
    \draw[gray, thin] (0,0) grid (6,6);
    \draw[gray,thick, postaction=torus horizontal] (0,0) -- (6,0);
    \draw[gray,thick, postaction=torus horizontal] (0,6) -- (6,6);
    \draw[gray,thick, postaction=torus vertical] (0,0) -- (0,6);
    \draw[gray,thick, postaction=torus vertical] (6,0) -- (6,6);
    \draw[black, ultra thick] (2,4) grid (5,6);
    \draw[black, ultra thick] (2,0) grid (5,1);
    \end{tikzpicture}
    \hspace{1cm} %%%%%%%%%%%%%%%%%%%%%%%%%%%%%%
    \begin{tikzpicture}[equation,scale=0.4]
    \draw[gray, thin] (0,0) grid (6,6);
    \draw[gray,thick, postaction=torus horizontal] (0,0) -- (6,0);
    \draw[gray,thick, postaction=torus horizontal] (0,6) -- (6,6);
    \draw[gray,thick, postaction=torus vertical] (0,0) -- (0,6);
    \draw[gray,thick, postaction=torus vertical] (6,0) -- (6,6);
     \foreach \x in {2,...,5}
        \foreach \y in {4,...,5}
            \node[draw, circle, fill=blue!20, inner sep=1.5pt] at (\x,\y+0.5) {};
    \foreach \x in {2,...,4}
        \foreach \y in {4,...,6}
        \node[draw, circle, fill=blue!20, inner sep=1.5pt] at (\x+0.5,\y) {};
    \foreach \x in {2,...,5}
        \foreach \y in {0}
            \node[draw, circle, fill=blue!20, inner sep=1.5pt] at (\x,\y+0.5) {};
    \foreach \x in {2,...,4}
        \foreach \y in {0,...,1}
        \node[draw, circle, fill=blue!20, inner sep=1.5pt] at (\x+0.5,\y) {};    
    \end{tikzpicture}\\[10mm]

    \begin{tikzpicture}[equation,scale=0.4]
    \draw[gray, thin] (0,0) grid (6,6);
    \draw[gray,thick, postaction=torus horizontal] (0,0) -- (6,0);
    \draw[gray,thick, postaction=torus horizontal] (0,6) -- (6,6);
    \draw[gray,thick, postaction=torus vertical] (0,0) -- (0,6);
    \draw[gray,thick, postaction=torus vertical] (6,0) -- (6,6);
    \draw[black, ultra thick] (0,1) grid (6,3);
    \end{tikzpicture}
    \hspace{1cm}
    \begin{tikzpicture}[equation,scale=0.4]
    \draw[gray, thin] (0,0) grid (6,6);
    \draw[gray,thick, postaction=torus horizontal] (0,0) -- (6,0);
    \draw[gray,thick, postaction=torus horizontal] (0,6) -- (6,6);
    \draw[gray,thick, postaction=torus vertical] (0,0) -- (0,6);
    \draw[gray,thick, postaction=torus vertical] (6,0) -- (6,6);
    \foreach \x in {0,...,6}
        \foreach \y in {1,...,2}
            \node[draw, circle, fill=blue!20, inner sep=1.5pt] at (\x,\y+0.5) {};
    \foreach \x in {0,...,5}
        \foreach \y in {1,...,3}
        \node[draw, circle, fill=blue!20, inner sep=1.5pt] at (\x+0.5,\y) {}; 
    \end{tikzpicture}
    \hspace{1.5cm}
    \begin{tikzpicture}[equation,scale=0.4]
    \draw[gray] (0,0) grid (6,6);
    \draw[gray,thick, postaction=torus horizontal] (0,0) -- (6,0);
    \draw[gray,thick, postaction=torus horizontal] (0,6) -- (6,6);
    \draw[gray,thick, postaction=torus vertical] (0,0) -- (0,6);
    \draw[gray,thick, postaction=torus vertical] (6,0) -- (6,6);
    \draw[black, ultra thick] (3,0) grid (5,6);
    \end{tikzpicture}
    \hspace{1cm}
    \begin{tikzpicture}[equation,scale=0.4]
    \draw[gray] (0,0) grid (6,6);
    \draw[gray,thick, postaction=torus horizontal] (0,0) -- (6,0);
    \draw[gray,thick, postaction=torus horizontal] (0,6) -- (6,6);
    \draw[gray,thick, postaction=torus vertical] (0,0) -- (0,6);
    \draw[gray,thick, postaction=torus vertical] (6,0) -- (6,6);
    \foreach \x in {3,...,5}
        \foreach \y in {0,...,5}
            \node[draw, circle, fill=blue!20, inner sep=1.5pt] at (\x,\y+0.5) {};
    \foreach \x in {3,...,4}
        \foreach \y in {0,...,6}
        \node[draw, circle, fill=blue!20, inner sep=1.5pt] at (\x+0.5,\y) {}; 
    \end{tikzpicture}\\[2mm]
    
\caption{Examples of rectangular regions within the square lattice $\Lambda_{N} \equiv \cE_{N}$. The top row shows two examples of proper rectangles, while the bottom row displays two examples of cylinders. In each case, the left image highlights the edges within the region, and the right image highlights the corresponding spins. In the figures that follow, we will adopt the edge-based (left) representation.}
\label{Figure:rectanglesTorus}
\end{figure}
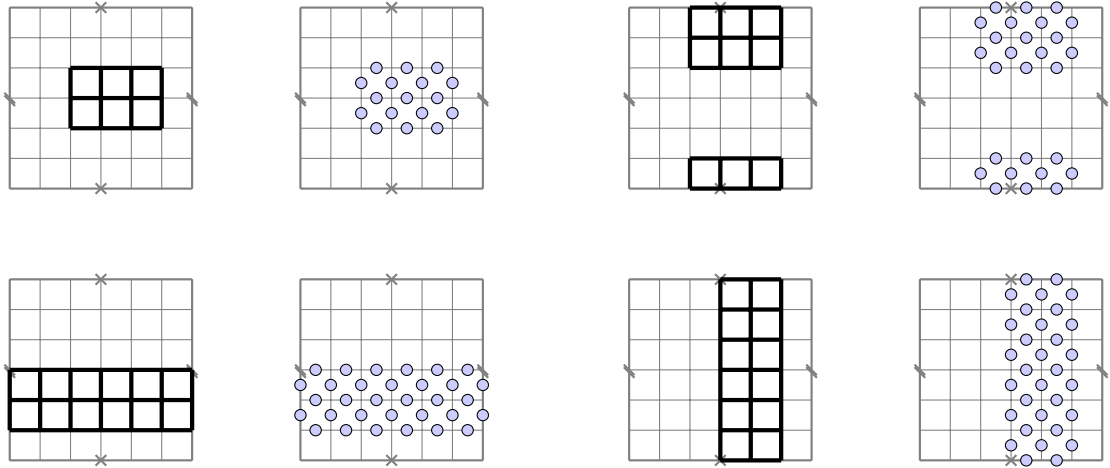

We will say that a rectangle with dimensions $d_{1}$ and $d_{2}$ is \emph{admissible} if it satisfies $d_{1} \leq 6d_{2}$ and $d_{2} \leq 6d_{1}$. This condition ensures that neither dimension is excessively longer than the other; in other words, admissible rectangles are not too narrow or too wide. We will denote by $\mathcal{F}$ the set of all proper rectangles that are admissible. If $\mu>0$ is an integer, we will denote by $\mathcal{F}_{\mu}$ the subset of admissible proper rectangles having both dimensions less than or equal to $\mu$.

%  We are going to define for every $N,\mu \in \mathbb{N}$ with $N \geq \mu \geq 2$ the following sets of rectangular regions in $\cE_{N}$:
 
% \begin{itemize}
%     \item[$\triangleright$] $\mathcal{F}_{N}$ is the set of all rectangular regions having at least two plaquettes per row and per column.
%     \item[$\triangleright$] $\mathcal{F}_{N}^{torus}$ is the family consisting of only one element, the whole torus. 
%     \item[$\triangleright$] $\mathcal{F}_{N}^{cylin}$ is the family of all cylinders having at least two plaquettes per row and per column.
%     \item[$\triangleright$] $\mathcal{F}_{N, \mu}^{rect}$ is the family of all proper rectangles having at least two and at most $\mu$ plaquettes per row and per column. 
% \end{itemize}

\begin{Assumption}\label{Assumtion2Dtorus}\mbox{}
Let us assume we have an invertible state $\sigma \in \mathcal{B}(\mathcal{H}_{\Lambda_N})$, as well as a positive non-increasing function $\delta: \mathbb{R} \to (0,1)$ satisfying the following conditions for every $\ell \geq 2$: 
\begin{enumerate}
\item[($i$)] Consider any partition of the torus $\Lambda_{N}$ along either of its two wrapping directions into four disjoint subsets, denoted $\Lambda_{N}=AB_{1}CB_{2}$, as illustrated in the next picture:

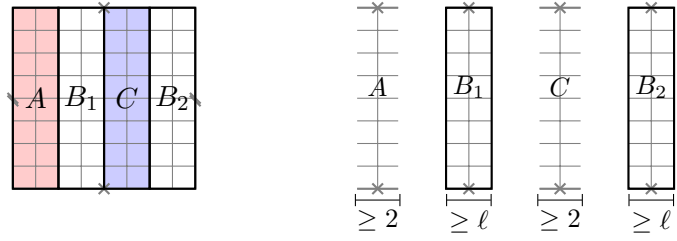
\[
\begin{tikzpicture}[equation,scale=0.3]
\begin{scope}
    \filldraw[red!20] (0,0) rectangle (2,8);
    \filldraw[blue!20] (4,0) rectangle (6,8);
    \draw[gray, thin] (0,0) grid (8,8);
    \draw[gray,thick, postaction=torus horizontal] (0,0) -- (8,0);
    \draw[gray,thick, postaction=torus horizontal] (0,8) -- (8,8);
    \draw[gray,thick, postaction=torus vertical] (0,0) -- (0,8);
    \draw[gray,thick, postaction=torus vertical] (8,0) -- (8,8);
    \draw[black, thick] (0,0) rectangle (2,8);
    \node at (1, 4) {$A$};
    \draw[black, thick] (2,0) rectangle (4,8);
    \node at (3, 4) {$B_{1}$};
    \draw[black, thick] (4,0) rectangle (6,8);
    \node at (5, 4) {$C$};
    \draw[black, thick] (6,0) rectangle (8,8);
    \node at (7, 4) {$B_{2}$};
\end{scope}
\begin{scope}[xshift=15cm]
    \draw[step=1,gray,thin] (0.1,0) grid (1.9,8);
    \draw[thick, gray, postaction=torus horizontal] (0.1,0) -- (1.9,0); 
    \draw[thick, gray, postaction=torus horizontal] (0.1,8) -- (1.9,8);
    \draw (1,4.5) node{\small $A$};
    \draw[|-|] (0,-0.5) -- (1, -0.5) node[below, black] {\small $ \geq 2$} -- (2,-0.5);
    
    \draw[step=1,gray,thin,xshift=2cm] (2,0) grid (4,8);
    \draw[thick, gray, postaction=torus horizontal,xshift=2cm] (2,0) -- (4,0); 
    \draw[thick, gray, postaction=torus horizontal,xshift=2cm] (2,8) -- (4,8); 
    \draw[black,thick,xshift=2cm] (2,0) rectangle (4,8);
    \draw[xshift=2cm] (3,4.5) node{\small $B_1$};
    \draw[|-|,xshift=2cm] (2,-0.5) -- (3, -0.5)   node[below, black] {\small $ \geq \ell$} -- (4,-0.5);
    
    \draw[step=1,gray,thin,xshift=4cm] (4.1,0) grid (5.9,8);
    \draw[thick, gray, postaction=torus horizontal,xshift=4cm] (4.1,0) -- (5.9,0); 
    \draw[thick, gray, postaction=torus horizontal,xshift=4cm] (4.1,8) -- (5.9,8); 
    \draw[xshift=4cm] (5,4.5) node{\small $C$};
    \draw[|-|,xshift=4cm] (4,-0.5) -- (5, -0.5)   node[below, black] {\small $ \geq 2$} -- (6,-0.5);

    \draw[step=1,gray,thin,xshift=6cm] (6,0) grid (8,8);
    \draw[thick, gray, postaction=torus horizontal,xshift=6cm] (6,0) -- (8,0); 
    \draw[thick, gray, postaction=torus horizontal,xshift=6cm] (6,8) -- (8,8); 
    \draw[xshift=6cm] (7,4.5) node{\small $B_2$};
    \draw[|-|,xshift=6cm] (6,-0.5) -- (7, -0.5) node[below, black] {\small $ \geq \ell$} -- (8,-0.5);
    \draw[black,thick,xshift=6cm] (6,0) rectangle (8,8);
\end{scope}
    \end{tikzpicture}
\]
In this decomposition, the regions $A$, $B_1$, $C$, and $B_2$ are arranged sequentially along the chosen direction; the regions $\mathcal{R}_1 := B_2AB_1$ and $\mathcal{R}_2 := B_1CB_2$ each form a cylinder with their intersection $B=B_{1}B_{2}$ consisting of two cylinders $B_1$ and $B_2$ separated by at least two plaquettes; and each of $B_1$ and $B_2$ contains at least $\ell$ plaquettes along the splitting direction. Then, it holds that
\[ \Delta_{\sigma}(A:C|D=\emptyset) \leq \delta(\ell)\,. \]

\item[($ii$)] Consider any cylinder $\mathcal{R}$ whose smaller dimension is larger than or equal to $\ell$, and partition it along its wrapping direction into four disjoint subsets, denoted $\mathcal{R} = AB_1CB_2$, as illustrated in the next picture where $D=\Lambda_{N} \setminus \mathcal{R}$: 
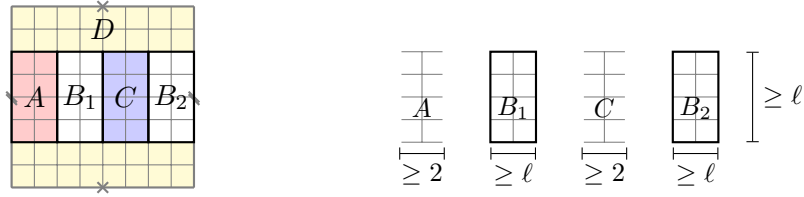
\[
\begin{tikzpicture}[scale=0.3]
\begin{scope}
    \filldraw[yellow!20] (0,0) rectangle (8,2);
    \filldraw[yellow!20] (0,6) rectangle (8,8);
    \filldraw[red!20] (0,2) rectangle (2,6);
    \filldraw[blue!20] (4,2) rectangle (6,6);
    \draw[gray, thin] (0,0) grid (8,8);
    \draw[gray,thick, postaction=torus horizontal] (0,0) -- (8,0);
    \draw[gray,thick, postaction=torus horizontal] (0,8) -- (8,8);
    \draw[gray,thick, postaction=torus vertical] (0,0) -- (0,8);
    \draw[gray,thick, postaction=torus vertical] (8,0) -- (8,8);
    \draw[black, thick] (0,2) rectangle (2,6);
    \node at (1, 4) {$A$};
    \draw[black, thick] (2,2) rectangle (4,6);
    \node at (3, 4) {$B_{1}$};
    \draw[black, thick] (4,2) rectangle (6,6);
    \node at (5, 4) {$C$};
    \draw[black, thick] (6,2) rectangle (8,6);
    \node at (7, 4) {$B_{2}$};
    \node at (4, 7) {$D$};
    \end{scope}
\begin{scope}[xshift=17cm, yshift=2cm]
%Region A
\draw[step=1,white!50!black,thin] (0.1,0) grid (1.9,4);
\draw (1,1.5) node{\small $A$};
\draw[|-|] (0,-0.5) -- (1, -0.5)   node[below, black] {\small $\geq 2$} -- (2,-0.5);
%Region B_1
\draw[step=1,white!50!black,thin,xshift=2cm] (2,0) grid (4,4);
\draw[xshift=2cm] (3,1.5) node{\small $B_1$};
\draw[black,thick,xshift=2cm] (2,0) rectangle (4,4);
\draw[|-|,xshift=2cm] (2,-0.5) -- (3, -0.5)   node[below, black] {\small $\geq \ell$} -- (4,-0.5);
%Region C
\draw[step=1,white!50!black,thin,xshift=4cm] (4.1,0) grid (5.9,4);
\draw[xshift=4cm] (5,1.5) node{\small $C$};
\draw[|-|,xshift=4cm] (4,-0.5) -- (5, -0.5)   node[below, black] {\small $\geq 2$} -- (6,-0.5);
%Region B_2
\draw[step=1,white!50!black,thin,xshift=6cm] (6,0) grid (8,4);
\draw[xshift=6cm] (7,1.5) node{\small $B_2$};
\draw[black,thick,xshift=6cm] (6,0) rectangle (8,4);
\draw[|-|,xshift=6cm] (6,-0.5) -- (7, -0.5) node[below, black] {\small $\geq \ell$} -- (8,-0.5);
%%Vertical size
\draw[|-|,xshift=7.5cm] (8,0) -- (8,2)   node[right, black] {\small $\geq \ell$} -- (8,4);
\end{scope}
\end{tikzpicture}
\]
In this decomposition, the regions $A$, $B_1$, $C$, and $B_2$ are arranged sequentially along the wrapping direction; the regions $\mathcal{R}_1 := B_2AB_1$ and $\mathcal{R}_2 := B_1CB_2$ are proper rectangles  with intersection $B=B_{1}B_{2}$ consisting of two proper rectangles $B_1$ and $B_2$ separated by at least two plaquettes; and each $B_{1}$ and $B_{2}$ contain at least $\ell$ plaquettes along both coordinate directions. Then, it holds that
\[ \Delta(A:C|D) \leq \delta(\ell)\,. \]

\item[($iii$)] Consider any proper rectangle $\mathcal{R}$ with dimensions less than or equal to $N-\ell$, and partition it along the coordinate direction of its larger dimension into three disjoint regions, denoted $\mathcal{R} = ABC$, as illustrated in the next picture, where $D:=\Lambda_{N} \setminus \mathcal{R}$:  

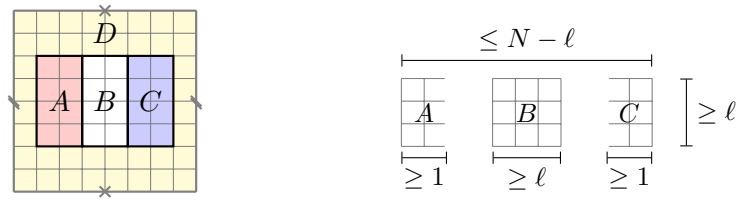
\[
\begin{tikzpicture}[scale=0.3]
\begin{scope}
    \filldraw[yellow!20] (0,0) rectangle (8,8);
    \filldraw[white] (1,2) rectangle (7,6);
    \filldraw[red!20] (1,2) rectangle (3,6);
    \filldraw[blue!20] (5,2) rectangle (7,6);
    \draw[gray, thin] (0,0) grid (8,8);
    \draw[gray,thick, postaction=torus horizontal] (0,0) -- (8,0);
    \draw[gray,thick, postaction=torus horizontal] (0,8) -- (8,8);
    \draw[gray,thick, postaction=torus vertical] (0,0) -- (0,8);
    \draw[gray,thick, postaction=torus vertical] (8,0) -- (8,8);
    \draw[black, thick] (1,2) rectangle (3,6);
    \node at (2, 4) {$A$};
    \draw[black, thick] (3,2) rectangle (5,6);
    \node at (4, 4) {$B$};
    \draw[black, thick] (5,2) rectangle (7,6);
    \node at (6, 4) {$C$};
    \node at (4, 7) {$D$};
\end{scope}
\begin{scope}[xshift=17cm, yshift=2cm]
    %Region A
     \draw[gray] (0,0) grid (1.9,3);
     \draw (1,1.5) node{\small $A$};
     \draw[|-|] (0,-0.5) -- (1, -0.5) 
    node[below, black] {\small $ \geq 1$} -- (2,-0.5);
     %Region B
     \draw[gray,xshift=2cm] (2,0) grid (5,3);
     \draw[xshift=2cm] (3.5,1.5) node{\small $B$};
     \draw[|-|,xshift=2cm] (2,-0.5) -- (3.5, -0.5) 
    node[below, black] {\small $\geq \ell$} -- (5,-0.5);
     %Region C
     \draw[gray, xshift=4cm] (5.1,0) grid (7,3);
      \draw[xshift=4cm] (6,1.5) node{\small $C$};
     \draw[|-|,xshift=4cm] (5,-0.5) -- (6, -0.5) 
    node[below, black] {\small $\geq 1$} -- (7,-0.5);
     %% Vertical size
    \draw[|-|,xshift=5.5cm] (7,0) -- (7, 1.5) 
    node[right, black] {\small $\geq \ell$} -- (7,3);
     %% Horizontal size
     \draw[|-|] (0,3.8) -- (5.5, 3.8) 
    node[above, black] {\small $\leq N-\ell$} -- (11,3.8);
    \end{scope}
\end{tikzpicture}
\]
In this decomposition, the regions $A$, $B$, and $C$ are arranged sequentially along the splitting direction, the regions $\mathcal{R}_1 := AB$ and $\mathcal{R}_2 := BC$ are proper rectangles with intersection $B = \mathcal{R}_1 \cap \mathcal{R}_2$ being a proper rectangle that contains at least $\ell$ plaquettes along both coordinate directions. Then, it holds that
\[  \Delta_{\sigma}(A:C|D) \leq \delta(\ell)\,. \]
\end{enumerate}
\end{Assumption}

Under the above conditions, we can prove the following.

\begin{Theo}\label{Theo:gap2}
Let  $\mathbf{H}_{\Lambda_N}$ be canonical purified Hamiltonian of $\sigma$. For fixed integers $N \geq \mu \geq 2^{8}$, let us denote  $\delta_{k}:= \delta(\lfloor \tfrac{\sqrt{\mu}}{8}(9/8)^{k/2}\rfloor\,)$  for each $k \geq 0$. Then,
\begin{equation}\label{equa:TheoGap2}
\gap(\mathbf{H}_{\Lambda_N}) \, \geq \,e^{-11}(1-\delta_{0})^2 \left(\, \prod_{k=1}^{\infty} (1-\delta_{k})\right) \, \inf\left\{ \frac{1}{\eta_{\cR}(\sigma)^{4}} \colon \cR \in \mathcal{F}_{\mu} \right\}\,. 
\end{equation}
\end{Theo}

We can make the same consideration for the infinite product appearing in the above expression that we made in Remark~\ref{Rema:infiniteProductConvergence}.

\begin{proof}
We can argue as in \cite[Theorem 2.5]{Lucia2023}, but dealing with admissible rectangles. We will use the floor $\lfloor \cdot \rfloor$ and ceiling $\lceil \cdot \rceil$ functions. First, let us decompose the torus $\Lambda_{N}$ as the union of two admissible cylinders $\Lambda_{N} = L \cup R$ taking
\[ L = \left[0,3\lfloor N/4 \rfloor \right] \times \mathbb{S}_{N} \mbox{ and } R = \left[ 2\lfloor N/4 \rfloor , \lfloor N/4 \rfloor \right] \times \mathbb{S}_{N}\,. \]
Observe that this decomposition leads to a partition $\Lambda_{N} = AB_{1}CB_{2}$ as in Assumption~\ref{Assumtion2Dtorus}.($i$), where $B_{1} = [0, \lfloor N/4 \rfloor] \times \mathbb{S}_{N}$ and $B_{2} = [2\lfloor N/4\rfloor, 3\lfloor N/4\rfloor] \times \mathbb{S}_{N}$ so that $L=B_{1}AB_{2}$ and $R=B_{2}CB_{1}$. Note that both $B_{1}$ and $B_{2}$ contain at least $\ell=\lfloor N/4 \rfloor$ plaquettes along both coordinate directions. Thus, by Assumption~\ref{Assumtion2Dtorus}.($i$),  $\| \Pi_{L} \Pi_{R} - \Pi_{L \cup R} \| = \Delta_{\sigma}(A:C|\emptyset) \leq \delta(\lfloor N/4\rfloor) \leq \delta_{0}$. This implies, by Lemma~\ref{Lemm:DivideAndConquer}
\begin{equation}\label{equa:gap2Daux1} 
\operatorname{gap}(\mathbf{H}_{\Lambda_N}) \geq \frac{1-\delta_{0}}{2} \min\{ \operatorname{gap}(\mathbf{H}_{L}), \operatorname{gap}(\mathbf{H}_{R}) \}\,. 
\end{equation}
Next, we write $L$ as the union of two proper rectangles $L_{0}$ and $R_{0}$  defined as
\[ L_{0} = \left[0,3\lfloor N/4\rfloor \right]  \times \left[ 0,3 \lfloor N/4 \rfloor \right]  \quad , \quad R_{0} = \left[0, 3\lfloor N/4 \rfloor \right]  \times \left[2\lfloor N/4 \rfloor , \lfloor N/4 \rfloor \right]\,. \]
Similarly, we write $R$ as a union of two proper rectangles $L_{1}$ and $R_{1}$ given by
\[ L_{1} = \left[ 2\lfloor N/4 \rfloor , \lfloor N/4 \rfloor \right]  \times \left[ 0 , 3\lfloor N/4\rfloor \right]  \quad , \quad R_{1} = \left[ 2\lfloor N/4\rfloor , \lfloor N/4 \rfloor \right]  \times \left[ 2\lfloor N/4 \rfloor , \lfloor N/4\rfloor \right]\,. \]
The dimensions $d_{1}$ and $d_{2}$ of each of these four rectangles satisfy $3\lfloor N/4 \rfloor \leq d_{1}, d_{2} \leq N-\lfloor N/4 \rfloor$, so they are admissible rectangles belonging to $\mathcal{F}_{d_{max}}$ where where $d_{max}=N-\lfloor N/4\rfloor$. 

Moreover, each of these two decompositions will lead to a partition as in Assumption~\ref{Assumtion2Dtorus}.($ii$). For instance, in the case of $L=L_{0} \cup R_{0}$ we have the partition $\Lambda_{N} = AB_{1}CB_{2}D$ where $B_{1} = \left[ 0, 3\lfloor N/4\rfloor \right] \times \left[ 0, \lfloor N/4\rfloor \right]$ and $B_{2} = \left[ 0, 3\lfloor N/4\rfloor \right] \times \left[ 2\lfloor N/4 \rfloor , 3\lfloor N/4\rfloor \right]$ so that $L_{0}=B_{2}AB_{1}$ and $R_{0} = B_{1}CB_{2}$. Note that both $B_{1}$ and $B_{2}$ contain at least $\ell = \lfloor N/4 \rfloor$ plaquettes along both coordinate directions. Thus, by Assumption~\ref{Assumtion2Dtorus}.($ii$), $\| \Pi_{L_0} \Pi_{R_0} - \Pi_{L} \| = \Delta_{\sigma}(A:C|\emptyset) \leq \delta(\lfloor N/4\rfloor) \leq \delta_{0}$, and similarly for $L_{1}$, $R_{1}$ and $R$. This yields, by Lemma~\ref{Lemm:DivideAndConquer} 
\begin{equation}\label{equa:gap2Daux2} 
\min\{ \operatorname{gap}(\mathbf{H}_{L}), \operatorname{gap}(\mathbf{H}_{R}) \} \geq \frac{1-\delta_{0}}{2}  \min_{X \in \mathcal{F}_{d_{max}}}\operatorname{gap}(\mathbf{H}_{X}). 
\end{equation}
Finally, to obtain a lower bound on the last term, we consider the same argument from \cite[Theorem 2.5]{Lucia2023} but working with admissible rectangles. 

For each $k \geq 0$, let $\mathcal{G}_{k}$ be the subset of $\mathcal{F}_{d_{max}}$ of all admissible proper rectangles of  with dimensions $d_{1}$ and $d_{2}$ satisfying $d_{1} \cdot d_{2} \leq \mu (3/2)^{k}$, and let $s_{k} = \lfloor (4/3)^{k/2}\rfloor$. We claim that for every $k \geq 1$
\begin{equation}\label{equa:gap2Daux3} 
\min_{X \in \mathcal{G}_{k}} \operatorname{gap}(\mathbf{H}_{X}) \geq \frac{1-\delta_{k}}{1+\frac{1}{s_{k}}} \min_{X' \in \mathcal{G}_{k-1}} \operatorname{gap}(\mathbf{H}_{X'})\,. 
\end{equation}
To check the claim, let $X \in \mathcal{G}_{k} \setminus \mathcal{G}_{k-1}$ for $k \geq 1$, that we can identify with the rectangle $[0,d_{1}] \times [0, d_{2}]$. Let assume without loss of generality that $d_{2} \leq d_{1} (\leq 6d_{2})$, so that 
\[  d_{2}^{2} \leq d_{1}d_{2} \leq \mu (3/2)^{k} \quad , \quad  d_{1}^{2} \geq d_{1}d_{2} > \mu(3/2)^{k-1}. \]
Then, consider 
\[ \ell_{k}:=\lfloor \frac{d_{1}}{6s_{k}} \rfloor  \geq \lfloor \frac{\sqrt{\mu} (3/2)^{\frac{k-1}{2}}}{6(4/3)^{\frac{k}{2}}} \rfloor = \lfloor \frac{\sqrt{\mu} (9/8)^{\frac{k}{2}}}{6(3/2)^{\frac{1}{2}}} \rfloor \geq \lfloor \frac{\sqrt{\mu}}{8} (9/8)^{\frac{k}{2}}\rfloor \geq 2\,,  \]
where the last inequality follows from $\mu \geq 2^8$, and define for each $j=0,1,\ldots, s_{k}-1$
\[ L_{j} = [0, \lceil d_{1}/3 \rceil + (2j+1) \ell_{k}] \times [0, d_{2}] \quad , \quad R_{j} = [\lceil d_{1}/3 \rceil + 2j \ell_{k}, d_{1}] \times [0, d_{2}]\,. \]
These are admissible proper rectangles belonging to $\mathcal{G}_{k-1}$. Indeed, the horizontal dimension $d_{1}'$ of $L_{j}$ satisfies $d_{1}' \leq d_{1} \leq 6d_{2}$ and its vertical dimension $d_{2}$ satisfies $d_{2} \leq d_{1} \leq 6 \lceil d_{1}/3 \rceil \leq 6d_{1}'$. On the other hand, the horizontal dimension $d_{1}''$ of $R_{j}$ satisfies $d_{1}'' \leq d_{1} \leq 6d_{2}$ and 
\[ d_{1}'' = d_{1} - \lceil d_{1}/3 \rceil - 2 j \ell_{k} \geq  d_{1} - d_{1}/3 -1 - d_{1}/3 = \frac{d_{1}}{3} - 1 \geq \frac{d_{1}}{6} \geq \frac{d_{2}}{6} \]
since $d_{1} \geq \sqrt{\mu} \geq 16$. The intersection
\[ L_{j} \cap R_{j} = [ \lceil d_{1}/3\rceil + 2 j \ell_{k} , \lceil d_{1}/3\rceil + (2j+1) \ell_{k} ] \times [0,d_{2}] \]
contains at least $\ell_{k} \geq \lfloor \frac{\sqrt{\mu}}{8} (9/8)^{k/2} \rfloor$ plaquettes per row and $d_{2} \geq d_{1}/6 \geq \ell_{k}$ plaquettes per column. Each decomposition $X=L_{j} \cup R_{j}$  leads to a partition $\Lambda=A_{j}B_{j}C_{j}D_{j}$ as in Assumption~\ref{Assumtion2Dtorus}.($iii$) with $\ell = \ell_{k}$ where $L_{j}=A_{j}B_{j}$, $R_{j} = B_{j}C_{j}$. Note that the dimensions of $L_{j}$ and $R_{j}$ are smaller than $k_{max} \leq N-\ell$. Thus, using the assumption we get $\| \Pi_{L_{j}} \Pi_{R_{j}} - \Pi_{X}\| \leq \delta(\ell_{k}) \leq \delta(\lfloor \frac{\sqrt{\mu}}{8} (\sqrt{9/8})^{k} \rfloor) = \delta_{k}$ for every $j=0,1,\ldots, s_{k}-1$, and so by Lemma~\ref{Lemm:DivideAndConquer} we get the claimed inequality \eqref{equa:gap2Daux3}. As a consequence, since $\mathcal{G}_{k} = \mathcal{F}_{d_{max}}$ for large $k$ and $\mathcal{G}_{0} \subset  \mathcal{F}_{\mu}$
\begin{equation}\label{equa:gap2Daux4} 
\min_{X \in \mathcal{F}_{d_{max}}} \operatorname{gap}(\mathbf{H}_{X}) \geq \left( \prod_{k=0}^{\infty} \frac{1-\delta_{k}}{1+\frac{1}{s_{k}}} \right) \min_{X \in \mathcal{G}_{0}}\operatorname{gap}(\mathbf{H}_{X}) \geq \left( \prod_{k=0}^{\infty} \frac{1-\delta_{k}}{1+\frac{1}{s_{k}}} \right) \min_{X \in \mathcal{F}_{\mu}}\operatorname{gap}(\mathbf{H}_{X}) \,. 
\end{equation}
We can combine inequalities \eqref{equa:gap2Daux1}, \eqref{equa:gap2Daux2} and \eqref{equa:gap2Daux4} with Theorem~\ref{Theo:roughBoundSpectralGap} to obtain the result, by observing moreover that $s_k \ge (4/3)^{\frac{k-1}{2}}$ when $k\ge 14$, and therefore
\[
\prod_{k=1}^\infty \frac{1}{1+\frac{1}{s_k}} \ge \exp(-\sum_{k=1}^\infty \frac{1}{s_k}) \ge \exp( -\sum_{k=1}^{13} \frac{1}{s_k} - \sum_{k=14}^\infty \qty(\frac{3}{4})^{\frac{k-1}{2}}) \ge \exp(-9).
\]
\end{proof}

\subsection{Gap in arbitrary dimensions}\label{sec:gap-higher-dim}
The result obtained for 1D chains and 2D tori in the previous sections can be generalized to higher dimensions. When working in this setting, there are different valid choices that can be made regarding which ``shapes'' of finite volumes one is interested in looking at. Here we present an exemplary result, constructed out of families of rectangles, but variations of this results are possible and might be needed depending on the specific model considered.

Let $\Lambda_N = \mathbb{Z}_{N}^D$ the $D$-dimensional tours. The regions we will consider are then products of intervals and $\mathbbm{Z}_N$: a \emph{hyper-cylinder} with $D-i$ open boundary conditions and $i$ periodic boundary condition is a set which, up to permutation of the coordinates, is equal to
\[ [a_{1}, b_1] \times \cdots \times [a_{D-i}, b_{D-i}] \times \mathbbm{Z}_N^i,\]
for $i=0,\dots, D$ and $(a_j,b_j)_{j=1,\dots,D-i}$ a set of $D-i$ pairs of elements of $\mathbbm{Z}_N$. Note that when $i=D$, a hyper-cylinder is simply $\mathbbm{Z}_N^D$ while when $i=0$, it is a product of intervals.
 
Let $\ell_k = (3/2)^{k/D}$, and for $k \le k_{\max} = \lceil D \log(N/2)/\log(3/2) \rceil $ let
\begin{equation}
    R(k) = [0, \ell_{k+1}] \times \ldots \times [0, \ell_{k+D}].
\end{equation}
We then define $\mathcal{F}_{k}$ as the set of products of intervals (i.e., hyper-cylinders with $D$ open boundary conditions) which are contained in $R(k)$ up to translations and permutations of the coordinates, we we write $\gap(\mathcal{F}_k)$ for the $\inf_{R \in \mathcal{F}_k} \gap(\mathbf{H}_R)$.
We also denote
\[
I_0 = \qty[0, \frac{3N}{4}] \qc I_1 = \qty[\frac{N}{2}, \frac{N}{4}]
\]
and
\[
C_{i_1,\dots, i_r} = I_{i_1} \times \cdots \times I_{i_r} \qc i_1,\dots, i_r = 0,1.
\]
\begin{Prop}
    With the notation defined above, let
    \begin{align}
        L_{i_1,\dots, i_{r-1}} &= C_{i_1,\dots, i_{r-1}} \times \qty[\frac{N}{4}, \frac{N}{2}] \times \mathbb{Z}_N^{D-r} ,\\
        R_{i_1,\dots, i_{r-1}} &=  C_{i_1,\dots, i_{r-1}} \times \qty[\frac{3N}{4}, 0] \times \mathbb{Z}_N^{D-r},
    \end{align}
    for each $r=1,\dots, s$ and each choice of $i_1,\dots, i_{r-1} = 0,1$. Denote 
    \begin{equation}
        \delta^{(r)} = \min_{i_1,\dots, i_{r-1} = 0,1} \Delta_\sigma\left( 
        L_{i_1,\dots, i_{r-1}} \colon
        R_{i_1,\dots, i_{r-1}} \,\middle|\,  C_{i_1,\dots, i_{r-1}} \times \mathbb{Z}_N^{D-r+1} \right).
    \end{equation}
    Then we have that
    \begin{equation}
        \gap(\mathbf{H}_{\Lambda_N}) \ge 2^{-D} \prod_{r=1}^D (1-\delta^{(r)})\, \cdot \, \gap (\mathcal{F}_{k_{\max}}).
    \end{equation}
\end{Prop}
\begin{proof}
It is easy to see that the sets $L_{i_1,\dots, i_{r-1}}$ and $R_{i_1,\dots, i_{r-1}}$ 
satisfy
\begin{align*}
L_{i_1,\dots, i_{r-1}} \cup R_{i_1,\dots, i_{r-1}} &= C_{i_1,\dots, i_{r-1}} \times \mathbb{Z}_N^{D-r+1} , \\
L_{i_1,\dots, i_{r-1}} \setminus R_{i_1,\dots, i_{r-1}} &= C_{i_1,\dots, i_{r-1}} \times \qty[\frac{N}{4}, \frac{N}{2}] \times \mathbb{Z}_N^{D-r} ,\\
R_{i_1,\dots, i_{r-1}} \setminus L_{i_1,\dots, i_{r-1}} &= C_{i_1,\dots, i_{r-1}} \times \qty[\frac{3N}{4}, 0] \times \mathbb{Z}_N^{D-r}.
\end{align*}
From the assumption and Theorem~\ref{Theo:martingaleVScorrelations}, it follows that
\[
\norm{P_{L_{i_1,\dots, i_{r-1}}} P_{R_{i_1,\dots, i_{r-1}}} - P_{C_{i_1,\dots, i_{r-1}} \times \mathbb{Z}_N^{D-r+1}}} \le \delta^{(r)} \qc \forall i_1,\dots, i_{r-1} = 0,1.
\]
The result then follows by using Lemma~\ref{Lemm:DivideAndConquer} iteratively, once for each dimension, with the choice $s=1$ in each case.
\end{proof}

In a sense, we have now reduced the problem from a periodic boundary condition one to a open boundary condition one. The crucial point about the family of sets $\mathcal{F}_k$ is that we can always decompose one element $R \in \mathcal{F}_{k+1}$ as the union of two elements of $\mathcal{F}_{k-1}$ with a ``large enough'' overlap, and moreover we can do this in $s_k = \lfloor \sqrt{\ell_k} \rfloor$ independent ways. This is summarized in the following proposition.
\begin{Prop}[{\cite[Proposition 1]{KL_2018}}]
Let $k \ge k_{\min} = \lceil D \log(64)/\log(3/2) \rceil$, and let $X \in \mathcal{F}_{k+1}\setminus \mathcal{F}_k$. We can then find $s_k = \lfloor \sqrt{\ell_k} \rfloor$ pairs $(L_i,R_i)_{i=1}^{s_k}$ of elements of $\mathcal{F}_k$ with the following properties:
\begin{enumerate}
    \item $X = L_i \cup R_i$ for every $i$;
    \item $\dist(L_i \setminus R_i, R_i \setminus L_i) \ge \frac{\ell_k}{8s_k}-2$;
    \item $L_i \cup R_i \cup L_j \cup R_j = \emptyset$ for $i\neq j$.
\end{enumerate}
\end{Prop}
Note that the condition $k\ge k_{\min}$ guarantees that $s_k \le \frac{1}{8}\ell_k$, an assumption required for the proof of \cite[Proposition 1]{KL_2018}.

We can therefore iteratively apply Lemma~\ref{Lemm:DivideAndConquer} to obtain an estimate on the gap of $\mathbf{H}_{\Lambda_N}$ which only depends on finite regions.
\begin{Prop}
    Let 
    \begin{equation}
        \delta_k = \inf_{(L, R)} \Delta_\sigma\qty(L\setminus R : R \setminus L \, | \, (L\cup R)^c ),
    \end{equation}
    where the infimum is taken over all pairs $L,R$ of elements in $\mathcal{F}_k$ satisfying
    \[
    \dist(L \setminus R, R\setminus L) \ge \frac{\ell_k}{8s_k}-2.
    \]
    Then we have that
    \begin{equation}
        \gap(\mathcal{F}_{k_{\max}}) \ge \prod_{r=k_{{\min}}}^\infty \qty(\frac{1-\delta_k}{1+\frac{1}{s_k}}) \,\cdot\, \min_{R \in \mathcal{F}_{k_{\min}}} \frac{1}{\eta_R(\sigma)^2}.
    \end{equation}
\end{Prop}

\subsection{Locality of the canonical purified Hamiltonian}
\label{sec:gibbs-states}

In the last section, we have seen how assuming a decay of $\Delta_\sigma(A:C|D)$, and in particular assuming the decay of the quantities $\delta_k$, leads to a lower bound on the spectral gap of $\mathbf{H}$.
We can obtain a weak converse of this fact, under the additional assumption that $\textbf{H}$ is itself local, i.e., that it can be decomposed as a sum of local terms with finite support. If this is the case, then we can apply known result to prove the following: if all the local gaps of $\mathbf{H}$ are bounded away from zero, then $\Delta_\sigma$ has to decay at a certain rate (note that it is not sufficient to simply assume a bound on the gap of $\mathbf{H}$). While this is does not quite show that the decay of $\Delta_\sigma$ is a necessary condition for $\mathbf{H}$ to be gapped, it does indicate that there is a strong connection between the two quantities.

\begin{Prop}
    Suppose that $\mathbf{H}_\Lambda$ has finite range $r>0$. With the notation of Section~\ref{sec:gap-higher-dim}:
    \begin{itemize}
        \item If $\gap(\mathcal{F}_k) \ge \lambda_k > 0$, then
        \begin{equation}
            \delta_k \le C_1 \exp(-\frac{C_2}{r} \sqrt{\lambda_k \ell_k}).
        \end{equation}
        for some positive constants $C_1$ and $C_2$ independent of $k$.
        \item If $\gap( \mathbf{H}_{C_{i_{1},\dots, i_{u}}}) \ge \lambda^{(u)}>0$ for every $i_1,\dots, i_u = 0,1$, then
        \begin{equation}
            \delta^{(u)} \le C_1 \exp(-\frac{C_2}{r} \sqrt{\lambda^{(u)} N}).
        \end{equation}
        for some positive constants $C_1$ and $C_2$ independent of $u$.
    \end{itemize}
\end{Prop}
\begin{proof}
    This is a straightforward consequence of the so-called Detectability Lemma \cite{Aharonov2009,Aharonov2011,Gosset2016}. In particular, we can directly apply \cite[Lemma 3.1]{2409.09685} to obtain all the bounds.
\end{proof}

A particularly important case for which we can show that $\mathbf{H}$ is local is when
$\sigma$ is the Gibbs state of a local Hamiltonian. Let $\Phi$ be a local, commuting interaction, in the sense that $[\Phi_X,\Phi_Y]=0$ for every pair of regions $X, Y$, and assume that $\sigma = e^{-\beta H_{\Lambda}}/\Tr(e^{- \beta H_{\Lambda}})$ is the Gibbs state associated to the Hamiltonian $H_{\Lambda} = \sum_{X \subset \Lambda} \Phi_{X}$. 
Then the canonical Hamiltonian $\mathbf{H}$ has finite range, as the next proposition shows.

\begin{Prop}
    If $\Phi$ is commuting and has finite range $r >0$, then $\mathbf{H}$ has finite range $2r$.
\end{Prop}
\begin{proof}
    This is a simple consequence of the formula for the interactions terms of $\mathbf{H}$: since
    \[
    \sigma_{\Lambda\setminus\{x\}} = Z_\beta^{-1} \Tr_x(e^{-\beta \sum_{X \subset \Lambda} \Phi_X}) = Z_\beta^{-1} e^{-\beta \sum_{X\not\ni x}\Phi_X} \Tr_x(e^{-\beta \sum_{X \ni x} \Phi_X})\,,
    \]
    we have that
    \[
    \Pi_x(Q) = \Tr_x(Q\sigma^{1/2}) \sigma_{\Lambda\setminus\{x\}}^{-1} \sigma^{1/2} = \Tr_x(Q e^{-\frac{\beta}{2} \sum_{X \ni x} \Phi_X}) \qty( \Tr_x (e^{-\beta \sum_{X \ni x} \Phi_X}) )^{-1} e^{-\frac{\beta}{2} \sum_{X \ni x} \Phi_X}
    \]
    which only acts in a neighborhood of radius $r$ around the point $x$.
\end{proof}

Under these assumptions, we can moreover obtain estimates on $\Delta_\sigma(A:C|D)$ from a particular condition on the marginals of $\sigma$, which we will use in Section~\ref{sec:quantum-double} in order to study 2D quantum double models.

\begin{Theo}\label{Theo:SufficientCondition1}
Let $\Phi$ be a commuting interaction on $\Lambda$ and let $\sigma = e^{-\beta H_\Lambda}/\Tr(e^{-\beta H_\Lambda})$ for some $\beta>0$. Assume we partition $\Lambda$ into four disjoint subsets $\Lambda=ABCD$ satisfying the following properties:
\begin{enumerate}
\item[(i)] There are no local interactions between $A$ and $C$, namely $\Phi_{X} = 0$ whenever $X \cap A \neq \emptyset$ and $X \cap C \neq \emptyset$.
\item[(ii)] For each $\cR \in \{ AB,B,BC,ABC\}$ there exist an operator $Q_{\partial \cR}$ supported in $\partial \cR$ and a real constant $\kappa_{\cR}>0$ such that
\[ \Tr_{\cR}(e^{-\beta H_{\cR}^{\partial}}) = \kappa_{\cR}(\mathbbm{1} + Q_{\partial |cR})\,. \]
\item[(iii)] $\kappa_{AB} \kappa_{BC} = \kappa_{B} \kappa_{ABC}$.
\item[(iv)] There is $\varepsilon \in (0,1)$ such that $\| Q_{\partial \cR}\|_{\infty} \leq \varepsilon \, e^{-2 \beta |\partial \cR| \, \| \Phi\|}$ for every $\cR \in \{ AB,BC,ABC,B\}$.
\end{enumerate}
Then, it holds that
\begin{equation}
\Delta_\sigma(A:C|D) \leq \qty(\frac{1+\varepsilon}{1-\varepsilon})^2 - 1 = \frac{4\varepsilon}{(1-\varepsilon)^2} \,.
\end{equation}
\end{Theo}

\begin{proof}
We will prove the result via the upper bound from Proposition~\ref{prop:estimatesMartinagleCondition} for one of the ordering of the product, since the proof will be analogous for the other one. Let us start by considering
\begin{equation}\label{equa:SufficientCondition1Aux1}
\sigma_{AD} \sigma_{D}^{-1} \sigma_{DC} \sigma_{ACD}^{-1}
= \Tr_{BC}(e^{-\beta H}) \Tr_{ABC}(e^{-\beta H})^{-1}  \Tr_{AB}(e^{-\beta H})  \Tr_{B}(e^{-\beta H})^{-1}\,. 
\end{equation}
Then, let us rewrite each of the four factors in the previous expression, first by extracting the factors of the exponential corresponding to interactions supported in the complement of the region that is being tracing out, and then using the condition ($ii$):
\begin{align*}
\Tr_{BC}(e^{-\beta H})  & = \Tr_{BC}(e^{-\beta H^{\partial}_{BC}}) e^{-\beta H_{AD}} = \kappa_{BC} (\mathbbm{1} + Q_{\partial BC}) e^{-\beta H_{AD}} \,,\\[2mm]
\Tr_{ABC}(e^{-\beta H})  & = \Tr_{ABC}(e^{-\beta H^{\partial}_{ABC}}) e^{-\beta H_{D}} = \kappa_{ABC} (\mathbbm{1} + Q_{\partial ABC}) e^{-\beta H_{D}}\,,\\[2mm]
\Tr_{AB^{c}}(e^{-\beta H})  & = \Tr_{AB}(e^{-\beta H^{\partial}_{AB}}) e^{-\beta H_{CD}} = \kappa_{AB} (\mathbbm{1} + Q_{\partial AB} )e^{-\beta H_{CD}}\,,\\[2mm]
\Tr_{B}(e^{-\beta H})  & = \Tr_{B}(e^{-\beta H^{\partial}_{B}}) e^{-\beta H_{ACD}} = \kappa_{B} (\mathbbm{1} + Q_{\partial B}) e^{-\beta H_{ACD}}\,.
\end{align*}
Inserting these expressions in \eqref{equa:SufficientCondition1Aux1}, we obtain
\begin{multline}\label{equa:SufficientCondition1Aux2}
\sigma_{AD} \sigma_{D}^{-1} \sigma_{DC} \sigma_{ACD}^{-1} =\\ (\mathbbm{1} + Q_{\partial BC}) e^{-\beta H_{AD}+\beta H_{D}}(\mathbbm{1} + Q_{\partial ABC})^{-1} (\mathbbm{1} + Q_{\partial AB} )e^{-\beta H_{CD}+\beta H_{ACD}} (\mathbbm{1} + Q_{\partial B})^{-1} \,. 
\end{multline}
Observe that the inverse of $\mathbbm{1} + Q_{\partial \cR}$ can be actually expanded $(\mathbbm{1} + Q_{\partial \cR})^{-1} = \sum_{m=0}^{\infty} (-1)^{m} Q_{\partial \cR}^{m}$, where this series is actually absolutely convergent whenever $\norm{Q_{\partial \cR}}<1$. Indeed, if we set
\[ \widehat{Q}_{\partial \cR} = (\mathbbm{1} + Q_{\partial \cR})^{-1} - \mathbbm{1}\,, \]
then $\widehat{Q}_{\partial \cR}$ is supported in $\partial \cR$ and satisfies, by condition ($iv$), the bound
\[ \|\widehat{Q}_{\partial \cR} \|_{\infty} \leq \sum_{m=1}^{\infty} \| Q_{\partial \cR}\|_{\infty}^{m} = \frac{ \norm{Q_{\partial \cR}}_{\infty}}{1-\norm{Q_{\partial \cR}}_{\infty}} \leq \frac{\varepsilon}{1-\varepsilon} e^{-2 \beta |\partial \cR| \| \Phi\|}\,. \]
Thus, we can rewrite \eqref{equa:SufficientCondition1Aux2} replacing $(\mathbbm{1}+Q_{\partial \cR})^{-1} = \mathbbm{1} + \widehat{Q}_{\partial \cR}$, 
\begin{multline}\label{equa:SufficientCondition1Aux3}
\sigma_{AD} \sigma_{D}^{-1} \sigma_{DC} \sigma_{ACD}^{-1}  =\\ (\mathbbm{1} + Q_{\partial BC}) e^{\beta H_{AD}-\beta H_{D}}(\mathbbm{1} + \widehat{Q}_{\partial ABC}) (\mathbbm{1} + Q_{\partial AB} )e^{\beta H_{CD}-\beta H_{ACD}} (\mathbbm{1} + \widehat{Q}_{\partial B}) \,. 
\end{multline}
Using condition ($i$), we deduce that $H_{CD}-H_{ACD} = H_{D} - H_{AD}$. For any operator $\cO$, let us define
\[ 
\cO'  := e^{\beta H_{AD} - \beta H_{D}} \cO e^{-\beta H_{AD} + \beta H_{D}}.
\]
Then, we can write the above expression as
\begin{equation}\label{equa:SufficientCondition1Aux4}
\sigma_{AD} \sigma_{D}^{-1} \sigma_{DC} \sigma_{ACD}^{-1}  = (\unit + Q_{\partial BC}) (\unit+\widehat{Q}_{\partial ABC}') (\unit+ Q_{\partial AB}') (\unit+\widehat{Q}_{\partial B})\,.
\end{equation}
We need estimates for the norm of $\widehat{Q}_{\partial ABC}'$ and $ Q_{\partial AB}'$. Note that both cases are obtained by evolving $\cO$ under imaginary time evolution. But we can easily estimate it using the following fact: if we evolve (imaginary time evolution) an observable $O_{X}$ supported on a region $X$ according to some commuting Hamiltonian $\widetilde{H}$ made of local (and commuting) interactions $\widetilde{\Phi}$, then we can rewrite
\[ e^{\beta \widetilde{H}} O_{X}e^{-\beta \widetilde{H}} = e^{\beta \widetilde{H}_{X}^{\partial}} O_{X}e^{-\beta\widetilde{H}_{X}^{\partial}}\,. \]
so that
\[ \| e^{\beta \widetilde{H}} O_{X}e^{-\beta \widetilde{H}} \|_{\infty} = \| e^{\beta \widetilde{H}_{X}^{\partial}} O_{X}e^{-\beta \widetilde{H}_{X}^{\partial}} \|_{\infty} \leq \| O_{X}\|_{\infty} \, e^{2 \| \beta \widetilde{H}_{X}^{\partial}\|_{\infty}} \leq \| O_{X}\|_{\infty} e^{2 \beta |X| \| \widetilde{\Phi}\|}\,. \]
This yields
\[ 
\| \widehat{Q}_{\partial ABC}'\|_{\infty}  \leq \| \widehat{Q}_{\partial ABC}\|_{\infty} e^{ 2\beta |\partial ABC| \,\| \Phi\|} \leq \frac{\varepsilon}{1-\varepsilon} \quad , \quad
\| Q_{\partial AB}'\|_{\infty}  \leq \| Q_{\partial AB}\|_{\infty} e^{ 2\beta |\partial AB| \,\| \Phi\|} \leq \varepsilon .\]
Finally, we have to use the observation that the polynomial
\[
(\unit + Q_{\partial BC}) (\unit+\widehat{Q}_{\partial ABC}') (\unit+ Q_{\partial AB}') (\unit+Q_{\partial B}) - \unit,
\]
has all positive coefficients, and therefore we can write
\begin{multline*}
\norm{(\unit + Q_{\partial BC}) (\unit+\widehat{Q}_{\partial ABC}') (\unit+ Q_{\partial AB}') (\unit+Q_{\partial B}) - \unit}_{\infty} \\ \le (1 + \norm*{Q_{\partial BC}}_{\infty}) (1+\norm*{\widehat{Q}_{\partial ABC}'}_{\infty}) (1+ \norm*{Q_{\partial AB}'}_{\infty}) (1+\norm*{\widehat{Q}_{\partial B}}_{\infty}) - 1\\
\le (1+\varepsilon) (1+\tfrac{\varepsilon}{1-\varepsilon})(1+\varepsilon) (1+\tfrac{\varepsilon}{1-\varepsilon})-1 \le
\qty(\frac{1+\varepsilon}{1-\varepsilon})^2 - 1.
\end{multline*}
This finishes the proof.
\end{proof}

\section{1D spin chains}
\label{sec:1D-chains}

In this section, we will first present as an example the application of the results of Section~\ref{sec:Estimating_the_gap} to the case of the 1D Ising model. We will then show that it is possible to obtain spectral gap estimates which are system-size independent for the purified canonical Hamiltonian associated to the Gibbs state of any 1D, local Hamiltonian (including non-commuting one), at any positive temperature.

\subsection{1D Ising Model}

\begin{figure}[hbt]
%Primer dibujo
    \centering
\begin{tikzpicture}
% Definimos el número de nodos
    \def\nodes{16}
    % Definimos el radio del círculo
    \def\radius{1}
    % Definimos el ángulo para la división en regiones
    \def\angreg{34}

    \begin{scope}
    % Dibujamos la circunferencia
    \draw[thick] (0,0) circle(\radius);

    % Dibujamos los nodos del lattice circular
    \foreach \i in {1,...,\nodes} {
        % Calculamos la posición de los nodos en el círculo
        \pgfmathsetmacro{\angle}{360/\nodes * (\i - 1)}
        \node[draw, circle, fill=blue!20, inner sep=2pt] at (\angle:\radius) {};
    }

    % Marcamos las cuatro regiones cortando perpendicularmente la circunferencia con segmentos finos

     \draw[thin] (\angreg:{\radius-0.2}) -- (\angreg:{\radius+0.5});

    \draw[thin] ({180-\angreg}:{\radius-0.2}) -- ({180-\angreg}:{\radius+0.5});

    %Determinamos las regiones

     \draw[ stealth-stealth] (\angreg+3:{\radius+0.3}) arc(\angreg+3:{180-\angreg-3}:{\radius+0.3});

     \node at (0, {\radius+0.5}) {$I$};
   
    \end{scope}

%Segundo dibujo

\begin{scope}[xshift=3.5cm]
    % Dibujamos la circunferencia
    \draw[thick] (0,0) circle(\radius);

    % Dibujamos los nodos del lattice circular
    \foreach \i in {1,...,\nodes} {
        % Calculamos la posición de los nodos en el círculo
        \pgfmathsetmacro{\angle}{360/\nodes * (\i - 1)}
        \node[draw, circle, fill=blue!20, inner sep=2pt] at (\angle:\radius) {};
    }

    % Marcamos las cuatro regiones cortando perpendicularmente la circunferencia con segmentos finos

     \draw[thin] (\angreg:{\radius-0.2}) -- (\angreg:{\radius+0.5});

    \draw[thin] ({180-\angreg}:{\radius-0.2}) -- ({180-\angreg}:{\radius+0.5});

    \node at (180-\angreg-10:{\radius+0.4}) {$1$};

    \node at (70:{\radius+0.3}) {$\cdot$};
    \node at (80:{\radius+0.3}) {$\cdot$};
    \node at (90:{\radius+0.3}) {$\cdot$};
    \node at (100:{\radius+0.3}) {$\cdot$};
    \node at (110:{\radius+0.3}) {$\cdot$};
    
    \node at (\angreg+10:{\radius+0.4}) {$m$};

    \end{scope}

%%% Tercer dibujo

\begin{scope}[xshift=8.5cm]
    % Dibujamos la circunferencia
    \draw[thick] (0,0) circle(\radius);

    % Dibujamos los nodos del lattice circular
    \foreach \i in {1,...,\nodes} {
        % Calculamos la posición de los nodos en el círculo
        \pgfmathsetmacro{\angle}{360/\nodes * (\i - 1)}
        \node[draw, circle, fill=blue!20, inner sep=2pt] at (\angle:\radius) {};
    }

    % Marcamos las cuatro regiones cortando perpendicularmente la circunferencia con segmentos finos

    \draw[thin] (-\angreg:{\radius-0.2}) -- (-\angreg:{\radius+0.5});

     \draw[thin] (\angreg:{\radius-0.2}) -- (\angreg:{\radius+0.5});

    \draw[thin] ({180-\angreg}:{\radius-0.2}) -- ({180-\angreg}:{\radius+0.5});

     \draw[thin] ({180+\angreg}:{\radius-0.2}) -- ({180+\angreg}:{\radius+0.5});

    %Determinamos las regiones

    \draw[ stealth-stealth] (-\angreg+3:{\radius+0.3}) arc(-\angreg+3:\angreg-3:{\radius+0.3});

    \node at ({\radius+0.5}, 0) {$C$};

     \draw[ stealth-stealth] (\angreg+3:{\radius+0.3}) arc(\angreg+3:{180-\angreg-3}:{\radius+0.3});

     \node at (0, {\radius+0.5}) {$I_{1}$};
     
    \draw[ stealth-stealth] ({180-\angreg+3}:{\radius+0.3}) arc({180-\angreg+3}:{180+\angreg-3}:{\radius+0.3});

    \node at ({-\radius-0.5}, 0) {$A$};

    \draw[ stealth-stealth] ({180+\angreg+3}:{\radius+0.3}) arc({180+\angreg+3}:{360-\angreg-3}:{\radius+0.3});

    \node at (0,{-\radius-0.5}) {$I_{2}$};
    \end{scope}

%%%%%

    \begin{scope}[xshift=12cm]
    % Dibujamos la circunferencia
    \draw[thick] (0,0) circle(\radius);

    % Dibujamos los nodos del lattice circular
    \foreach \i in {1,...,\nodes} {
        % Calculamos la posición de los nodos en el círculo
        \pgfmathsetmacro{\angle}{360/\nodes * (\i - 1)}
        \node[draw, circle, fill=blue!20, inner sep=2pt] at (\angle:\radius) {};
    }

    % Marcamos las cuatro regiones cortando perpendicularmente la circunferencia con segmentos finos

    \draw[thin] (-\angreg:{\radius-0.2}) -- (-\angreg:{\radius+0.5});

     \draw[thin] (\angreg:{\radius-0.2}) -- (\angreg:{\radius+0.5});

    \draw[thin] ({180-\angreg}:{\radius-0.2}) -- ({180-\angreg}:{\radius+0.5});

     \draw[thin] ({180+\angreg}:{\radius-0.2}) -- ({180+\angreg}:{\radius+0.5});

    \node at (180+\angreg-10:{\radius+0.3}) {$l$};

    \node at (180-\angreg+10:{\radius+0.3}) {$k$};

    \node at (\angreg-10:{\radius+0.3}) {$j$};

    \node at (-\angreg+10:{\radius+0.3}) {$i$};
    
    \end{scope}
    
    \end{tikzpicture}
    \caption{On the left, a selected subinterval $I$ of the ring, whose sites are identified with $[1,m]$. On the right, a partition of the $1D$ ring into four subintervals $\Lambda_{N}=A I_{1} C I_{2}$, where $I_{1}$ and $I_{2}$ shield $A$ from $C$. The endpoints of $A$ are marked as $l$ and $k$, and the endpoints of $C$ as $i$ and $j$.} 
    \label{fig:Ising1D}
\end{figure}
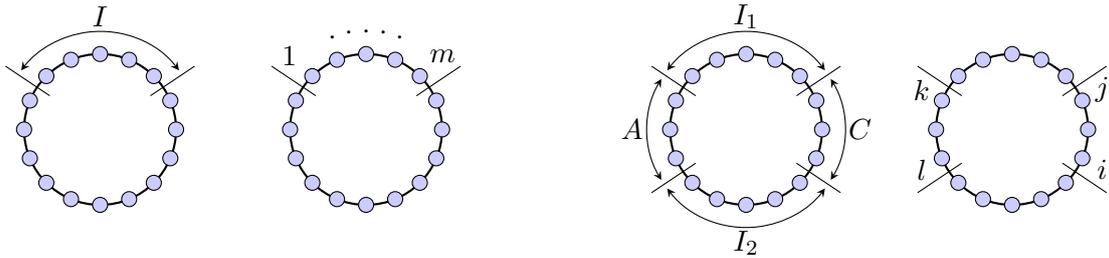

We consider $\Lambda = \mathbb{Z}_{N}$. Recall that the Hamiltonian of the Ising model is given by the nearest-neighbor interaction interaction
\[ H = \sum_{k=1}^{N} - Z_{k} Z_{k+1}\,. \]
Let $I$ be an interval of the ring $\Lambda$ that we can enumerate as $I=[1,m]$ (see Figure~\ref{fig:Ising1D}). Then, 
\[ e^{-\beta H_{\Lambda}} = \prod_{j=1}^{N} e^{\beta Z_{j} Z_{j+1}} =  \cosh(\beta)^{N} \prod_{j=1}^{N} (\mathbbm{1} + \tanh(\beta) Z_{j} Z_{j+1})\,. \]
If we apply $\Tr_{I^{c}}(\cdot)$ on the above expression, then we can take out those interactions supported in $I$. Expanding the product of the remaining terms, we get a linear combination of products of $Z_{j}'s$. The trace of each summand will be zero if for some $j \in I^{c}$ there is exactly one factor  $Z_{j}$. Therefore,
\begin{equation}\label{equa:marginalGibbsIsing} 
\Tr_{I^c}\left(e^{-\beta H_{\Lambda}}\right) =  2^{|I^c|} \cosh(\beta)^{N} (\mathbbm{1} + \tanh(\beta)^{N-m+1}Z_{1} Z_{m}) \, e^{-\beta H_{I}}\,.
\end{equation}
In particular, the partition function is \begin{equation}
\label{equa:partitionFunctionIsing}
\Tr\left(e^{-\beta H_\Lambda}\right) =2^{N} \cosh(\beta)^{N} (1+\tanh(\beta)^{N})\,.
\end{equation}

\begin{Prop}\label{Prop:correlationDecay1DIsing}
Let $\sigma = \sigma_{\beta}$ be the Gibbs state of the 1D Ising model. Consider any partition $\Lambda_{N} = A I_{1} C I_{2}$ into four nonempty intervals as in Figure~\ref{fig:Ising1D}, where $I_{1}$ and $I_{2}$ shield $A$ from $C$,  and let $\ell := \max\{|I_{1}|, |I_{2}| \} \geq 1$. Then, for any $D \in \{\emptyset, I_{1}, I_{2}\}$, it holds that
\[ \Delta_{\sigma}(A:C|D) 
\leq \delta(\ell) := 1 - \left( \frac{1-\tanh(\beta)^{\ell}}{1+\tanh(\beta)^{\ell}} \right)^{2}  = \frac{4 \tanh(\beta)^{\ell}}{(1+\tanh(\beta)^{\ell})^{2}}\,.\]
\end{Prop}

\begin{proof}
Let us discuss first the case in which $D=I_{2}$ and denote $B = I_{1}$ (the case $D=I_{1}$ will also hold by symmetry). We denote the end points of the interval $A$ by $i$ and $j$, and the endpoints of $C$ by $k$ and $l$, as in Figure~\ref{fig:Ising1D}. The endpoints of $I_{1}=D$ will then be $l+1$ and $i-1$. Using \eqref{equa:marginalGibbsIsing} 
\begin{align*}
\sigma_{AD} \sigma_{D}^{-1} \sigma_{DC} \sigma_{ADC}^{-1}
& = \Tr_{BC}\left( e^{-\beta H}\right) \cdot \Tr_{ABC}\left( e^{-\beta H}\right)^{-1} \cdot \Tr_{AB}\left( e^{-\beta H}\right) \cdot \Tr_{B}\left( e^{-\beta H}\right)^{-1}   \\[2mm]
& = \left( \mathbbm{1} + \tanh(\beta)^{|BC|+1} Z_{l+1} Z_{j} \right) \cdot \left( \mathbbm{1} + \tanh(\beta)^{|ABC|+1} Z_{i-1} Z_{l+1} \right)^{-1} \cdot \\[2mm]
& \quad \quad \cdot \left( \mathbbm{1} + \tanh(\beta)^{|AB|+1} Z_{i-1} Z_{k} \right) \cdot \left( \mathbbm{1} + \tanh(\beta)^{|B|+1} Z_{j} Z_{k} \right)^{-1}\\[2mm]
& \quad \quad \cdot e^{-\beta H_{BC} + \beta H_{ABC} - \beta H_{AB} + \beta H_{B}}\,.
\end{align*}
Note that $-H_{BC} +H_{ABC} - H_{ABC} + H_{B} = 0$, since there are no (nonzero) interactions whose support has nonempty intersection with both $A$ and $C$. We thus easily estimate the supremum norm of the previous expression by
\[ \|\sigma_{AD} \sigma_{D}^{-1} \sigma_{DC} \sigma_{ADC}^{-1} \|_{\infty} \leq \frac{(1+\tanh(\beta)^{|BC|+1}) (1+\tanh(\beta)^{|AB|+1})}{(1-\tanh(\beta)^{|B|+1}) (1-\tanh(\beta)^{|ABC|+1})} \leq \left(\frac{1+\tanh(\beta)^{\ell}}{1-\tanh(\beta)^{\ell}}\right)^2\,. \]
Next we consider the case $D=\emptyset$ and denote $B=I_{1}I_{2}$. First, note that
\begin{equation*}
\Tr_{B}(e^{-\beta H_\Lambda}) = 2^{|B|} \left( \mathbbm{1} + \tanh(\beta)^{|I_{1}|+1} Z_{j} Z_{k} \right) \left( \mathbbm{1} + \tanh(\beta)^{|I_{2}|+1} Z_{i} Z_{l} \right) e^{-\beta H_{A}-\beta H_{C}}\,.  \\
\end{equation*}
Combining the last equality with \eqref{equa:partitionFunctionIsing} and \eqref{equa:marginalGibbsIsing}, we get
\begin{align*} \sigma_{A} \sigma_{C} \sigma_{AC}^{-1} & =  \Tr_{BC}\left( e^{-\beta H_\Lambda}\right) \cdot \Tr_{ABC}\left( e^{-\beta H_\Lambda}\right)^{-1} \cdot \Tr_{AB}\left( e^{-\beta H_\Lambda}\right) \cdot \Tr_{B}\left( e^{-\beta H_\Lambda}\right)^{-1}\\[2mm]
& =\left( \mathbbm{1} + \tanh(\beta)^{|AB|+1} Z_{i} Z_{j}\right) \, \left(1+\tanh(\beta)^{N}\right)^{-1}  \,\left(\mathbbm{1} + \tanh(\beta)^{|BC|+1} Z_{l} Z_{k}\right)  \\
& \hspace{2cm} 
\left( \mathbbm{1} + \tanh(\beta)^{|I_{1}|+1} Z_{j} Z_{k}\right)^{-1} 
\left( \mathbbm{1} + \tanh(\beta)^{|I_{2}|+1} Z_{i} Z_{l}\right)^{-1}
\end{align*}
Taking the supremum norm of this operator, we conclude
\[ \| \sigma_{A} \sigma_{C} \sigma_{AC}^{-1} \|_{\infty} \leq \frac{\left( 1 + \tanh(\beta)^{|BC|+1} \right) \left( 1 + \tanh(\beta)^{|AB|+1} \right)}{ \left( 1 - \tanh(\beta)^{|I_{1}|+1} \right) \left( 1 - \tanh(\beta)^{|I_{2}|+1} \right) } \leq \left( \frac{1+\tanh(\beta)^\ell}{1-\tanh(\beta)^{\ell}}\right)^{2}\,. \]
This finishes the proof.
\end{proof}

\begin{Coro}\label{coro:ising-model-gap}
Let $N \geq 2$ and let $\mathbf{H}_{\Lambda_N}$ be the canonical purified Hamiltonian associated to the Gibbs state $\sigma_{\beta}$ of the Ising model. Then, 
\begin{equation}\label{equa:coro-ising-model-gap} 
\operatorname{gap}(\mathbf{H}_{\Lambda_N}) \geq e^{-5-76\beta} e^{-34 \beta^{2}}  \left(\prod_{k=1}^{\infty} \frac{1-e^{-\frac{1}{2}(9/8)^k}}{1+e^{-\frac{1}{2}(9/8)^k}} \right)^2 \,, \end{equation}
where the last product converges to a positive constant independent of $N$ and $\beta$.
\end{Coro}

 \begin{proof} 
 Let us fix $\mu=9$. Observe that, by Remark~\ref{Rema:etaProperties}.($iv$),
for any finite interval $I \subset \Lambda_N$ with $|I| \leq 9$, we can estimate 
 \[\eta_{I}(\sigma)^{4} \leq e^{4\beta |I| \| \Phi\|} \leq e^{8\beta |I|} \leq e^{72 \beta}.\] 
 Thus, if $N \leq \mu=9$, we can apply Theorem~\ref{Theo:roughBoundSpectralGap} to estimate $\operatorname{gap}(\mathbf{H}_{\Lambda_N}) \geq e^{-72 \beta}$, and so inequality \eqref{equa:coro-ising-model-gap} clearly holds in this case. Let us now assume that $N > \mu=9$. If we apply Theorem~\ref{Theo:gapIntervals1D} with these $N$ and $\mu$, and with the function $\delta(\ell)$ from Proposition~\ref{Prop:correlationDecay1DIsing}, we get
 \begin{equation}\label{equa:coro-ising-model-gap-Aux1} 
 \operatorname{gap}(\mathbf{H}_{\Lambda_N}) \geq e^{-5} \left( \prod_{k=0}^{\infty}(1-\delta_{k}) \right) e^{-72 \beta}\,. 
 \end{equation}
 Using that $(1-1/x)^{x} \leq e^{-1}$ for $x>1$, we can estimate
 \[ \tanh(\beta)^{\lfloor (9/8)^k \rfloor} \leq \tanh(\beta)^{\frac{1}{2}(9/8)^k} \leq (1-e^{-2\beta})^{\frac{1}{2}(9/8)^k} \leq \exp( - \tfrac{1}{2} (9/8)^{k} e^{-2\beta})\,. \]
Using this inequality, and since $(9/8)^{17} \geq e^{2}$, for all $k \geq k_{0}:= \lceil 17 \beta \rceil$, we deduce that
\[ \tanh(\beta)^{\lfloor (9/8)^{k} \rfloor} \leq \exp(-\tfrac{1}{2} (9/8)^{k-k_{0}})\,.\] 
As a consequence, since the function $y \mapsto \frac{1-y}{1+y}$ is decreasing on $0<y<1$, we can estimate 
 \begin{align*} 
 \prod_{k=0}^{\infty} \frac{1-\tanh(\beta)^{\lfloor (9/8)^k \rfloor}}{1+\tanh(\beta)^{\lfloor (9/8)^k \rfloor}} 
 & \geq
 \prod_{k=0}^{k_{0}}\frac{1-\tanh(\beta)}{1+\tanh(\beta)} \,
 \prod_{k=k_{0}+1}^{\infty} \frac{1-e^{-\frac{1}{2}(9/8)^{(k-k_{0})}}}{1+e^{-\frac{1}{2}(9/8)^{(k-k_{0})}}}\\
 & = e^{-\beta (k_{0}+1)} \prod_{k=1}^{\infty} \frac{1-e^{-\frac{1}{2}(9/8)^k}}{1+e^{-\frac{1}{2}(9/8)^k}}\\ 
 & \geq e^{-17\beta^2 -2\beta} \prod_{k=1}^{\infty} \frac{1-e^{-\frac{1}{2}(9/8)^k}}{1+e^{-\frac{1}{2}(9/8)^k}}\,. 
 \end{align*}
 Observe that the final infinite product is convergent. Applying this lower bound to \eqref{equa:coro-ising-model-gap-Aux1}, we conclude the result.
 \end{proof}

% \begin{proof}
% If we apply Theorem~\ref{Theo:gapIntervals1D} using the function $\delta(\ell)$ from Proposition~\ref{Prop:correlationDecay1DIsing}, we deduce the first inequality. We also need using that for a finite interval $J$ with $|J| = 9$, we can estimate $\eta(\sigma_{J}) \leq e^{\beta |J| \| \Phi\|} \leq e^{\beta 18}$ as pointed out in Remark~\ref{Rema:etaProperties}.(iv). To prove the second inequality, let us first observe that
% \[ (9/8)^{k} = ((9/8)^{6})^{k/6} \geq 2^{k/6} \geq 2^{\lfloor k/6 \rfloor}\,. \]
% Using then that the function $y \mapsto \frac{1-y}{1+y}$ is decreasing on $y >0$,
% \[ \prod_{k=0}^{\infty} \frac{1-\tanh(\beta)^{(9/8)^k}}{1+\tanh(\beta)^{(9/8)^k}} \geq \prod_{k=0}^{\infty} \frac{1-\tanh(\beta)^{2^{\lfloor k/6 \rfloor}}}{1+\tanh(\beta)^{2^{\lfloor k/6 \rfloor}}} = \left( \prod_{k=0}^{\infty} \frac{1-\tanh(\beta)^{2^{k}}}{1+\tanh(\beta)^{2^{k}}}\right)^{6}\,.  \]
% Finally, we can use~\ref{} to simplify
% \[ \prod_{k=0}^{\infty} \frac{1-\tanh(\beta)^{2^{k}}}{1+\tanh(\beta)^{2^{k}}} = \frac{1-\tanh(2\beta)}{1+\tanh(2\beta)} = e^{-4 \beta}\,. \]
% Combining these estimates we conclude the result.
% \end{proof}

\subsection{1D General case}

Let $\Phi$ be any local interaction on $\Lambda_N = \mathbb{Z}_{N}$ with finite range $r>0$ and strength $\| \Phi\| \leq J$ for some $J>0$. As a main result, we demonstrate that for every (inverse) temperature $\beta >0$, the canonical purified Hamiltonian associated to the Gibbs state $\sigma_{\beta}$ has a spectral gap that is lower bounded  by a constant independent of the system size. To simplify notation, we will absorb the (inverse) temperature parameter $\beta$ of the Gibbs state into the local interaction, and thus work only with $\sigma = e^{-H}/\Tr(e^{-H})$, where $\beta$ is implicitly included in $J$. 

\begin{Theo}\label{Theo:1DCorrelationDecay}
Under the above conditions, there exist positive real constants $c=c(r,J,d)>0$ and $\alpha=\alpha(r,J,d) >0$ such that the following holds. Let $\Lambda_N$ be partitioned into four consecutive intervals $\Lambda_N=AI_{1}CI_{2}$ as in Figure \eqref{fig:1DringSplitting}, with $|I_{1}|, |I_{2}| \geq \ell$ for some integer $\ell \geq 1$. Then, for any $D \in \{ I_{1}, I_{2}, \emptyset\}$ we get
\[
 \Delta_{\sigma}(A:C|D) \leq c e^{-\alpha \sqrt{\ell}}\,.
\]
\end{Theo}

The rate of decay in the above result is determined by the correlation decay obtained by Kimura and Kuwahara in \cite{kimura25} for Gibbs states on finite chains, that we will use as an auxiliary result. When the local interaction is translation-invariant, this estimate can be strengthened to exhibit exponential decay, by using a previous result of Araki \cite{araki69} (see also \cite{Bluhm2022}). However, for the sake of generality, and since this subexponential rate of decay suffices for our purposes, we will rely on the subexponential result.

The proof of Theorem~\ref{Theo:1DCorrelationDecay} is long and follows ideas that have already been developped in \cite{Bluhm2022, Gondolf24}. It is given in detail in Appendix~\ref{sec:Proof1DGeneral}, together with the proof of the following result.

\begin{Prop}\label{Prop:eta1D}
Under the above conditions, there exists a real constant $c'=c'(r,J)>0$ such that for every interval $I \subset \Lambda_N$ 
\[ \eta_{I}(\sigma) \leq c' e^{J \cdot |I|} \]
\end{Prop}

From these results, we immediately get the following:

\begin{Theo}
Let $\mathbf{H}_{\Lambda_N}$ be the canonical purified Hamiltonian associated to the Gibbs state $\sigma$. Then, there is a constant $\gamma = \gamma(r,J,d)>0$ independent of $N$, such that 
\[ \operatorname{gap}(\mathbf{H}_{\Lambda_N}) \geq \gamma \,.\]

\end{Theo}

\begin{proof}
Let $\alpha, c>0$ be the constants appearing in Theorem~\ref{Theo:1DCorrelationDecay}, and let us fix 
\[ \mu = \lceil 18 \left(1+\tfrac{1+\ln{(1+c)}}{\alpha}\right)^2\rceil \geq 18.\]
For every interval $I \subset \Lambda_{N}$ with $|I| \leq \mu$, we can estimate by Proposition~\ref{Prop:eta1D}
\[ \eta_{I}(\sigma)^{4} \leq c' e^{4 J \mu}\,. \]
Thus, if $N \leq \mu$ have, by Theorem~\ref{Theo:roughBoundSpectralGap},
\[ \operatorname{gap}(\mathbf{H}_{\Lambda_N}) \geq (1/c') e^{-4J \mu}\,. \]
On the other hand, if $N > \mu$, we can apply Theorem~\ref{Theo:gapIntervals1D}  to estimate
\[ \operatorname{gap}(\mathbf{H}_{\Lambda_N}) \geq e^{-5} \cdot \left(   \prod_{k=0}^{\infty} (1-\delta_{k})\right) (1/c') e^{-4J \mu} \geq e^{-5} \cdot \left(   \prod_{k=0}^{\infty} (1-e^{-(9/8)^{k/2}})\right) (1/c') e^{-4J \mu} \,, \]
where we have used that for every $k \geq 0$,
\[ \delta_{k} = \delta(\lfloor\tfrac{\mu}{9}(9/8)^{k}\rfloor) \leq \delta(\tfrac{\mu}{18}(9/8)^{k}) \leq \tfrac{c}{c+1} e^{-  (9/8)^{k/2}} \leq e^{- (9/8)^{k/2} } < 1\,. \]
Note that the resulting infinite product is convergent to a positive constant, finishing the proof.
\end{proof}

\section{2D Quantum Double Models}\label{sec:quantum-double}
In this section, we will recall the construction of the quantum double models by Kitaev \cite{Kitaev2003}, and then show that their Gibbs states satisfy the sufficient condition for the gap of the purified canonical Hamiltonian $\mathbf{H}$ introduced in Section~\ref{sec:gibbs-states}.

\subsection{Definition}
Quantum double models \cite{Kitaev2003} are a class of quantum spin systems defined on a square lattice with periodic boundary conditions. Let us recall the definition of their Hamiltonian.
Let $\Lambda_{N}$ be the set of midpoints of the edges of the square lattice $\mathbb{Z}_N \times \mathbb{Z}_{N}$ for some $N \in \mathbb{N}$. Let us denote by $\cV=\cV_{N} \equiv \mathbb{Z}_N \times \mathbb{Z}_{N}$ the set of vertices, and by $\cE=\cE_{N}$ the set of edges of $\mathbb{Z}_N \times \mathbb{Z}_{N}$. Each edge is given an orientation: for simplicity, we will assume that all horizontal edges point to the left, while vertical edges point downwards.

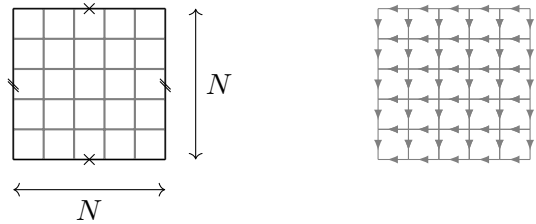
\[
\begin{tikzpicture}[equation, scale=0.4]

    \draw[step=1.0,gray,thick] (0,0) grid (5,5);
    \draw[postaction=torus horizontal] (0,0) -- (5,0);
    \draw[postaction=torus horizontal] (0,5) -- (5,5);
    \draw[postaction=torus vertical] (5,0) -- (5,5);
    \draw[postaction=torus vertical] (0,0) -- (0,5);

    \draw[<->] (0,-1) -- node[below]{$N$} (5,-1);
    \draw[<->] (6,0) --  node[right]{$N$} (6,5);

\begin{scope}[xshift=12cm, decoration={
    markings,
    mark=at position 0.7 with {\arrow{latex}}}]
\foreach \x in {1,...,5} 
\foreach \y in {1,...,5}{

\begin{scope}[xshift=\x cm, yshift=\y cm]
\draw[postaction={decorate}, gray] (0,0)  -- (-1,0); 
\draw[postaction={decorate}, gray] (-1,0)  -- (-1,-1); 
\draw[postaction={decorate}, gray] (0,0)  -- (0,-1); 
\draw[postaction={decorate}, gray] (0,-1)  -- (-1,-1); 
\end{scope}
}
\end{scope}

\end{tikzpicture}
\]
Let us fix an arbitrary finite group $G$ and denote by $\ell_2(G)$ the complex finite dimensional Hilbert space with orthonormal basis $\{ \ket{g} \mid g \in G\}$. At each edge $e \in \cE$ we have a local Hilbert space $\mathcal{H}_e$ and a space of observables $\cB_e$ defined as
\[
\mathcal{H}_{e} = \ell_{2}(G) \quad \text{and} \quad  \mathcal{B}_{e} = \mathcal{B}(\mathcal{H}_{e}) = \mathcal{M}_{|G|}(\mathbb{C}).
\]
The quantum double Hamiltonian for group $G$ on $\Lambda_N$ is then given by
\begin{equation}
    H_N := - \sum_{s} A_s - \sum_{p} B_p
\end{equation}
where $s$ runs over the \emph{stars} of $\Lambda_N$ (the set of 4 edges incident to a given vertex), $p$ runs over the \emph{plaquettes} of $\Lambda_N$ (the set of 4 edges forming a face in $\Lambda_N$). Each local operator $A_{s}$ and $B_{p}$ is supported on $s$ and $p$, respectively. 
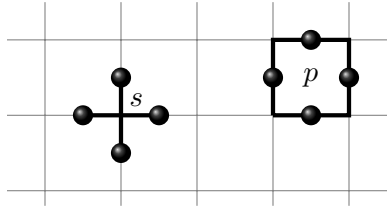
\[
\begin{tikzpicture}[equation]   
  \draw[step=1.0,gray,thin] (-1.5,-1.2) grid (3.5,1.5);
  \draw[ultra thick] (-0.5,0) -- (0.5, 0);
  \draw[ultra thick] (0,-0.5) -- (0,0.5);
  \shade[ball color=black] (-0.5,0) circle (0.8ex);
  \shade[ball color=black] (0.5,0) circle (0.8ex);
  \shade[ball color=black] (0,-0.5) circle (0.8ex);
  \shade[ball color=black] (0,0.5) circle (0.8ex);    
  \draw (0.2, 0.2) node {$s$};
  \begin{scope}[xshift=2cm, yshift=0cm]
  \draw[ultra thick] (0,0) -- (1, 0) -- (1,1) -- (0,1) -- (0,0);
  \shade[ball color=black] (0.5,0) circle (0.8ex);
  \shade[ball color=black] (1,0.5) circle (0.8ex);
  \shade[ball color=black] (0.5,1) circle (0.8ex);
  \shade[ball color=black] (0,0.5) circle (0.8ex);
\draw (0.5, 0.5) node {$p$};
  \end{scope}
\end{tikzpicture}
\]
We define the \emph{left} ($L_g$) and \emph{right} ($R_g$) \emph{regular representation} of $G$ on $\ell_2(G)$ as
\[ L^{g} = \sum_{h \in G} \ketbra{gh}{h}  \quad , \quad R^{g}=\sum_{h \in G} \ketbra{hg^{-1}}{h}\,. \]
For each $g\in G$ and each star $s$, we define the operator
\begin{equation}
  \begin{tikzpicture}[equation]   
  \draw[step=1.0,gray,thin] (-1,-1) grid (1,1);
  \draw[ultra thick] (-0.5,0) -- (0.5, 0);
  \draw[ultra thick] (0,-0.5) -- (0,0.5);
  \shade[ball color=black] (-0.5,0) circle (0.8ex);
  \shade[ball color=black] (0.5,0) circle (0.8ex);
  \shade[ball color=black] (0,-0.5) circle (0.8ex);
  \shade[ball color=black] (0,0.5) circle (0.8ex);    
  \draw (0.2, 0.2) node {$s$};
\end{tikzpicture} \hspace{2cm}
 A_{s}(g) \, =  
  \begin{tikzpicture}[equation]  
    \draw (0,0) node {$\otimes$};
    \draw (0,0.5) node {$L^{g}$};
    \draw (0.5,0) node {$L^{g}$};
    \draw (0,-0.5) node {$R^{g}$};
    \draw(-0.5,0) node {$R^{g}$};
  \end{tikzpicture}
\end{equation}
(where the choice between the left and right regular representation is due to the orientation of the edges). 
This is a representations of $G$, which we can extend to its group algebra $\ell^{2}(G)$. In particular, taking $\lambda := \frac{1}{\abs{G}} \sum_{g \in G} g$ the Haar integral of $\ell^{2}(G)$, we define the \emph{star operator} as
\[ A_{s} \, = A_{s}(\lambda) = \, \displaystyle \frac{1}{|G|}\sum_{g \in G} A_{s}(g)\,. \]
Similarly, for $g\in G$, let us denote by $\delta_g$ the delta function which takes value one at $g$ and zero elsewhere. We then define $B_p(\delta_g)$ be a operator acting on the edges of the plaquette $p$ as
\[
  \begin{tikzpicture}[equation]   
  \draw[step=1.0,gray,thin] (-0.5,-0.5) grid (1.5,1.5);
  \draw[ultra thick] (0,0) -- (1, 0) -- (1,1) -- (0,1) -- (0,0);
  \shade[ball color=black] (0.5,0) circle (0.8ex);
  \shade[ball color=black] (1,0.5) circle (0.8ex);
  \shade[ball color=black] (0.5,1) circle (0.8ex);
  \shade[ball color=black] (0,0.5) circle (0.8ex);
\draw (0.5, 0.5) node {$p$};
\end{tikzpicture}
\hspace{1cm}
 B_p(\delta_g) = \, \sum_{g_{1}, g_{2}, g_{3}, g_{4} \in G} \delta_g(g_{1} g_{2} g_{3}^{-1} g_{4}^{-1}) \,\,\,\,
  \begin{tikzpicture}[equation]  
    \draw (0,0) node {$\otimes$};
    \draw (0,0.5) node {$\dyad{g_{1}}$};
    \draw(-1,0) node {$\dyad{g_{2}}$};
    \draw (0,-0.5) node {$\dyad{g_{3}}$};
    \draw (1,0) node {$\dyad{g_{4}}$};
  \end{tikzpicture}
\]
(once again, the choice of the inverses are due to the orientation of the edges). We can extend by linearity the definition of $B_p(\cdot)$ to obtain a representation of the algebra $\ell^{2}(G)$ of functions from $G$ to $\mathbb{C}$. In particular, $\delta_1$ is its Haar integral, which can be expressed in terms of the character of the regular representation as $\delta_{1} =  \frac{1}{\abs{G}} \chi^{\text{reg}}$. The plaquette operator $B_{p}$ is then given by 
\begin{equation}\label{equa:plaquetteOperatorDefinition}
 B_{p} = B_p(\delta_1) = \, \displaystyle \frac{1}{|G|} \sum_{g_{1}, g_{2}, g_{3}, g_{4} \in G} \chi^{reg}(g_{1} g_{2} g_{3}^{-1} g_{4}^{-1}) \,\,\,\,
  \begin{tikzpicture}[equation]  
    \draw (0,0) node {$\otimes$};
    \draw (0,0.5) node {$\dyad{g_{1}}$};
    \draw(-1,0) node {$\dyad{g_{2}}$};
    \draw (0,-0.5) node {$\dyad{g_{3}}$};
    \draw (1,0) node {$\dyad{g_{4}}$};
  \end{tikzpicture}
\end{equation}
It can be easily checked that both star and plaquette operators $A_{s}$ and $B_{p}$ are orthogonal projections and moreover commute with each other. Therefore that $H_N$ is a local commuting Hamiltonian.

When $G$ is an Abelian group, the plaquette operators admit a more gentle form. In this case, we can associate to $G$ its \emph{dual group} $\widehat{G}$, made of all homomorphisms $\chi:G \to \mathbb{T}$ where $\mathbb{T} = \{ z \in \mathbb{C} \colon |z| = 1 \}$, also referred as \emph{characters}. Characters act on $\ell_2(G)$ as diagonal operators, by defining $\Pi_\chi \ket{g} = \chi(g) \ket{g}$. The inverse of a given character $\chi$ is denoted by $\overline{\chi}$ since it is obtained by conjugation: $\overline{\chi}(g) = \overline{\chi(g)}$ for every $g \in G$. In particular, the unit of this group is the \emph{trivial character} $\mathbf{1}_{G}$ constantly equal to one. This set has the same cardinality as $G$ (actually, they are isomorphic groups), moreover and $\sum_{X \in \widehat{G}} \chi = \chi^{reg}$. Inserting this expression in \eqref{equa:plaquetteOperatorDefinition} and rearranging summands we get
\begin{equation}
  \begin{tikzpicture}[equation]   
  \draw[step=1.0,gray,thin] (-0.5,-0.5) grid (1.5,1.5);
  \draw[ultra thick] (0,0) -- (1, 0) -- (1,1) -- (0,1) -- (0,0);
  \shade[ball color=black] (0.5,0) circle (0.8ex);
  \shade[ball color=black] (1,0.5) circle (0.8ex);
  \shade[ball color=black] (0.5,1) circle (0.8ex);
  \shade[ball color=black] (0,0.5) circle (0.8ex);
\draw (0.5, 0.5) node {$p$};
\end{tikzpicture}
\hspace{2cm}
B_{p} \, = \, B_{p}(\delta_{1}) = \, \displaystyle\frac{1}{|G|} \, \sum_{\chi \in \widehat{G}} B_{p}(\chi) \quad \text{ where } \quad  B_{p}(\chi):= 
\begin{tikzpicture}[equation]  
    \draw (0,0) node {$\otimes$};
    \draw (0,0.4) node {$\Pi_{\chi}$};
    \draw(-0.4,0) node {$\Pi_{\chi}$};
    \draw (0,-0.4) node {$\Pi_{\overline{\chi}}$};
    \draw (0.45,0) node {$\Pi_{\overline{\chi}}$};
  \end{tikzpicture};
\end{equation}
It is easy to verify that $A_s(g)$ and $B_p(\chi)$ commute for every $s$, $p$, $g$ and $\chi$.\\

Let us introduce some further notation. We denote by $\cS$ and $\cP$ the sets of all stars and plaquettes, respectively. Given a subset $\cR \subset \cE$, we define $\cS_{\cR} \subset \cS$ as the set of all stars $s$ such that $s \cap \cR \neq \emptyset$, and similarly, define $\mathcal{P}_{\cR} \subset \mathcal{P}$ as the set of all plaquettes $p$ such that $p \cap \cR \neq \emptyset$. Let $\overline{\cR}_{S}$ denote the union of all stars in $\cS_{\cR}$, and $\overline{\cR}_{P}$ the union of all plaquettes in $\mathcal{P}_{\cR}$. Finally, define $\overline{\cR} := \overline{\cR}_{S} \cup \overline{\cR}_{P}$.

% \noindent Given a subset $\cR \subset \cE$ we introduce the following notation:
% \begin{itemize}
% \item $\cS$ and $\mathcal{P}$ are the set of all stars and plaquettes, respectively.
% \item $\cS_{\cR}$ is the set of all stars $s \in \cS$ such that $s \cap \cR \neq \emptyset$. 
% \item $\mathcal{P}_{\cR}$ is the set of all plaquettes $p \in \mathcal{P}$ such that $p \cap \cR \neq \emptyset$.
% \item $\overline{\cR}_{S}$ is the union of all stars in $\cS_{\cR}$. 
% \item $\overline{\cR}_{P}$ is the union of all plaquettes in $\mathcal{P}_{\cR}$.
% \item $\overline{\cR}$ is the union of $\overline{\cR}_{S}$ and $\overline{\cR}_{P}$.
% \end{itemize}
\begin{Defi}
We will say that $\cR$ is \emph{connected by stars} if given any two edges $e,e' \in \cR$ there is a finite sequence of stars $s_{1}, \ldots, s_{n} \in \cS_{\cR}$ such that $e \in s_{1}$, $e' \in s_{n}$ and $\emptyset \neq s_{j} \cap s_{j+1} \subset \cR$ for every $j$. Observe that this is equivalent to saying that given any two stars $s, s' \in \cS_{\cR}$, there is a finite sequence $s=s_{1}, s_{2}, \ldots, s_{n} = s'$ such that $\emptyset \neq s_{j} \cap s_{j+1} \subset \mathcal{R}$ for every $j$.

Analogously, we will say that $\cR$ is \emph{connected by plaquettes} if given any two edges $e,e' \in \cR$ there is a finite sequence of plaquettes $p_{1}, \ldots, p_{n} \in \mathcal{P}_{\cR}$ such that $e \in p_{1}$, $e' \in p_{n}$ and $\emptyset \neq p_{j} \cap p_{j+1} \subset \cR$ for every $j$. Observe that this is equivalent to saying that given any two stars $p, p' \in \cS_{\cR}$, there is a finite sequence of plaquettes $p=p_{1}, p_{2}, \ldots, p_{n} = p'$ such that $\emptyset \neq p_{j} \cap p_{j+1} \subset \cR$ for every $j$.

% Finally, two regions $\cR_1$ and $\cR_2$ are said to be \emph{disconnected by stars} (resp. \emph{by plaquettes}) if there are no stars (resp. plaquettes) intersecting both, i.e., if $\cS_{\cR_1} \cap \cS_{\cR_2} = \emptyset$ (resp. $\mathcal{P}_{\cR_1} \cap \mathcal{P}_{\cR_2} = \emptyset$). %In this case, we denote their union as $\cR_1 \sqcup \cR_2$.
\end{Defi}

\subsection{Marginals of the Gibbs state}
Let us consider the marginal of the Gibbs state $\sigma=\sigma_{\beta}$ on the complement $\cR^c$ of a region $\cR$
\[ \sigma_{\cR^{c}} = \operatorname{Tr}_{\cR}(\rho) = \frac{1}{Z_\beta} \operatorname{Tr}_{\cR}(e^{-\beta H}) \,. \]
If we focus on the partial trace, we can extract the local factors whose support is disjoint with the region $\cR$ 
\[ \operatorname{Tr}_{\cR}(e^{-\beta H}) = \operatorname{Tr}_{\cR}(e^{-\beta H_{\cR}^{\partial}}) e^{-\beta H_{\cR^{c}}}. \]
We can further expand this expression as
\[ \operatorname{Tr}_{\cR}\left( \prod_{s \in \cS} e^{\beta A_{s}} \, \prod_{p \in \mathcal{P}} e^{\beta B_{p}} \right) =  \operatorname{Tr}_{\cR}\left( \prod_{s \in \cS_{\cR}} e^{\beta A_{s}} \, \prod_{p \in \mathcal{P}_{\cR}} e^{\beta B_{p}} \right) \,\, \prod_{s \in \cS \setminus \cS_{\cR}} e^{\beta A_{s}} \, \prod_{p \in \mathcal{P} \setminus \mathcal{P}_{\cR}} e^{\beta B_{p}}\,. \]

\begin{Theo}[Abelian groups]\label{Theo:abelianModelsMarginals}
Let $G$ be an Abelian group and define
\[ A_{\cS_\cR} := \frac{1}{|G|} \sum_{g \in G} \prod_{s \in \cS_\cR} A_{s}(g) 
\quad , \quad 
B_{\mathcal{P}_\cR} := \frac{1}{|G|} \sum_{\chi \in \hat{G}} \prod_{p \in \mathcal{P}_\cR} B_{p}(\chi)\,. \]
These are orthogonal projections that belong to the commuting algebra generated by the local operators $\{ A_{s}(g), B_{p}(\chi) \colon s \in \cS ,p \in \cP , g \in G ,\chi \in \widehat{G}\}$. If $\cR$ is connected by stars and plaquettes, then
\begin{multline} 
\operatorname{Tr}_{\cR}\left( \prod_{s \in \cS_{\cR}} e^{\beta A_{s}} \, \prod_{p \in \mathcal{P}_{\cR}} e^{\beta B_{p}} \right) = \textstyle \kappa_{\cR} \left( \left(1- \left(\frac{\gamma_{\beta}}{1+\gamma_{\beta}}\right)^{|\cS_\cR|}\right) \mathbbm{1} + |G|\left(\frac{\gamma_{\beta}}{1+\gamma_{\beta}}\right)^{|\cS_\cR|} A_{\cS_\cR} \right) \notag \\
\textstyle \cdot \left( \left(1-\left(\frac{\gamma_{\beta}}{1+\gamma_{\beta}}\right)^{|\cP_\cR|}\right) \mathbbm{1} + |G|\left(\frac{\gamma_{\beta}}{1+\gamma_{\beta}}\right)^{|\cP_\cR|} B_{\cP_\cR} \right),
\end{multline}
where we are denoting $\gamma_{\beta} := (e^{\beta}-1)/|G|$ and $\kappa_{\cR} := |G|^{|\cR|} \left( 1 + \gamma_{\beta} \right)^{|\cS_\cR|} \left( 1 + \gamma_{\beta} \right)^{|\cP_\cR|}$. As a consequence, $\sigma_{\cR^{c}}$ belongs to the above commuting algebra.
\end{Theo}

The proof of the above result is provided in Appendix~\ref{Appendix:Marginals_AbelianQDM}. In particular, the above result implies that if $\cR$ is connected by stars and plaquettes, then $\Tr_{\cR}(e^{-\beta H_{\cR}^{\partial}})$ is close to a multiple of the identity when $\cR$ is \emph{large}. This property can, in fact, be extended  to every (not necessarily abelian) group, as we show in the next result.
\begin{Theo}\label{thm:marginals-decayQDM}
Let $G$ be an arbitrary group and let $\mathcal{R}$ be a region connected by stars and plaquette operators. Then, for every $\beta \in \mathbb{R}$
\begin{multline}
\textstyle
\kappa_\cR 
\left( 1-\left(\frac{\gamma_{\beta}}{1+\gamma_{\beta}}\right)^{|\cP_\cR|}  \right) \, \left( 1-\left(\frac{\gamma_{\beta}}{1+\gamma_{\beta}}\right)^{|\cS_\cR|}    \right) \, \mathbbm{1}\\[2mm]
\leq
\Tr_{\cR}\left( \prod_{s \in \cS_{\cR}} e^{\beta A_{s}} \, \prod_{p \in \mathcal{P}_{\cR}} e^{\beta B_{p}} \right) 
\leq  \hspace{3cm}\\[2mm]
\textstyle
\kappa_\cR 
\left( 1  + (|G|-1)\left(\frac{\gamma_{\beta}}{1+\gamma_{\beta}}\right)^{|\cP_\cR|}  \right) \, \left( 1  + (|G|-1)\left(\frac{\gamma_{\beta}}{1+\gamma_{\beta}}\right)^{|\cS_\cR|}  \right) \, \mathbbm{1}\,,
\end{multline}
%\[ 
%\kappa_{\cR} \, \left(1 - \left( \frac{\gamma_{\beta}}{1+\gamma_{\beta}}\right)^{m} \right)^{2} \mathbbm{1} \leq \Tr_{\cR}\left( \prod_{s \in \cS_{\cR}} e^{\beta A_{s}} \, \prod_{p \in \mathcal{P}_{\cR}} e^{\beta B_{p}} \right) \, \leq \, \kappa_{\cR} \, \left(1 + (|G|-1)\left( \frac{\gamma_{\beta}}{1+\gamma_{\beta}}\right)^{m} \right)^{2} \mathbbm{1}\,, 
%\]
where we are denoting $\gamma_{\beta} := (e^{\beta}-1)/|G|$ and
$\kappa_{\cR} := |G|^{|\cR|} \left( 1 + \gamma_{\beta} \right)^{|\cS_\cR|} \left( 1 + \gamma_{\beta} \right)^{|\cP_\cR|}$. As a consequence, if we define $m_{\cR}:=\min\{ |\cS_{\cR}|, |\mathcal{P}_{\cR}|\}$, then there is a Hermitian operator $Q_{\partial \cR}$ supported on $\partial \cR$ such that
\[ \Tr_{\cR}(e^{-\beta H_{\cR}^{\partial}})  = \kappa_{\cR} (\mathbbm{1} + Q_{\partial \cR}) \quad \text{ and } \quad \| Q_{\partial \cR}\|_{\infty} \leq (|G|^{2}-1) \left( \frac{\gamma_{\beta}}{1+\gamma_{\beta}}\right)^{m_\cR}. \]
\end{Theo}
\begin{proof}
The first part of the theorem is proved in Appendix~\ref{Appendix:Marginals_GeneralQDM}. The second part is a straightforward consequence of it. On the one hand,
\begin{align*}
Q_{\partial \cR}:=\frac{1}{\kappa_\cR}\Tr_{\cR}(e^{-\beta H_{\mathcal{R}}^{\partial}}) - \mathbbm{1} 
& \textstyle \leq \left(\left(1+(|G|-1) \left( \frac{\gamma_{\beta}}{1+\gamma_{\beta}}\right)^{m_\cR}\right)^{2}-1\right) \mathbbm{1} \\
& \textstyle = (|G|-1) \left( \frac{\gamma_{\beta}}{1+\gamma_{\beta}}\right)^{m_\cR} \left( 2 + (|G|-1) \left( \frac{\gamma_{\beta}}{1+\gamma_{\beta}}\right)^{m_\cR} \right) \mathbbm{1}\\
&\textstyle  \leq (|G|-1) (|G|+1) \left( \frac{\gamma_{\beta}}{1+\gamma_{\beta}}\right)^{m_\cR} \mathbbm{1}\\
& \textstyle \leq (|G|^{2}-1) \left( \frac{\gamma_{\beta}}{1+\gamma_{\beta}}\right)^{m_{\cR}} \mathbbm{1}\,.
\end{align*}
On the other hand,
\begin{align*} 
\textstyle  Q_{\partial \cR} \geq \left(\left(1 - \left( \frac{\gamma_{\beta}}{1+\gamma_{\beta}}\right)^{m_\cR} \right)^{2} - 1\right) \mathbbm{1}   = - \left( \frac{\gamma_{\beta}}{1+\gamma_{\beta}}\right)^{m_\cR} \left(2- \left( \frac{\gamma_{\beta}}{1+\gamma_{\beta}}\right)^{m_\cR} \right) \mathbbm{1}\geq - 2 \left( \frac{\gamma_{\beta}}{1+\gamma_{\beta}}\right)^{m_\cR} \mathbbm{1}\,. 
\end{align*}
This finishes the proof.
\end{proof}

% \begin{proof}
% Let us denote
% \[ A_{\cS_\cR} = \frac{1}{|G|} \sum_{g \in G} \prod_{s \in \cS_\cR} A_{s}(g) 
% \quad , \quad 
% B_{\mathcal{P}_\cR} = \frac{1}{|G|} \sum_{\chi \in \hat{G}} \prod_{p \in \mathcal{P}_\cR} B_{p}(\chi)\,. \]
% It is easy to see that these are orthogonal projections, as they are self-adjoint and idempotent. We have the expression:
% \[ 
% \Tr_{\cR}(\prod_{s \in \cS_{\mathcal{R}}} e^{\beta_{1}A_{s}} \prod_{p \in \mathcal{P}_{\cR}} e^{\beta B_{p}}) = 
% \kappa_{\cR} \left( \mathbbm{1} + \left(\frac{\gamma_{\beta}}{1+\gamma_{\beta}}\right)^{|\cS_{\cR}|}\sum_{\substack{g \in G\\ g \neq 1}} \prod_{s \in \cS_{\mathcal{R}}} A_{s}(g)\right)  
% \left( \mathbbm{1} + \left( \frac{\gamma_{\beta}}{1+\gamma_{\beta}}\right)^{|\mathcal{P}_{\cR}|}  \sum_{\substack{\chi \in \widehat{G}\\ \chi \neq \mathbf{1}}} \prod_{p \in \mathcal{P}_\cR} B_{p}(\chi) \right) 
% \]
% where
% \[ \kappa_{\cR}:= |G|^{|\cR|}(1+\gamma_{\beta})^{|\cS_\cR|} (1+\gamma_{\beta})^{|\mathcal{P}_\cR|} \]

% \end{proof}

\subsection{Spectral gap: Abelian group case}

\begin{Theo}\label{Theo:CorrelationDecayAbelianQDM}
Let $N \geq 8$, let $G$ be an Abelian group, and let $\sigma = \sigma_{\beta}$ be the Gibbs state associated to the quantum double model on $\Lambda_{N}$. Consider a partition $\Lambda_N=ABCD$ into four subsets  and an integer $\ell \geq 2$ satisfying any of the following configurations:
\begin{enumerate}
\item[$(i)$] $\Lambda=AB_{1}CB_{2}$ as in Assumption~\ref{Assumtion2Dtorus}.$(i)$, where $B=B_{1} B_{2}$ and $D=\emptyset$. 
\item[$(ii)$] $\Lambda = AB_{1}CB_{2}D$ as in Assumption~\ref{Assumtion2Dtorus}.$(ii)$, where \mbox{$B=B_{1}B_{2}$}. 
\item[$(iii)$] $\Lambda = ABCD$ as in Assumption~\ref{Assumtion2Dtorus}.$(iii)$.  
\end{enumerate}
Then, the following inequality holds
\begin{equation}\label{equa:CorrelationDecayAbelianQDM}
\Delta_{\sigma}(A:C|D) \leq \delta(\ell):=1- \left(\frac{1 - \left( \frac{\gamma_{\beta}}{1+\gamma_{\beta}}\right)^{\ell^{2}} }{1 - \left( \frac{\gamma_{\beta}}{1+\gamma_{\beta}}\right)^{\ell^{2}}  + |G|  \left( \frac{\gamma_{\beta}}{1+\gamma_{\beta}}\right)^{\ell^{2}}}\right)^{6} \leq 6 |G| \left( \frac{\gamma_\beta}{1+\gamma_\beta}\right)^{\ell^2}\,. 
\end{equation}
\end{Theo}
\begin{proof} 
To prove the second inequality in \eqref{equa:CorrelationDecayAbelianQDM}, simply note that for every $a,b \geq 0$ with $a+b \geq 1$ it holds that
\[ 1-\frac{a^{6}}{(a+b)^{6}} = \frac{(a+b)^{6} - a^{6}}{(a+b)^{6}} \leq \frac{6(a+b)^{5} b}{(a+b)^{6}} = \frac{6b}{a+b} \leq 6b \,.\]
We now turn to the proof of the first inequality. Observe that in all three cases ($i$)-($iii$), rectangles $BC$, $ABC$ and $AB$ are connected by plaquettes and by stars. Then, by Theorem~\ref{Theo:abelianModelsMarginals}, the marginals $\sigma_{AD}, \sigma_{D}, \sigma_{DC}$ commute with each other, and so we can apply Proposition~\ref{prop:commutingMartingaleUpperBound} to estimate 
\begin{equation}\label{equa:equa:abeliancaseAux0}
\Delta_{\sigma}(A:C|D) \leq 1- \| \sigma_{AD} \sigma_{D}^{-1} \sigma_{DC} \sigma_{ADC}^{-1}\|_{\infty}^{-1}\,. 
\end{equation}
Let us then consider
\begin{equation}\label{equa:abeliancaseAux1}
\sigma_{AD} \sigma_{D}^{-1} \sigma_{DC} \sigma_{ACD}^{-1}
= \Tr_{BC}(e^{-\beta H}) \Tr_{ABC}(e^{-\beta H})^{-1}  \Tr_{AB}(e^{-\beta H})  \Tr_{B}(e^{ -\beta H})^{-1}\,, 
\end{equation}
and the following factorizations
\begin{equation}\label{equa:abeliancaseAux2}
\begin{split}
\Tr_{BC}(e^{-\beta H})   = \Tr_{BC}(e^{-\beta H^{\partial}_{BC}}) e^{-\beta H_{AD}} \quad & , \quad \Tr_{ABC}(e^{-\beta H})   = \Tr_{ABC}(e^{-\beta H^{\partial}_{ABC}}) e^{- \beta H_{D}}\,, \\[2mm]
\Tr_{AB}(e^{-\beta H})   = \Tr_{AB}(e^{-\beta H^{\partial}_{AB}}) e^{-\beta H_{CD}} \quad & , \quad \Tr_{B}(e^{-\beta H})   = \Tr_{B}(e^{-\beta H^{\partial}_{B}}) e^{-\beta H_{ACD}} \,.
\end{split}
\end{equation}
Next, we split the proof into two parts: we first address case ($iii$), and later consider cases ($i$)-($ii$).

In the case ($iii$), the rectangles $ABC$, $AB$, $BC$ and $B$ are all connected by stars and plaquettes. Therefore, if we substitute the expressions from \eqref{equa:abeliancaseAux2} into \eqref{equa:abeliancaseAux1}, and use that all eight resulting factors commute with each other by  Theorem~\ref{Theo:abelianModelsMarginals}, we get
\[ \sigma_{AD} \sigma_{D}^{-1} \sigma_{DC} \sigma_{ACD}^{-1} = \Tr_{BC}(e^{-\beta H_{BC}^{\partial}}) \, \Tr_{ABC}(e^{-\beta H_{ABC}^{\partial}})^{-1} \, \Tr_{AB}(e^{-\beta H_{AB}^{\partial}}) \, \Tr_{B}(e^{-\beta H_{B}^{\partial}})^{-1}\,. \]
We have used that $H_{AD} - H_{D} = H_{ACD} - H_{CD}$ since the regions $A$ and $C$ are separated by at least two plaquettes, ensuring that no interaction terms connect them directly. We can now apply again Theorem~\ref{Theo:abelianModelsMarginals}, and use that $|\cS_\cR|, |\cP_\cR| \geq \ell^{2}$ for $\cR \in \{ AB,B,BC,ABC\}$, to estimate 
\begin{align*} 
\|\sigma_{AD} \sigma_{D}^{-1} \sigma_{DC} \sigma_{ACD}^{-1} \|_{\infty} 
& \leq \|\Tr_{BC}(e^{-\beta H_{BC}^{\partial}})\|_{\infty} \cdot \|\Tr_{ABC}(e^{-\beta H_{ABC}^{\partial}})^{-1}\|_{\infty}\\
& \hspace{3cm}  \cdot \|\Tr_{AB}(e^{-\beta H_{AB}^{\partial}})\|_{\infty} \cdot \|\Tr_{B}(e^{-\beta H_{B}^{\partial}})^{-1}\|_{\infty} \\[2mm]
&  \leq \frac{\kappa_{AB} \kappa_{BC}}{\kappa_{B} \kappa_{ABC}}  \left(\frac{1-\left( \frac{\gamma_{\beta}}{1+\gamma_{\beta}}\right)^{\ell^2}+|G| \left( \frac{\gamma_{\beta}}{1+\gamma_{\beta}}\right)^{\ell^2}}{1-\left( \frac{\gamma_{\beta}}{1+\gamma_{\beta}}\right)^{\ell^2}} \right)^{4}\,.
\end{align*}
Finally, note that  $\kappa_{AB}  \kappa_{BC} = \kappa_{B} \kappa_{ABC}$, since $|\cS_{AB}| + |\cS_{BC}| = |\cS_{B}| + |\cS_{ABC}|$ and $|\cP_{AB}| + |\cP_{BC}| = |\cP_{B}| + |\cP_{ABC}|$. These two equalities hold because, as we already explained before, $A$ and $C$ are separated by at least two plaquettes, so that $\cS_{A} \cap \cS_{C} = \cP_{A} \cap \cP_{C} = \emptyset$. Applying the resulting inequality to \eqref{equa:equa:abeliancaseAux0}, we conclude the result.

In the cases ($i$)-($ii$), rectangles $ABC$, $AB$ and $BC$ are connected by stars and plaquettes, but $B$ is not. However, it is union of two rectangles $B_{1}$ and $B_{2}$ that are separated by at least two plaquettes, so that $\cS_{B_{1}} \cap \cS_{B_{2}} = \emptyset$ and $\cP_{B_{1}} \cap \cP_{B_{2}} = \emptyset$. Therefore
\begin{equation}\label{equa:abeliancaseAux3}
\Tr_{B}(e^{-\beta H_{B}^{\partial}}) = \Tr_{B}(e^{-\beta (H^{\partial}_{B_{1}} + H^{\partial}_{B_{2}})}) = \Tr_{B_{1}}(e^{-\beta H_{B_{1}}^{\partial}}) \Tr_{B_{2}}(e^{-\beta H_{B_{2}}^{\partial}})\,. 
\end{equation}
Moreover, $B_{1}$ and $B_{2}$ are connected by stars and plaquettes. We can then substitute \eqref{equa:abeliancaseAux2}  and \eqref{equa:abeliancaseAux3} into \eqref{equa:abeliancaseAux1}, and using that all the resulting factors commute with each other by Theorem \ref{Theo:abelianModelsMarginals}, we get that
\begin{multline*} 
\sigma_{AD} \sigma_{D}^{-1} \sigma_{DC} \sigma_{ACD}^{-1} = \Tr_{BC}(e^{-\beta H_{BC}^{\partial}}) \, \Tr_{ABC}(e^{-\beta H_{ABC}^{\partial}})^{-1} \\ \Tr_{AB}(e^{-\beta H_{AB}^{\partial}}) \, \Tr_{B_1}(e^{-\beta H_{B_1}^{\partial}})^{-1} \Tr_{B_2}(e^{-\beta H_{B_2}^{\partial}})^{-1}\,. 
\end{multline*}
We have used that $H_{AD} - H_{D} = H_{ACD} - H_{CD}$ since the regions $A$ and $C$ are separated by at least two plaquettes, ensuring that no interaction terms connect them directly. Applying again Theorem \ref{Theo:abelianModelsMarginals}, and using $|\cS_\cR|, |\cP_\cR| \geq \ell^{2}$ for $\cR \in \{ AB,B_1,B_2,BC,ABC\}$, we can estimate
\begin{align*} 
\|\sigma_{AD} \sigma_{D}^{-1} \sigma_{DC} \sigma_{ACD}^{-1} \|_{\infty} 
& \leq \|\Tr_{BC}(e^{-\beta H_{BC}^{\partial}})\|_{\infty} \cdot \|\Tr_{ABC}(e^{-\beta H_{ABC}^{\partial}})^{-1}\|_{\infty} \cdot \|\Tr_{AB}(e^{-\beta H_{AB}^{\partial}})\|_{\infty}  \\[2mm]
& \hspace{3cm} \cdot \|\Tr_{B_{1}}(e^{-\beta H_{B_{1}}^{\partial}})^{-1}\|_{\infty} \cdot \|\Tr_{B_{2}}(e^{-\beta H_{B_{2}}^{\partial}})^{-1}\|_{\infty}  \\[2mm]
& \leq \frac{\kappa_{AB} \kappa_{BC}}{ \kappa_{B_{1}} \kappa_{B_{2}} \kappa_{ABC}}  \frac{\left(1-\left( \frac{\gamma_{\beta}}{1+\gamma_{\beta}}\right)^{\ell^2}+|G| \left( \frac{\gamma_{\beta}}{1+\gamma_{\beta}}\right)^{\ell^2}\right)^{4}}{\left(1-\left( \frac{\gamma_{\beta}}{1+\gamma_{\beta}}\right)^{\ell^2}\right)^{6}} .
\end{align*}
Finally, we just need to use that $\kappa_{AB} \kappa_{BC} = \kappa_{B_{1}} \kappa_{B_{2}} \kappa_{ABC}$. This is a consequence of two facts we already mentioned, namely that  $A$ and $C$ (resp. $B_{1}$ and $B_{2}$) are separated by at least two plaquettes, so that $\cS_{A} \cap \cS_{C} = \cP_{A} \cap \cP_{C} = \emptyset$ (resp. $\cS_{B_1} \cap \cS_{B_2} = \cP_{B_1} \cap \cP_{B_2} = \emptyset$). Applying the resulting inequality to \eqref{equa:equa:abeliancaseAux0}, we conclude the proof.
\end{proof}

\begin{Coro}\label{Coro:gapAbelian}
Let $G$ be an Abelian group, let $\sigma = \sigma_{\beta}$ be the Gibbs state associated to the quantum double model, and let $\mathbf{H}_{\Lambda_N}$ be the canonical purified Hamiltonian of $\sigma$. Then, for every $N \geq 2$ we have
\begin{equation}\label{equa:CorogapAbelian} \gap(\mathbf{H}_{\Lambda_N}) \geq e^{-11} e^{-(54 \beta^{2} -72 \beta)}(3+\ln|G|)^{-54 \beta } e^{-2^{22}\beta}\, \prod_{k \geq 1}(1-e^{-(9/8)^{k}}) \,, 
\end{equation}
where the last infinite product converges to a positive constant independent of $N$ and $\beta$.
\end{Coro}

\begin{proof}
Let us fix $\mu=2^{8}=256$. As pointed out in Remark~\ref{Rema:etaProperties}.$(iv)$, we can estimate for every rectangle $\cR  \in \mathcal{F}_{\mu}$  
\[\eta_{\cR}(\sigma)^{4} \leq e^{2^2\beta  \| H-H_{\cR^{c}}\|_{\infty}}  \leq e^{2^4 \beta  |\cR|} \leq e^{2^{6}\beta \mu^{2}}  = e^{\beta  2^{22}}\,.\] 
Therefore, if $N \leq \mu$, we can simply use Theorem~\ref{Theo:roughBoundSpectralGap} to deduce that
\[ \operatorname{gap}(\mathbf{H}_{\Lambda_N}) \geq e^{-\beta 2^{22}}, \]
and so inequality \eqref{equa:CorogapAbelian} cleary holds. Let us then assume that $N>\mu=2^{8}$. Applying Theorem~\ref{Theo:gap2} with these $\mu$ and $N$, and with the function $\delta(\ell)$ from Theorem~\ref{Theo:CorrelationDecayAbelianQDM}, we get 
\begin{equation}\label{equa:TheoGap2Abelian}
\gap(\mathbf{H}_{\Lambda_N}) \, \geq \,e^{-11} (1-\delta_{0})^2 \left(\, \prod_{k=0}^{\infty} (1-\delta_{k})\right) e^{-\beta 2^{22}}\,. 
\end{equation}
Here, for each $k \geq 0$, if we denote $\ell_{k} := \lfloor 2(9/8)^{k/2}\rfloor \geq (9/8)^{k/2}$, then $\delta_{k}$ is given by
\[\delta_{k}:= 1- \left(\frac{1 - \left( \frac{\gamma_{\beta}}{1+\gamma_{\beta}}\right)^{\ell_k^{2}} }{1 - \left( \frac{\gamma_{\beta}}{1+\gamma_{\beta}}\right)^{\ell_k^2}  + |G|  \left( \frac{\gamma_{\beta}}{1+\gamma_{\beta}}\right)^{\ell_k^2}}\right)^{6}\,. \]  
Next, observe that the function $y \mapsto (1-y)/(1+(|G|-1)y)$ is decreasing for $0<y <1$. Then, we can lower estimate for every $k \geq 0$
\begin{equation}\label{equa:gapAbelian_Aux1}
1-\delta_{k}  \geq \left(\frac{1 - \left( \frac{\gamma_{\beta}}{1+\gamma_{\beta}}\right)^{1} }{1 - \left( \frac{\gamma_{\beta}}{1+\gamma_{\beta}}\right)^{1}  + |G|  \left( \frac{\gamma_{\beta}}{1+\gamma_{\beta}}\right)^{1}}\right)^{6} = e^{-6\beta}. 
\end{equation}
 Using that $(1-1/x) \leq e^{-1/x}$ for $x>1$ and that $1+\gamma_{\beta} \leq e^{\beta}$, we have (see equation \eqref{equa:CorrelationDecayAbelianQDM}) 
\[
\delta_{k} \leq 6|G|\left( \frac{\gamma_{\beta}}{1+\gamma_{\beta}}\right)^{\ell_{k}^{2}} \leq  6|G|\exp\left(-\frac{\ell_{k}^2}{1+\gamma_{\beta}}\right) \leq \exp\left(\ln{6}+\ln|G| \, -\left(9/8\right)^{k} e^{-\beta}\right) . 
\]
Next, we fix $k_{0} = 9\lceil \beta + \ln(3+\ln|G|)\rceil$ . Since $(9/8)^{9} \geq e$, we deduce from the previous inequality that for every $k \geq k_{0}$ 
\begin{equation}\label{equa:gapAbelian_Aux2} 
\delta_{k} \leq \exp\left( \ln{6}+\ln{|G|} -(9/8)^{k-k_{0}}(3+\ln|G|) \right) \leq  \exp\left(  -(9/8)^{k-k_{0}} \right)<1\,.
\end{equation}
Finally, applying first \eqref{equa:gapAbelian_Aux1}, and later \eqref{equa:gapAbelian_Aux2} we get
\begin{align*} 
(1-\delta_{0})^{2}\prod_{k=0}^{\infty}(1-\delta_{k}) \, & = (1-\delta_{0})^{2}\prod_{0 \leq k \leq k_{0}}(1-\delta_{k}) \prod_{k > k_{0}} (1-\delta_{k}) \\[2mm]
& \geq e^{-12\beta} e^{-6\beta (k_{0}+1)}  \, \prod_{k > k_{0}} (1-\delta_{k})\\[2mm] 
& \geq \, e^{-54 \beta^{2} -72 \beta}(3+\ln|G|)^{-54 \beta } \, \prod_{k \geq 1}(1-e^{-(9/8)^{k}})\,, 
\end{align*}
where we have used that $k_{0}+1 \leq 9 \beta + 9 \ln(3+\ln|G|)+10$.
Applying these inequalities to \eqref{equa:TheoGap2Abelian} we conclude the result.
\end{proof}

\subsection{Spectral gap: General group case}
In the general case, we can establish the following bound.
\begin{Theo}\label{theo:CorrelationDecayGeneralQDM}
Let $G$ be an arbitrary group,  let $\sigma = \sigma_{\beta}$ be the Gibbs state associated to the quantum double model and let $\mu_{\beta} := \lceil 2^{9}e^{\beta}(1+\ln{|G|}) \rceil$. Consider a partition of $\Lambda$ into four disjoint subsets $\Lambda=ABCD$ and an integer $\ell \geq \mu_{\beta}$ satisfying one of the following configurations:
\begin{enumerate}
\item[$(i)$] $\Lambda=AB_{1}CB_{2}$ as in Assumption~\ref{Assumtion2Dtorus}.$(i)$, where $B=B_{1} B_{2}$ and $D=\emptyset$. 
\item[$(ii)$] $\Lambda = AB_{1}CB_{2}D$ as in Assumption~\ref{Assumtion2Dtorus}.$(ii)$, where \mbox{$B=B_{1}B_{2}$}. 
\item[$(iii)$] $\Lambda = ABCD$ as in Assumption~\ref{Assumtion2Dtorus}.$(iii)$..  
\end{enumerate}
Then, the following inequality holds
\begin{equation}\label{equa:CorrelationDecayGeneralQDM}
\Delta_{\sigma}(A:C|D) \leq  \delta(\ell):=\left( \frac{\gamma_\beta}{1+\gamma_\beta}\right)^{\ell^{2}/2}\,. 
\end{equation}
\end{Theo}

\begin{proof}
Let us consider any of the three possible decompositions ($i$)-($iii$). Observe that in the three cases, if $\cR \in \{ AB, BC, ABC\}$, since these are regions connected by stars and plaquettes, we can apply Theorem~\ref{thm:marginals-decayQDM} to decompose
\begin{equation}\label{equa:CorrelationDecayGeneralQDMaux1} 
\Tr_{\cR}(e^{-\beta H_{\cR}^\partial}) = \kappa_{\cR} (\mathbbm{1} + Q_{\partial \cR}) \quad  \text{ where } \quad \|Q_{\partial \cR}\|_{\infty} \leq (|G|^{2}-1) \left( \frac{\gamma_{\beta}}{ 1+ \gamma_{\beta}}\right)^{m_\mathcal{R}}\,.
\end{equation}
To deal with the case  $\cR=B$, note that in the configuration ($iii$), $B$ is also connected by stars and plaquettes, so the same estimate \eqref{equa:CorrelationDecayGeneralQDMaux1}  applies in this case. However, in the configurations ($i$) and ($ii$), $B$ is the union of two subsets $B_{1}$ and $B_{2}$ each of which is connected by stars and plaquettes. Since, by hypothesis, $\cS_{B_{1}} \cap \cS_{B_{2}} = \emptyset$ and $\cP_{B_{1}} \cap \cP_{B_{2}} = \emptyset$, we can decompose
\[ \Tr_{B}(e^{-\beta H_{B}^\partial}) = \Tr_{B_{1}}(e^{-\beta H_{B_1}^\partial}) \, \Tr_{B_2}(e^{-\beta H_{B_2}^\partial})  \,, \]
and apply next Theorem~\ref{thm:marginals-decayQDM} to further rewrite
\[ \Tr_{B}(e^{-\beta H_{B}^\partial}) = \kappa_{B_{1}} \kappa_{B_{2}} (\mathbbm{1} + Q_{\partial B_{1}}) (\mathbbm{1} + Q_{\partial B_{2}}) = \kappa_{B}(\mathbbm{1} + Q_{\partial B})\,, \]
where $\kappa_{B}:=\kappa_{B_{1}} \kappa_{B_{2}}$ and $Q_{\partial B} := Q_{\partial B_{1}} + Q_{\partial B_{2}} + Q_{\partial B_{1}} Q_{\partial B_{2}}$. We can moreover estimate, denoting $m_{B}:=\min\{ m_{B_{1}}, m_{B_{2}}\}$:
\[ 
\|Q_{\partial B} \|_{\infty}  \leq 2  (|G|^{2}-1) \left( \frac{\gamma_{\beta}}{ 1+ \gamma_{\beta}}\right)^{m_{B}} + (|G|^{2}-1)^{2} \left( \frac{\gamma_{\beta}}{ 1+ \gamma_{\beta}}\right)^{2m_B} \leq |G|^{4} \left( \frac{\gamma_{\beta}}{ 1+ \gamma_{\beta}}\right)^{m_{B}}\,.
\]
Therefore, for every $\cR \in \{AB, BC, ABC,B\}$ in any configuration ($i$)-($iii$) we have
\begin{equation}\label{equa:CorrelationDecayGeneralQDMaux1.5}  
\| Q_{\partial \cR}\|_{\infty} e^{2 |\partial \cR| \| \Phi\|} \leq \| Q_{\partial \cR}\|_{\infty} e^{8 |\partial \cR| } \leq |G|^{4} \left( \frac{\gamma_{\beta}}{1+\gamma_{\beta}}\right)^{m_\cR} e^{8 |\partial \cR| }\,. 
\end{equation}
Now, we claim that the condition $\ell \geq \mu_{\beta}$ yields that
\begin{equation}\label{equa:CorrelationDecayGeneralQDMaux2} 
\| Q_{\partial \cR}\|_{\infty} e^{2 |\partial R| \| \Phi\|} \leq \varepsilon:=\frac{1}{e^2} \left( \frac{\gamma_{\beta}}{1+\gamma_{\beta}}\right)^{\ell^{2}/2} \,. 
\end{equation}
Before proving this claim, let us explain how this helps to conclude \eqref{equa:CorrelationDecayGeneralQDM} finishing the proof. Indeed, we just need to apply Theorem~\ref{Theo:SufficientCondition1}, whose four hypothesis ($i$)-($iv$) clearly hold: The first hypothesis holds since $A$ and $C$ are separated by at least $\ell \geq 2$ plaquettes; the second and fourth hypothesis are satisfied by \eqref{equa:CorrelationDecayGeneralQDMaux1} and the $\varepsilon$ given in claim \eqref{equa:CorrelationDecayGeneralQDMaux2}, and the third hypothesis $\kappa_{ABC} \kappa_{B} = \kappa_{AB} \kappa_{BC}$ holds in the three possible configurations of $ABCD$,  as we already checked in Proposition~\ref{Theo:CorrelationDecayAbelianQDM}. Therefore, Theorem~\ref{Theo:SufficientCondition1} yields that 
\[ \Delta_{\sigma}(A:C|D) \leq \frac{4 (1/e^2)}{(1-1/e^2)^2} \left( \frac{\gamma_{\beta}}{1+\gamma_{\beta}} \right)^{\ell^2/2} \leq \left( \frac{\gamma_{\beta}}{1+\gamma_{\beta}} \right)^{\ell^2/2}\,, \]
and so \eqref{equa:CorrelationDecayGeneralQDM} holds.

Let us then demonstrate claim \eqref{equa:CorrelationDecayGeneralQDMaux2} for every $\cR \in \{AB, BC, B, ABC\}$ in any configuration ($i$)-($iii$). First, observe that $\cR \in \{ ABC,AB,BC,B\}$ can be either a rectangle or a union of two rectangles:
\begin{itemize}
\item If $\cR$ is a cylinder or a rectangle of dimensions $a,b$, then $m_{\cR} \geq ab$ whereas $|\partial \cR| \leq 4(a+b)$.
\item If $\cR$ is a union of two rectangles $B_{1}$ and $B_{2}$ as in configurations ($ii$) and ($iii$), each of which has dimensions $a,b$, then $m_{\cR} \geq  ab$ and $|\partial \cR| \leq |\partial B_{1}| + |\partial B_{2}| \leq 8(a+b)$. 
\end{itemize}
Consequently, for every $\cR$, since $a,b \geq \ell \geq \mu_{\beta}$, we can then estimate
\[ m_{\cR} \geq \frac{ab}{8(a+b)} \, |\partial \cR| = \frac{1}{8} \frac{1}{\frac{1}{a} + \frac{1}{b}} \, |\partial \cR| \geq \frac{\mu_{\beta}}{16} \, |\partial \cR|\,. \]
Next, let us rewrite \eqref{equa:CorrelationDecayGeneralQDMaux1.5} as
\begin{equation*} 
\| Q_{\partial \cR}\|_{\infty} e^{2|\partial \cR| \| \Phi\|}  \leq \left( \frac{\gamma_{\beta}}{1+\gamma_{\beta}} \right)^{\frac{m_\cR}{2}} \left( \frac{\gamma_{\beta}}{1+\gamma_{\beta}}\right)^{\frac{m_\cR}{4}} \, |G|^{4} \, \left( \frac{\gamma_{\beta}}{1+\gamma_{\beta}} \right)^{\frac{m_\cR}{4}} \, e^{8|\partial \cR|}\,.
\end{equation*}
We use now that $(1-1/x) \leq e^{-1/x}$ for every $x>1$. This implies that
\[ \left( \frac{\gamma_{\beta}}{1+\gamma_{\beta}} \right)^{\frac{m_\cR}{4}} = \left( 1-\frac{1}{1+\gamma_{\beta}} \right)^{\frac{m_\cR}{4 }}  \leq e^{-\frac{m_\cR}{4 (1+\gamma_\beta)}} \leq e^{-\frac{m_{\cR}}{4e^{\beta}}}\leq e^{-\frac{|\partial \cR| \mu_{\beta}}{2^6 e^{\beta}} }\,, \]
and therefore
\[ \| Q_{\partial R}\|_{\infty} e^{2|\partial R| \| \Phi\|}  \leq  \left( \frac{\gamma_{\beta}}{1+\gamma_{\beta}} \right)^{\frac{m_\cR}{2}}  e^{-\frac{\mu_{\beta}}{4 e^\beta}} \, |G|^{4} \, e^{-\frac{|\partial \cR| \mu_{\beta}}{2^{6} e^\beta}} \, e^{8|\partial \cR|}\,. \]
% \[ \| Q_{\partial R}\|_{\infty} e^{2|\partial R| \| \Phi\|} \leq \left( \frac{\gamma_{\beta}}{1+\gamma_{\beta}} \right)^{\frac{m_\cR}{2}} e^{-\frac{|\partial \cR|\mu_\beta }{32 e^{\beta}}}|G|^{4} e^{8|\partial \cR|} = \left( \frac{\gamma_{\beta}}{1+\gamma_{\beta}} \right)^{\frac{m_\cR}{2}} e^{-\frac{|\partial \cR| \mu_{\beta}}{2^{5} e^{\beta}}   + 4 \ln |G| + 8 |\partial \cR|} \,. \]
Thus, using that $\mu_{\beta} \geq 2^{9}e^{\beta}(1+\ln{|G|})$, this guarantees that 
\[ \| Q_{\partial R}\|_{\infty} e^{2|\partial R| \| \Phi\|} \leq  \left( \frac{\gamma_{\beta}}{1+\gamma_{\beta}} \right)^{\frac{m_\cR}{2}} e^{-2} \leq   \left( \frac{\gamma_{\beta}}{1+\gamma_{\beta}} \right)^{\frac{\ell^2}{2}} e^{-2}\,. \]
where we used that $m_{\cR} \geq \ell^2$ for every $\cR \in \{ ABC,AB,BC,B\}$ in any of the possible configurations ($i$)-($iii$).  This concludes the proof of the claim, and so the proof of theorem.
\end{proof}

\begin{Coro}\label{Coro:gapNonAbelian}
Let $G$ be a group, let $\sigma = \sigma_{\beta}$ be the Gibbs state associated to the quantum double model, and let $\mathbf{H}_{\Lambda_N}$ be the canonical purified Hamiltonian of $\sigma$. Denote $ \mu_{\beta}:= \lceil 2^{9}e^{\beta}(1+\ln{|G|})\rceil $. Then, for every $N \geq 2$ we have 
\begin{equation} \label{equa:coro-gapNonAbelian}\gap(\mathbf{H}_{\Lambda_{N}}) \geq  e^{-11} (1-e^{-1})^{2} e^{-\beta 2^{18} \mu_{\beta}^2}\,  \, \prod_{k \geq 0}\left( 1-e^{-(9/8)^{k}}\right) \,. \end{equation}
where the last infiniteproduct converges to a positive constant independent of $N$ and $\beta$.
\end{Coro}

\begin{proof}
Let us fix $\mu=2^6\mu_{\beta}^{2}$ where $ \mu_{\beta}= \lceil 2^{9}e^{\beta}(1+\ln{|G|})\rceil $ is given in Theorem~\ref{theo:CorrelationDecayGeneralQDM}. Observe that, by Remark~\ref{Rema:etaProperties}.($iv$), for every rectangle $\cR  \in \mathcal{F}_{\mu}$  we can estimate
\[\eta_{\cR}(\sigma)^{4} \leq e^{2^2\beta |\cR| \| H-H_{R^{c}}\|} \leq e^{2^4\beta   |\cR|} \leq e^{\beta 2^6\mu^{2}} \leq e^{\beta 2^{18} \mu_{\beta}^2}\,.\] 
Thus, if $N \leq \mu=2^{6}\mu_{\beta}^{2}$ we can simply use Theorem~\ref{Theo:roughBoundSpectralGap} to show that 
\[\operatorname{gap}(\mathbf{H}_{\Lambda_N}) \geq e^{-\beta 2^{18} \mu_{\beta}^2},\]
and so inequality \eqref{equa:coro-gapNonAbelian} clearly holds. Let us now assume that $N> \mu=2^{6}\mu_{\beta}$. We now apply Theorem~\ref{Theo:gap2} with these $N$ and $\mu$, and with $\delta(\ell)$ given in Theorem~\ref{theo:CorrelationDecayGeneralQDM}, obtaining that
\begin{equation}\label{equa:Coro-gapNonAbelian} 
\gap(\mathbf{H}) \geq e^{-11} (1-\delta_{0})^2 \, \left(\prod_{k=1}^{\infty} (1 - \delta_{k})\right) e^{-\beta 2^{18} \mu_{\beta}^2} \,, 
\end{equation}
where for every $k \geq 0$
\[ \delta_{k} = \delta(\lfloor \tfrac{\sqrt{\mu}}{8}(9/8)^{k/2} \rfloor) = \delta(\lfloor \mu_{\beta} (9/8)^{k/2} \rfloor) \leq \left( \frac{\gamma_\beta}{1+\gamma_\beta} \right)^{\frac{1}{2}\lfloor \mu_{\beta}(9/8)^{k/2}\rfloor^{2}} \leq \left( \frac{\gamma_\beta}{1+\gamma_\beta} \right)^{\frac{\mu_{\beta}^{2}}{8}(9/8)^{k}}  \,. \]
We use now that $(1-1/x) \leq e^{-1/x}$ for every $x>1$. This yields
\[ \delta_{k} \leq \exp\left( -\tfrac{\mu_{\beta}^{2}}{8(1+\gamma_{\beta})} (9/8)^{k}\right) \leq \exp\left( -\tfrac{\mu_{\beta}^{2}}{8 e^{\beta}} (9/8)^{k}\right) \leq e^{-(9/8)^{k}}\,. \]
Applying this upper bound to \eqref{equa:Coro-gapNonAbelian}, we conclude the result.
\end{proof}

\section{Conclusions and Final Comments}

In this work, we have introduced the purified canonical Hamiltonian $\mathbf{H}$ as a tool to connect dynamical properties of a reversible dissipative semigroup (its spectral gap, controlling the convergence rate and mixing time),  with static properties of its invariant state (the decay of $\Delta_\sigma(A:C|D)$). To conclude, let us now comment on some open questions and possible directions for future work that arise from our contribution.

\paragraph{Optimality of spectral gap estimates.}

A natural question to put forward is whether the spectral gap estimates for Davies generators we obtained are in any sense sharp or optimal, and more specifically, whether their dependence on the inverse temperature $\beta$ is the best possible. To the best of our knowledge, the optimal scaling of the gap as the temperature goes to zero is not known. Previous results \cite{Alicki2009} showed that for both the 1D quantum Ising model and the 2D toric code model a scaling as $\order{\exp(-C\beta)}$, while our estimates obtained in Corollaries~\ref{coro:ising-model-gap} and~\ref{Coro:gapAbelian} only scale as $\order{\exp(-C\beta^2)}$, making them not optimal in this sense. In \cite{Kmr2016}, a spectral gap bound for Davies generators of any abelian quantum double model was obtained: although this estimate is not system-size independent (it scales as inverse-polynomial in the system size), its dependence on $\beta$ is once again a simple exponential. We leave as an open question whether this is the correct $\beta$-dependence of the gap, as $\beta$ goes to infinity, and whether this also applies to non-abelian models, for which (see Corollary~\ref{Coro:gapNonAbelian}) we can only prove a scaling of $\order{\exp(\exp(-C\beta))}$ (a double exponential), in line with the scaling we first obtained in \cite{Lucia2023}.

\paragraph{Non-commuting models.}
As we mentioned in the introduction, the definition of the purified parent Hamiltonian $\mathbf{H}$ and the connection between its spectral gap and the mixing condition $\Delta_\sigma(A:C|D)$ can be made even for the case of states $\sigma$ which are positive-temperature Gibbs states of non-commuting Hamiltonians. Under this setup, Davies generators can be defined, and by requiring the thermal bath to satisfy Assumption~\ref{assumption:davies-1}, we obtain the a lower bound on their gap given by Proposition~\ref{prop:davies-local-primitivity}, which depends not only on the gap of $\mathbf{H}$ but also on a ``local gap'' of the Davies generator $\mathcal{D}_x$. Unfortunately, in this case, the Davies generators are not necessarily finite-range, and therefore we do not know how to estimate the gap of $\mathcal{D}_x$ as we did in the case of Gibbs states of commuting models (in Proposition~\ref{prop:local-davies-gap-commuting}). This raises the question of whether it is possible to either estimate the gap of $\mathcal{D}_x$ for these cases, or if it is possible to find other finite-range QMS generators which can be compared with $\mathbf{H}$.

While we do not the answer of either these problems, we would like to mention a recent construction of a (quasi-)local QMS generator having a Gibbs state of a non-commuting model as an invariant state~\cite{2311.09207}. The key difference is that, contrary to the Davies generators, these generators do \emph{not} satisfy the stronger GNS version of detailed balance, the one with parameter $s\neq 1/2$ (see Section~\ref{sec:OpenQuantumSystems}), but only the weaker one with parameter $s=1/2$, also known as KMS-reversibility.
Therefore, a third solution to the problem of analyzing QMS generators for non-commuting models would be to extend 
 our results to the case of KMS-reversible generators. We do not know the answer to this question, but let us explain the obstacle in analyzing KMS-reversible generators with our approach.

Our construction of the canonical purified Hamiltonian is done in terms of projections onto the subspaces 
\[ W_{X} = \{ (\mathbbm{1}_{X} \otimes O) \sigma^{1/2} \colon O \in \mathcal{B}(\mathcal{H}_{X^{c}}) \}\,, \]
which naturally arise from considering locally primitive generators satisfying the GNS-reversibility.
In the weaker case of KMS-reversible generators, we could have made a similar construction starting from the subspaces
\[ \tilde{W}_{X} = \{ \sigma^{1/4} (\mathbbm{1}_{X} \otimes O)  \sigma^{1/4} \colon O \in \cB(\cH_{X^{c}}) \}\,. \]
If we repeat the arguments from  the proof of Lemma~\ref{Prop:explicitprojectionCanonicalPurified}, we get that the orthogonal projection $\tilde{\Pi}_{X}(Q)$ onto $\tilde{W}_X$ is characterized by 
\begin{equation}
    \label{eq:kms-symmetric-projection}
    \Tr_{X}(\sigma^{1/4}Q\sigma^{1/4}) = \Tr_{X}(\sigma^{1/4} \tilde{\Pi}_{X}(Q)\sigma^{1/4})\,.
\end{equation}
While this approach could in principle work, it seems difficult to obtain an explicit formula for $\tilde{\Pi}_{X}(Q)$ from this expression, as it was done for $\Pi_X$ in \eqref{eq:explicitprojectionCanonicalPurified}. Therefore we do not know if it is possible to obtain an explicit characterization of the martingale condition for the projections $\tilde{\Pi}_X$ in terms of a spatial mixing condition on $\sigma$, in the spirit of Theorem~\ref{Theo:martingaleVScorrelations}.

In \cite{KB_2016}, only the weaker $1/2$-detailed balance condition is required, and in fact the authors can also prove a spectral gap estimate for a generator which is \emph{not} GNS-symmetric. This is the case of the so-called \emph{heat-bath} generator, given by $\cD_x(Q) = \mathbbm{E}_x(Q) - Q$, where $\mathbbm{E}_x(Q)$ is the Petz recovery map associated to region $X$:
\[
\mathbbm{E}_X(Q) = \sigma_{X^{c}}^{-1/2} \Tr_X[ \sigma^{1/2} Q   \sigma^{1/2}]\sigma_{X^{c}}^{-1/2}.
\]
Since the Petz recovery map satisfies detailed balance only for $s=1/2$, but not for $s\neq 1/2$, its purification with respect to the GNS scalar product does not yield a self-adjoint operator, and consequently we are unable to compare it to the canonical purified Hamiltonian $\mathbf{H}$.

To conclude this comparison, we observe that if we apply the inverse of the purification procedure of Section~\ref{sec:OpenQuantumSystems} to the purified canonical Hamiltonian, i.e.,
 if we consider an operator $\cK_X : \cB(\cH) \to \cB(\cH)$ defined by
\[
\Pi_X(Q) = -\cK_X(Q\sigma^{-1/2})\sigma^{1/2},
\]
then a simple algebraic manipulation gives us that
\[
\cK_X(Q) = \Tr_X[Q\sigma] \sigma_{X^c}^{-1} - Q := T_X(Q) - Q \,.
\]
The operator $T_X(Q) = \Tr_X[Q\sigma] \sigma_{X^c}^{-1}$ has a clear resemblance to the Petz recovery map. There is a crucial difference though: the Petz recovery map is a completely positive map, and so it can be used to construct a QMS generator, while $T_X$ does not need to be, and in particular the operator $\cK_X$ defined in this way does not necessarily define a QMS.

\paragraph{Hopf algebra quantum double models.}

The quantum double models presented in Section~\ref{sec:quantum-double} can be further generalized, by replacing the group algebra $\ell_2(G)$ with an appropriate abstract algebraic structure. It is believed that models defined starting from a weak Hopf $C^*$-algebra \cite{buerschaper_hierarchy_2013,molnar_matrix_2022} include a representative of every possible non-chiral 2D topological ordered phase of matter. We expect that the analysis of the marginals of the quantum double models carried out in this work can be extended to this more abstract setting (although it might not be straightforward to rewrite the results of Section~\ref{sec:quantum-double} purely in terms of the algebraic structure of a weak $C^*$-Hopf algebra), and that our results can be used to show that the Davies generator for this larger class of models is always gapped at any finite temperature, confirming the commonly held believe that there is no thermally stable 2D topological order.

\paragraph{Log-Sobolev estimates.}
As we have mentioned, the estimates on the spectral gap of the generator of a quantum Markov semigroup immediately translate, via standard techniques, to bounds on the \emph{mixing time}, i.e., the time it takes in the worst case scenario for the semigroup to reach a small ball around its fixed point. In the setting of quantum spin systems on a lattice that we are considering, the system-size independent spectral gap bound which we obtain becomes a polynomial (in the system size) bound on the mixing time. We expect that this bound is not sharp, and in fact we know this is the case for 1D models with commuting interactions \cite{Bardet2024, Bardet2023}: in this case, it is possible to prove a stronger condition, called a (modified) log-Sobolev inequality, that implies that the mixing time scales only logarithmically in the size of the system. We believe that a log-Sobolev inequality holds also for the case of the 2D quantum double models (as suggested by the absence of a positive-temperature phase transition, given that we can prove a constant spectral gap for every positive temperature), but we are currently unable to prove it even for the simpler case of abelian models. One indication that this should be the case is the similarity between the spatial mixing condition given by the decay of $\Delta_\sigma(A:C|D)$ and the mixing condition used in the proofs of the log-Sobolev inequality (see for example \cite{Bardet2024, Bardet2021}).

\printbibliography

\appendix

\newpage

\section{Proof of the 1D general result}\label{sec:Proof1DGeneral}

Recall that, along this section, $\Lambda = \mathbb{Z}_{N}$ and $\Phi$ is a local interaction on $\Lambda = \mathbb{Z}_{N}$ with finite range $r>0$ and strength $\| \Phi\| \leq J$ for some $J>0$. We also denote by $H$ the corresponding Hamiltonian and by $\sigma = e^{-H}/\Tr(e^{-H})$ the Gibbs state.

\subsection{Locality estimates}

For every pair of operators $O_{1}, O_{2}$, we denote the adjoint action of $e^{-O_{1}}$ on $O_{2}$ as
\[ \Gamma(O_{2};O_{1}) := e^{-O_{1}}O_{2}e^{O_{1}}\,. \]
Additionally, we introduce the following Araki expansionals \cite[Theorem 3]{araki73} given by
\begin{align}\label{equa:ExpansionalAraki}
E(O_{2}; O_{1}) &:= e^{O_{2}-O_{1} } e^{O_{1}} = \operatorname{Exp}_{r}\left( \int_{0}^{1} ds \, \Gamma(O_{2};sO_{1})\right)\,,\\
E(O_{2}; O_{1})^{-1} & := e^{-O_{1}} e^{O_{1}-O_{2}} = \operatorname{Exp}_{l}\left( \int_{0}^{1} ds \, \Gamma(-O_{2};sO_{1})\right)\,. 
\end{align}
We next formulate the main locality property of these elements when $O_{1}$ is a finite-range Hamiltonian and $O_{2}$ is an operator supported in a certain subset of the lattice. Its proof follows from \cite[Theorem 2.3 and Section 2.2.1]{perez23}, and extends Araki's original result from \cite{araki73}.

\begin{Theo}\label{Theo:localityProp}
There exists a constant $\mathcal{G} = \mathcal{G}(r, J) \geq 1$, independent of the system size, such that for every observable $Q$ supported on a set $Z$ which is union of $k$ intervals $Z=\cup_{j=1}^{k}[a_{j},b_{j}]$, if we denote $Z_{n} = \cup_{j=1}^{k}[a_{j} - n, b_{j}+n]$ for each $n \geq 0$, then for every  $0 \leq n \leq m$:
\begin{align*}
\|\Gamma(Q;H_{Z_n})\| & \leq \| Q\| \,\mathcal{G}^{k|Z|}\,,\\[2mm]
\| \Gamma(Q;H_{Z_m}) - \Gamma(Q;H_{Z_n}) \| & \leq \| Q\| \, \mathcal{G}^{k|Z|} \, \frac{\mathcal{G}^{n}}{(\lfloor n/r\rfloor +1)! }\,.
\end{align*}
As a consequence, for the Araki expansional,
\begin{align*} 
\| E(Q;H_{Z_n})\|  \, , \, \| E(Q;H_{Z_n})^{-1}\| & \leq \exp{\mathcal{G}^{k|Z|} \| Q\|}\,,\\[2mm]
\| E(Q;H_{Z_{m}}) - E(Q;H_{Z_{n}}) \| \, , \, \| E(Q;H_{Z_{m}})^{-1} - E(Q;H_{Z_{n}})^{-1} \|  & \leq \exp{\mathcal{G}^{k|Z|}\| Q\|} \frac{\mathcal{G}^{n}}{(\lfloor n/r\rfloor  + 1)!}\,. 
\end{align*}
\end{Theo}

\begin{Rema}
For any $X \subset \Lambda$,we can actually replace $H_{Z_{n}}$ with $H_{Z_{n} \cap X}$ in all the above inequalities. To see this, we define a new modified local interaction $\Phi'$ by setting $\Phi_{Y}' = \Phi_{Y}$ if $Y \subset X$ and $\Phi'_{Y} = 0$ otherwise. Note that $\Phi'$ also has finite range $r$ and strength $\| \Phi'\| \leq \| \Phi\| \leq J$, so Theorem~\ref{Theo:localityProp} applies equally to $\Phi'$. Let $H'$ be the Hamiltonian associated to $\Phi'$. Then, by construction, we have $H'_{Z_{n}} = H_{X \cap Z_{n}}$. We will sometimes use this remark when referring to the theorem.
\end{Rema}

\subsection{Proof of Proposition~\ref{Prop:eta1D}}

Fix any interval $I \subset \Lambda$. Taking $Q = e^{-\frac{1}{2}H_{I^c}}/\Tr(e^{-H_\Lambda})$ which belongs to $\cB_{I^c}$, we have by the very definition of $\eta_{I}(\sigma)$ given in Theorem~\ref{Theo:roughBoundSpectralGap}, that 
\begin{align*} 
\eta_{I}(\sigma) 
& \leq \| Q \sigma^{-1/2}\| \cdot \|  \sigma^{1/2}Q^{-1}\| \\[2mm]
& = \left\| e^{-\frac{1}{2}H_{I^c}} e^{\frac{1}{2}H}\right\| \cdot \left\| e^{-\frac{1}{2}H} e^{\frac{1}{2}H_{I^c}} \right\|\\[2mm]
& = \left\| e^{\frac{1}{2}H_{I}}e^{-\frac{1}{2}(H_{I}+H_{I^c})} e^{\frac{1}{2}H}\right\| \cdot \left\| e^{-\frac{1}{2}H} e^{\frac{1}{2}(H_{I}+H_{I^c})} e^{\frac{1}{2}H_{I}}\right\|\\[2mm]
& \leq e^{\| H_{I}\|} \, \left\| e^{-\frac{1}{2}(H_{I}+H_{I^c})} e^{\frac{1}{2}H}\right\| \, \left\|e^{-\frac{1}{2}H} e^{\frac{1}{2}(H_{I}+H_{I^c})} \right\|\\[2mm]
& = e^{\| H_{I}\|} \, \left\| E(\tfrac{1}{2}(H - H_{I} - H_{I^{c}}), \tfrac{1}{2} H)\right\| \, \left\| E(\tfrac{1}{2}(H - H_{I} - H_{I^{c}}), \tfrac{1}{2}H)^{-1}\right\|\,.
\end{align*}
Finally, we apply Theorem~\ref{Theo:localityProp} to estimate the norm of the expansionals. Observe that, after canceling the terms $\Phi_X$ fully supported inside $I$ or $I^c$, the difference $H-H_{I}-H_{I^c}$ is supported in the union of two intervals of length $2r$ (around the boundaries between $I$ and $I^c$). Thus, we conclude
\[ \eta_{I}(\sigma) \leq e^{\| \Phi\| \, |I|} \exp(2\mathcal{G}^{8r} \tfrac{1}{2}\| H -H_{I} - H_{I^c}\|) \leq e^{J \, |I|}  \exp(\mathcal{G}^{8r} 2rJ)\,. \]
This finishes the proof.

\subsection{Proof of Theorem~\ref{Theo:1DCorrelationDecay}}

In view of Proposition~\ref{prop:estimatesMartinagleCondition}, Theorem~\ref{Theo:1DCorrelationDecay} will follow from the following stronger result:

\begin{Theo}\label{Theo:1DCorrelationDecayAux}
There exist constants $c=c(J,r,d), \alpha=\alpha(J,r,d) >0$, independent of the system size, such that for every partition $\Lambda=ABCD$ as in Figure~\ref{Fig:splitRingABCD}.(a) with $|B|, |D| \geq 3\ell$ for some real $\ell >0$, then 
\begin{align}
\label{equa:1DCorrelationDecayAux1} \| \sigma_{A} \sigma_{C} \sigma_{AC}^{-1} - \mathbbm{1} \|_{\infty} & \leq c e^{-\alpha \sqrt{\ell}} \\[2mm]
\label{equa:1DCorrelationDecayAux2}  \| (\sigma_{AD} \sigma_{D}^{-1} \sigma_{DC}) \sigma_{ADC}^{-1}  - \mathbbm{1}\|_{\infty}& \leq   c e^{-\alpha \sqrt{\ell}}\\[2mm] 
\label{equa:1DCorrelationDecayAux3}   \|  \sigma_{ADC}^{-1}(\sigma_{AD} \sigma_{D}^{-1} \sigma_{DC})  - \mathbbm{1}\|_{\infty}& \leq   c e^{-\alpha \sqrt{\ell}}
\end{align}
\end{Theo}

For the proof of the previous result, we will also assume that the intervals $B$ and $D$ are split into three consecutive intervals $B=B_{1}B_{2}B_{3}$ and $D=D_{3}D_{2}D_{1}$ as in Figure~\ref{Fig:splitRingABCD}.(b) being $|B_{i}|, |D_{i}| \geq \ell$ for $i=1,2,3$.

 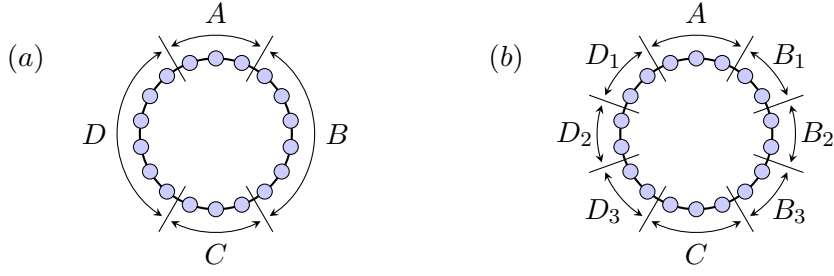
\begin{figure}[h]
 \centering
%Primer dibujo
\begin{subfigure}{.4\textwidth}
\centering
\begin{tikzpicture}
\node at (-2.5,1) {$(a)$};
% Definimos el número de nodos
    \def\nodes{18}
    % Definimos el radio del círculo
    \def\radius{1}
    % Definimos el ángulo para la división en regiones
    \def\angreg{40}
    \def\ang{60}
    
    \begin{scope}
    % Dibujamos la circunferencia
    \draw[thick] (0,0) circle(\radius);

    % Dibujamos los nodos del lattice circular
    \foreach \i in {1,...,\nodes} {
        % Calculamos la posición de los nodos en el círculo
        \pgfmathsetmacro{\angle}{360/\nodes * (\i - 1)+10}
        \node[draw, circle, fill=blue!20, inner sep=2pt] at (\angle:\radius) {};
    }

    % Marcamos las cuatro regiones cortando perpendicularmente la circunferencia con segmentos finos

    \draw[thin] (-120:{\radius-0.2}) -- (-120:{\radius+0.5});

    \draw[thin] (-60:{\radius-0.2}) -- (-60:{\radius+0.5});

     \draw[thin] (60:{\radius-0.2}) -- (60:{\radius+0.5});

     \draw[thin] (120:{\radius-0.2}) -- (120:{\radius+0.5});

    %Determinamos las regiones

    \draw[ stealth-stealth] (60+3:{\radius+0.3}) arc(60+3:120-3:{\radius+0.3});
    \node at (90:{\radius+0.6}) {$A$};

    \draw[ stealth-stealth] (-60+3:{\radius+0.3}) arc(-60+3:60-3:{\radius+0.3});
    \node at (0:{\radius+0.6}) {$B$};

     \draw[ stealth-stealth] (-120+3:{\radius+0.3}) arc(-120+3:-60-3:{\radius+0.3});
    \node at (270:{\radius+0.6}) {$C$};
    
    \draw[ stealth-stealth] (120+3:{\radius+0.3}) arc(120+3:240-3:{\radius+0.3});
    \node at (180:{\radius+0.6}) {$D$};

    \end{scope}
\end{tikzpicture}
\end{subfigure}%
\begin{subfigure}{.4\textwidth}
\centering
\begin{tikzpicture}
\node at (-2.5,1) {$(b)$};
% Definimos el número de nodos
    \def\nodes{18}
    % Definimos el radio del círculo
    \def\radius{1}
    % Definimos el ángulo para la división en regiones
    \def\angreg{40}
    \def\ang{60}
    
    \begin{scope}
    % Dibujamos la circunferencia
    \draw[thick] (0,0) circle(\radius);

    % Dibujamos los nodos del lattice circular
    \foreach \i in {1,...,\nodes} {
        % Calculamos la posición de los nodos en el círculo
        \pgfmathsetmacro{\angle}{360/\nodes * (\i - 1)+10}
        \node[draw, circle, fill=blue!20, inner sep=2pt] at (\angle:\radius) {};
    }

    % Marcamos las cuatro regiones cortando perpendicularmente la circunferencia con segmentos finos

    \draw[thin] (-160:{\radius-0.2}) -- (-160:{\radius+0.5});

    \draw[thin] (-120:{\radius-0.2}) -- (-120:{\radius+0.5});

    \draw[thin] (-60:{\radius-0.2}) -- (-60:{\radius+0.5});
    
    \draw[thin] (-20:{\radius-0.2}) -- (-20:{\radius+0.5});

    \draw[thin] (20:{\radius-0.2}) -- (20:{\radius+0.5});

     \draw[thin] (60:{\radius-0.2}) -- (60:{\radius+0.5});

     \draw[thin] (120:{\radius-0.2}) -- (120:{\radius+0.5});

     \draw[thin] (160:{\radius-0.2}) -- (160:{\radius+0.5});

    %Determinamos las regiones

    \draw[ stealth-stealth] (60+3:{\radius+0.3}) arc(60+3:120-3:{\radius+0.3});
    \node at (90:{\radius+0.6}) {$A$};

    \draw[ stealth-stealth] (20+3:{\radius+0.3}) arc(20+3:60-3:{\radius+0.3});
    \node at (40:{\radius+0.6}) {$B_{1}$};
    
    \draw[ stealth-stealth] (-20+3:{\radius+0.3}) arc(-20+3:20-3:{\radius+0.3});
    \node at (0:{\radius+0.6}) {$B_{2}$};

    \draw[ stealth-stealth] (-60+3:{\radius+0.3}) arc(-60+3:-20-3:{\radius+0.3});
    \node at (-40:{\radius+0.6}) {$B_{3}$};

     \draw[ stealth-stealth] (-120+3:{\radius+0.3}) arc(-120+3:-60-3:{\radius+0.3});
    \node at (270:{\radius+0.6}) {$C$};
    
    \draw[ stealth-stealth] (120+3:{\radius+0.3}) arc(120+3:160-3:{\radius+0.3});
    \node at (140:{\radius+0.6}) {$D_{1}$};
    
    \draw[ stealth-stealth] (160+3:{\radius+0.3}) arc(160+3:200-3:{\radius+0.3});
    \node at (180:{\radius+0.6}) {$D_{2}$};

    \draw[ stealth-stealth] (200+3:{\radius+0.3}) arc(200+3:240-3:{\radius+0.3});
    \node at (220:{\radius+0.6}) {$D_{3}$};
    \end{scope}
\end{tikzpicture}
\end{subfigure}
\caption{In $(a)$, the 1D ring is split into four consecutive intervas $ABCD$. In $(b)$, intervals $B$ and $D$ are split into three adjacent subintervals $B=B_{1}B_{2}B_{3}$ and $D=D_{1}D_{2}D_{3}$. Intervals $B_{1}$ and $D_{1}$ are adjacent to $A$, whereas $B_{3}$ and $D_{3}$ are adjacent to $C$.}
\label{Fig:splitRingABCD}
\end{figure} 

The rate of decay is determined by the following result proved by Kimura and Kuwahara and \cite{kimura25} for Gibbs states on finite chains:

\begin{Theo}
Let $\Phi$ be a local interaction with finite range $r >0$ and $\| \Phi\| \leq J$ for some $J > 0$, and let $\sigma = e^{-H}/\Tr(e^{-H})$ be the Gibbs state.  Then, there are constants $c', \alpha' >0$ such that
\[ \|\Tr(\sigma^{B}Q_{B_{1}} Q_{B_{3}}) - \Tr(\sigma^{B}Q_{B_{1}})\Tr(\sigma^{B}Q_{B_{3}})  \| \leq \xi(\ell)=c' e^{-\alpha' \sqrt{\ell}} \]
\end{Theo}

When the local interaction is translation-invariant, this estimate can be strengthened to exhibit exponential decay, as a consequence of a seminal result by Araki \cite{araki69}, see also \cite{Bluhm2022}. However, for the sake of generality, and since this subexponential rate of decay suffices for our purposes, we will rely on this last result.

Within this setting, we define a new local interaction $\widetilde{\Phi}$ by $\widetilde{\Phi}_{X} = \Phi_{X}$ if $X \subset A,B,C$ or $D$, and $\widetilde{\Phi}_{X} = 0$ otherwise. In other words, $\widetilde{\Phi}$ is obtained from $\Phi$ by supressing interactions between the intervals $A,B,C,D$. Note that $\widetilde{\Phi}$ also has finite range $r>0$ and $\|\widetilde{\Phi}\| \leq J$. Let us denote by $\widetilde{H}$ the corresponding Hamiltonian, and use $\widetilde{\sigma}^{X} = e^{-\widetilde{H}_{X}}/ \widetilde{Z}_{X}$ where $\widetilde{Z}_{X} = \Tr(e^{-\widetilde{H}_X})$ for the corresponding local Gibbs state. Note that
\[ \widetilde{H}_{ABCD} = H_{A} + H_{B} + H_{C} + H_{D} \quad , \quad \widetilde{\sigma}^{ABCD} = \sigma^{A} \otimes \sigma^{B} \otimes \sigma^{C} \otimes \sigma^{D}\,. \]
Let us define
\[ E_{ABCD}^{(AC)} := e^{-H_{ABCD}}e^{H_{A}+H_{B}+H_{C}+H_{D}} = e^{-H_{ABCD}} e^{\widetilde{H}_{ABCD}}\,. \]
We can prove the following.

%\begin{tcolorbox}[colframe=black, colback=gray!5]
%\textbf{Step 1}
\begin{Prop}[\bfseries Step 1]\label{Prop:CorrelationDecayStep1}  
With the above notation, the following equalities hold
\begin{equation}\label{equa:Step1_1} 
\begin{split}
\sigma_{AD} \sigma_{D}^{-1} \sigma_{DC} \sigma_{ADC}^{-1} = \Tr_{BC}\left( \widetilde{\sigma}^{BC} E_{ABCD}^{(AC)}\right) \cdot \Tr_{ABC}\left( \widetilde{\sigma}^{ABC}  E_{ABCD}^{(AC)}\right)^{-1}\\[2mm] 
\cdot \Tr_{AB}\left( \widetilde{\sigma}^{AB}  E_{ABCD}^{(AC)}\right)
\cdot \Tr_{B}\left( \widetilde{\sigma}^{B} E_{ABCD}^{(AC)}\right)^{-1}\,,
\end{split}
\end{equation}
\begin{equation}\label{equa:Step1_2} 
\begin{split}
\sigma_{A} \sigma_{C} \sigma_{AC}^{-1} 
= \Tr_{BCD}\left(\widetilde{\sigma}^{BCD}  E_{ABCD}^{(AC)}\right) \cdot \Tr_{ABCD}\left(\widetilde{\sigma}^{ABCD} E_{ABCD}^{(AC)}\right)^{-1} \\[2mm]
  \quad \cdot \Tr_{DAB}\left(\widetilde{\sigma}^{DAB} E_{ABCD}^{(AC)}\right)    \,     \Tr_{BD}\left(\widetilde{\sigma}^{BD}  E_{ABCD}^{(AC)}\right)^{-1} \,.
\end{split}
\end{equation} 
\end{Prop}
%\end{tcolorbox}

\begin{proof}
To verify \eqref{equa:Step1_1} , we denote $H=H_{ABCD}$ and manipulate the expressions by introducing factors supported outside the traced regions, which cancel appropriately. 
\begin{align*} 
\sigma_{AD} \sigma_{D}^{-1} \sigma_{DC} \sigma_{ADC}^{-1} 
& =  \Tr_{BC}\left( e^{-H}\right) \, \Tr_{ABC}\left( e^{-H}\right)^{-1} \, \Tr_{AB}\left( e^{-H}\right) \, \Tr_{B}\left( e^{-H}\right)^{-1}\\[2mm]
& = \Tr_{BC}\left( e^{-H} e^{H_{A} + H_{D}}\right) \, \Tr_{ABC}\left( e^{-H} e^{H_{D}}\right)^{-1}  \\[2mm]
& \hspace{3.5cm} \cdot \Tr_{AB}\left( e^{-H} e^{H_{C} + H_{D}}\right) \, \Tr_{B}\left( e^{-H}e^{H_{A} + H_{C} + H_{D}}\right)^{-1}\\[2mm]
& = \Tr_{BC}\left( e^{-H_{B} - H_{C}} E_{ABCD}^{(AC)}\right) \cdot \Tr_{ABC}\left( e^{-H_{A} -H_{B} + H_{C}} E_{ABCD}^{(AC)}\right)^{-1} \\[2mm]
& \hspace{3.5cm} \cdot  \Tr_{AB}\left( e^{-H_{A} -H _{B}} E_{ABCD}^{(AC)}\right) \cdot \Tr_{B}\left( e^{-H_{B}} E_{ABCD}^{(AC)}\right)^{-1}\\[2mm]
& = \Tr_{BC}\left( \tfrac{e^{-H_{B} - H_{C}}}{Z_{B}Z_{C}} E_{ABCD}^{(AC)}\right) \cdot \Tr_{ABC}\left( \tfrac{e^{-H_{A} -H_{B} - H_{C}}}{Z_{A} Z_{B} Z_{C}} E_{ABCD}^{(AC)}\right)^{-1} \\[2mm]
& \hspace{3.5cm} \cdot  \Tr_{AB}\left( \tfrac{e^{-H_{A} -H _{B}}}{Z_{A} Z_{B}} E_{ABCD}^{(AC)}\right) \cdot \Tr_{B}\left( \tfrac{e^{-H_{B}}}{Z_{B}} E_{ABCD}^{(AC)}\right)^{-1}\,.
\end{align*}
This proves \eqref{equa:Step1_1}. To see \eqref{equa:Step1_2} , we begin by rewriting 
\[ \sigma_{A} \sigma_{C} \sigma_{AC}^{-1} = \Tr_{BCD}(e^{-H}) \cdot \Tr_{ABCD}(e^{-H})^{-1} \cdot \Tr_{ABD}(e^{-H}) \cdot \Tr_{BCD}(e^{-H})^{-1}\,. \]
Then, applying the same strategy as above yields the desired identity.
\end{proof}

In view of equations \eqref{equa:Step1_1} and \eqref{equa:Step1_2}, it seems clear that to prove Theorem~\ref{Theo:1DCorrelationDecayAux} we need to analyze the properties of the operator $E^{(AC)}_{ABCD}$. Let us then define
\[ W^{(AC)} := H_{ABCD} - \widetilde{H}_{ABCD}  = H_{ABCD} - H_{A}-H_{B}-H_{C}-H_{D}\,. \] 
Note that cancelling summands $\Phi_{X}$ in the right hand-side of the last expression, we easily see that $W$ can be written as a sum of nonzero terms $\Phi_{X}$ where $X$ supported around the boundaries of $A$ and $C$. More specifically, if $A=[a_{1}, a_{2}]$ and $C=[c_{1}, c_{2}]$, let us denote 
\begin{align*}
\alpha_{m} & := [a_{1} - (r+m), a_{1}+(r+m)] \cup [a_{2}-(r+m), a_{2}+(r+m)]\,,\\[2mm]
\gamma_{m} & := [c_{1}-(r+m), c_{1} +(r+m)] \cup [c_{2}-(r+m), c_{2} +(r+m)]\,.
\end{align*}
Then, all the remaining nonzero summands $\Phi_{X}$ that form $W^{(AC)}$ are supported in $\alpha_{0}$ or $\beta_{0}$ (only one of them, if $\ell \geq r$). This allows us to decompose $W^{(AC)}$ into two sums, one $W^{(A)}$ supported in $\alpha_{0}$ and another one $W^{(C)}$ supported in $\gamma_{0}$, so that
\[ W^{(AC)} = W^{(A)} + W^{(C)}\,. \]
Moreover, since $|\alpha_0|, |\gamma_0| \leq 6r$, we can estimate
\[ \| W^{(A)}\| \,,\, \| W^{(C)}\| \leq 6r \| \Phi\| \leq 6r J\,. \]
Let us define for every $X \subset \Lambda = ABCD$ and $Y \in \{ A,C,AC\}$ the following expansional:
\begin{equation}\label{equa:defiEXY} 
E_{X}^{(Y)} := E(-W^{(Y)}; \widetilde{H}_{X}) = e^{-(W^{(Y)} + \widetilde{H}_{X} )}e^{\widetilde{H}_{X}}\,. 
\end{equation}
According to the locality estimates from Theorem~\ref{Theo:localityProp}, these expansionals should be localized around the support of $W^{(Y)}$, that is, around the boundary of $A$ and/or $C$.

%\begin{tcolorbox}[colframe=black, colback=gray!5]
% \textbf{Step 2:}
\begin{Prop}[\bfseries Step 2]\label{Prop:CorrelationDecayStep2} 
 With the previous notation, assume that $\ell \geq r$. Then, we can approximately factorize each term of the right hand-side of \eqref{equa:Step1_1}, ommitting the inverses, as follows:
\begin{align} 
\label{equa:approxStep2_1_1}\Tr_{BC}\left( \widetilde{\sigma}^{BC} E_{ABCD}^{(AC)}\right) & \approx \Tr_{B}\left( \widetilde{\sigma}^{B} E_{ABCD}^{(A)}\right) \, \Tr_{BC}\left( \widetilde{\sigma}^{BC}  E_{ABCD}^{(C)}\right) \\[2mm]
\label{equa:approxStep2_1_2}
\Tr_{ABC}\left( \widetilde{\sigma}^{ABC}  E_{ABCD}^{(AC)}\right) & \approx \Tr_{AB}\left( \widetilde{\sigma}^{AB} E_{ABCD}^{(A)}\right)  \, \Tr_{BC}\left( \widetilde{\sigma}^{BC} E_{ABCD}^{(C)}\right)\\[2mm] 
\label{equa:approxStep2_1_3}
\Tr_{AB}\left( \widetilde{\sigma}^{AB} E_{ABCD}^{(AC)}\right) & \approx \Tr_{AB}\left( \widetilde{\sigma}^{AB}  E_{ABCD}^{(A)}\right) \, \Tr_{B}\left( \widetilde{\sigma}^{B} E_{ABCD}^{(C)}\right) \\[2mm]
\label{equa:approxStep2_1_4}
\Tr_{B}\left( \widetilde{\sigma}^{B} E_{ABCD}^{(AC)}\right) & \approx \Tr_{B}\left( \widetilde{\sigma}^{B} E_{ABCD}^{(A)}\right) \, \Tr_{B}\left( \widetilde{\sigma}^{B} E_{ABCD}^{(C)}\right)
\end{align}
We can also approximately factorize each term of the right hand-side of  \eqref{equa:Step1_2},  ommitting the inverses, as follows: 
\begin{align} 
\label{equa:approxStep2_2_1}
\Tr_{BCD}\left(\widetilde{\sigma}^{BCD} E_{ABCD}^{(AC)}\right)  
& \approx \Tr_{BD}\left(\widetilde{\sigma}^{BD} E_{ABCD}^{(A)}\right) \Tr_{BCD}\left(\widetilde{\sigma}^{BCD} E_{ABCD}^{(C)}\right) \\[2mm]
\label{equa:approxStep2_2_2}
\Tr_{ABCD}\left(\widetilde{\sigma}^{ABCD} E^{(AC)}_{ABCD}\right)  
& \approx \Tr_{ABD}\left(\widetilde{\sigma}^{ABD} \, E_{ABCD}^{(A)}\right) \, \Tr_{BCD}\left(\widetilde{\sigma}^{BCD} \, E_{ABCD}^{(C)}\right) \\[2mm]
\label{equa:approxStep2_2_3}
\Tr_{ABD}\left(\widetilde{\sigma}^{ABD} E^{(AC)}_{ABCD}\right)  
& \approx \Tr_{ABD}\left(\widetilde{\sigma}^{ABD} \, E_{ABCD}^{(A)}\right) \, \Tr_{BD}\left(\widetilde{\sigma}^{BD} \, E_{ABCD}^{(C)}\right) \\[2mm]
\label{equa:approxStep2_2_4}
\Tr_{BD}\left(\widetilde{\sigma}^{BD} E^{(AC)}_{ABCD}\right)  
& \approx \Tr_{BD}\left(\widetilde{\sigma}^{BD} \, E_{ABCD}^{(A)}) \, \Tr_{BD}(\widetilde{\sigma}^{BD} \, E_{ABCD}^{(C)}\right) 
\end{align}
Indeed, the additive error (in the operator norm) of these approximations $x \approx y$ can be estimated in all these cases by 
\[\| x-y\| \leq \mathcal{G}''' \,\left(\xi(\ell) + \frac{\mathcal{G}^{\ell}}{(\lfloor \ell/r \rfloor)!}\right)\] 
for some constant $\mathcal{G}''' = \mathcal{G}'''(r,J,d)$, where recall that $d$ is the local dimension of the system.
\end{Prop}
%\end{tcolorbox}

The proof of Step 2 is based on three lemmas.

\begin{Lemm}\label{Lemm:splitTwo}
With the previous notation, for every $Y \in \{ A,C,AC\}$, every $X \subset \Lambda$ and every $0\leq n \leq m$ 
\begin{align*} 
\| E_{X}^{(Y)}\| \, , \, \| (E_{X}^{(Y)})^{-1}\| & \leq \exp{12rJ\mathcal{G}^{4(12r)}}\\[2mm]
\| E_{X \cap (\alpha_{n}\gamma_{n})}^{(Y)} - E_{X \cap (\alpha_{m}\gamma_{m})}^{(Y)}\|  & \leq \exp{12rJ\mathcal{G}^{4(12r)}} \frac{\mathcal{G}^{n}}{(\lfloor n / r\rfloor + 1)!}\,.
\end{align*}
Moreover, if $\ell \geq r$ we can estimate
\begin{align*}\label{Lemm:splitTwo1} 
\|E_{ABCD}^{(C)} - E_{B_{3}CD_{3}}^{(C)} 
\| \, , \, \|E_{ABCD}^{(A)} - E_{D_{1}AB_{1}}^{(A)} 
\| \, , \, \|E_{ABCD}^{(AC)} - E_{D_{1}AB_{1}}^{(A)} E_{B_{3}CD_{3}}^{(C)}\| & \leq   \frac{ 2 \, e^{12rJ \mathcal{G}^{12r}} \, \mathcal{G}^{\ell-r}}{(\lfloor \ell/r \rfloor)!}\,,
\end{align*}
where $\cG$ is the constant from Theorem~\ref{Theo:localityProp}.
\end{Lemm}

\begin{proof}
To prove the first and second inequalities, we just need to recall the definition of $E_{X}^{(Y)}$ as an expansional given in \eqref{equa:defiEXY}, and notice that $W^{(Y)}$ is supported in $\alpha_{0}\gamma_{0}$, which is the union  of four intervals of length $\leq 3r$. Thus, applying Theorem~\ref{Theo:localityProp} we conclude the result.

Let us next prove the second part. Assume that $\ell \geq r$. Then, since $|B_{2}|, |D_{2}| \geq \ell \geq r$, we have $\widetilde{H}_{D_{1}AB_{1}B_{3}CD_{3}} = \widetilde{H}_{D_{1}AB_{1}}  + \widetilde{H}_{B_{3}CD_{3}} $, and therefore
\[ E^{(A)}_{D_{1}AB_{1}} E^{(C)}_{B_{3}CD_{3}} = E^{(AC)}_{D_{1}AB_{1}B_{3}CD_{3}} \, , \,  
E^{(A)}_{D_{1}AB_{1}}  = E^{(A)}_{D_{1}AB_{1}B_{3}CD_{3}}
\, , \, E^{(C)}_{B_{3}CD_{3}} = E^{(C)}_{D_{1}AB_{1}B_{3}CD_{3}}
\,. \]
Again by Theorem~\ref{Theo:localityProp}, since $W^{(Y)}$ is supported in $\alpha_{0}\gamma_{0}$ for every $Y \in \{ A,C,AC\}$, we can establish 
\[ \| E^{(Y)}_{ABCD} - E^{(Y)}_{\alpha_{n}\gamma_{n}}\| \,\, , \,\,  \| E^{(Y)}_{D_{1}AB_{1}B_{3}CD_{3}} - E^{(Y)}_{(D_{1}AB_{1}B_{3}CD_{3}) \cap \alpha_{n}\gamma_{n}}\| \leq e^{12rJ \mathcal{G}^{12r}} \frac{\mathcal{G}^{n}}{(\lfloor n/r \rfloor+1)!}\,.\]
Taking $n= \ell-r$ we have that $(D_{1}AB_{1}B_{3}CD_{3}) \cap \alpha_{n}\gamma_{n} = \alpha_{n}\gamma_{n}$ since $|B_{i}|, |D_{i}| \geq \ell$.  Therefore
\[ \| E_{ABCD}^{(Y)} - E_{D_{1}AB_{1}}^{(Y)} E_{B_{3}CD_{3}}\| \leq \| E_{ABCD}^{(Y)} - E_{\alpha_{n} \gamma_{n}}^{(Y)} \| + \| E_{D_{1}AB_{1}B_{3}CD_{3}}^{(Y)} - E^{(Y)}_{\alpha_{n} \gamma_{n}} \|\,, \]
and applying the above estimates we get the desired conclusion.
\end{proof}

We will need the following easy observation.
\begin{Lemm}\label{Lemm:productDecomp}
Let $Q$ be an operator on $\cH_{A_{1}} \otimes \ldots \otimes \cH_{A_{n}}$. Then, we can express $Q$ as
\[ Q = \sum_{j=1}^{D^2} Q_{A_{1}}^{(j)} \otimes \ldots \otimes Q_{A_{n}}^{(j)} \quad \text{such that} \quad\sum_{j=1}^{D^2}\|Q_{A_{1}}^{(j)}\|_{\infty} \cdot \ldots \cdot \| Q_{A_{n}}^{(j)}\|_{\infty} \leq D^2 \| Q\|_{\infty} \,, \]
where $D = \dim(\cH_{A_{1}} \otimes \ldots \otimes \cH_{A_{n}} ) =\dim(\cH_{A_{1}}) \cdot \ldots \cdot \dim(\cH_{A_{n}})$.
\end{Lemm}
\begin{proof}
For each $i=1, \dots, n$, consider an orthonormal basis  of $\cB(\cH_{A_i})$. Then, the set of al possible tensor products of these basis elements forms an orthonormal basis of $\cB(\cH_{A_1} \otimes \ldots \otimes \cH_{A_n})$. This allows us to decompose $Q$ as a sum of orthogonal elements 
\[ Q  = \sum_{j=1}^{D^2}  Q_{A_{1}}^{(j)} \otimes \ldots \otimes Q_{A_{n}}^{(j)} \quad \text{such that} \quad \| Q\|_{2}^{2} = \sum_{j=1}^{D^2} \|Q_{A_{1}}^{(j)}\|_{2}^{2} \cdot \ldots \cdot \|Q_{A_{n}}^{(j)}\|_{2}^{2}\,.\]
Applying the Cauchy-Schwartz inequality and the trivial bound $\| Q\|_{2} \leq D \| Q\|_{\infty}$, we conclude that
\[ \sum_{j=1}^{D^2}\|Q_{A_{1}}^{(j)}\|_{\infty} \cdot \ldots \cdot \| Q_{A_{n}}^{(j)}\|_{\infty} \leq D^2 \| Q\|_{\infty} \leq \sqrt{D^{2}} \| Q\|_{2}^{2} \leq D^{2} \| Q\|_{\infty}\,, \]
completing the proof.
\end{proof}

\begin{Lemm}\label{Lemm:TensorProductFact}
There is a positive constant $\mathcal{G}''=\mathcal{G}''(r,J,d)$, where recall that $d$ is the local dimension of the system, such that we can decompose as finite sums
\[ E_{D_{1}AB_{1}}^{(A)}= \sum_{j}  Q_{D_{1}}^{(j)} Q_{A}^{(j)} Q_{B_{1}}^{(j)} \quad \text{and} \quad E_{B_{3}CD_{3}}^{(C)}= \sum_{k}  Q_{B_{3}}^{(k)} Q_{C}^{(k)} Q_{D_{3}}^{(k)}\,, \]
where each $Q_{X}^{(\cdot)}$ is supported in $X$, and satisfying
\[  \sum_{j} \|Q_{D_{1}}^{(j)}\| \, \| Q_{A}^{(j)} \| \, \| Q_{B_{1}}^{(j)} \|  \,\, , \,\,  \sum_{k} \|Q_{B_{3}}^{(k)}\| \, \| Q_{C}^{(k)} \| \, \| Q_{D_{3}}^{(k)} \| \,\, \leq \,\, \mathcal{G}''\,. \]
\end{Lemm}

\begin{proof}
The argument is identical, so we argue with the first case $D_{1}AB_{1}$. Recall the expansional formula of $E_{D_{1}AB_{1}}^{(A)}$ given in terms of $W^{(A)}$, which is supported in $\alpha_{0}$, and $\widetilde{H}_{D_{1}AB_{1}}$. Let us define for each $n \geq 0$
\[ E_{n}^{(A)} := E_{(D_{1}AB_{1}) \cap \alpha_{n}}^{(A)} = E(-W^{(A)}; \widetilde{H}_{\alpha_{m} \cap D_{1}AB_{1}})\,. \]
Observe that this sequence is eventually equal to $E_{D_{1}AB_{1}}$. Thus, we can decompose as a series (which is actually a finite sum)
\[ E^{(A)}_{D_{1}AB_{1}} = E_{0}^{(A)} + \sum_{n=1}^{\infty} E^{(A)}_{n} - E^{(A)}_{n-1}\,. \]
Moreover, by Theorem~\ref{Theo:localityProp}, we have for every $0 \leq n \leq m$ the following estimates 
\[ \| E^{(A)}_{n} \| \leq \mathcal{G}' \quad , \quad \|E_{n}^{(A)} - E_{m}^{(A)} \| \leq \mathcal{G}' \, \frac{\mathcal{G}^{n}}{(\lfloor n/r \rfloor) !}\,. \]
Using Lemma~\ref{Lemm:productDecomp}, we can decompose
\[ E_{0}^{(A)} =  \sum_{j=1}^{d^{2|\alpha_{0}|}} Q_{D_{1} \cap \alpha_0}^{(j)} Q_{A\cap \alpha_0}^{(j)} Q_{B_{1}\cap \alpha_0}^{(j)} \quad , \quad E_{n}^{(A)} - E_{n-1}^{(A)} = \sum_{j=1}^{d^{2|\alpha_{n}|}} Q_{D_{1}\cap \alpha_{n}}^{(j)} Q_{A \cap \alpha_{n}}^{(j)} Q_{B_{1} \cap \alpha_{n}}^{(j)} \,, \]
so that
\[ E^{(A)}_{D_{1}AB_{1}} = \sum_{n=0}^{\infty} \sum_{j=1}^{d^{2|\alpha_{n}|}} Q_{D_{1} \cap \alpha_{n}}^{(j)} Q_{A \cap \alpha_{n}}^{(j)} Q_{B_{1} \cap \alpha_{n}}^{(j)} \]
and
\begin{align*} 
\sum_{n=0}^{\infty} \sum_{j=1}^{d^{|\alpha_{n}|}} \| Q_{D_{1} \cap \alpha_{n}}^{(j)} Q_{A \cap \alpha_{n}}^{(j)} Q_{B_{1} \cap \alpha_{n}}^{(j)} \| & \leq d^{2|\alpha_{0}|} \| E_{0}^{(A)}\| + \sum_{n=1}^{\infty} d^{2|\alpha_{n}|} \|E_{n}^{(A)} - E_{n-1}^{(A)} \| \\
& \leq d^{12r} \mathcal{G} + \sum_{n=1}^{\infty} d^{12(r+n)} \frac{\mathcal{G}^{n}}{(\lfloor n/r \rfloor )!}\,.
\end{align*}
Note that the last series is absolutely convergente and converges to a constant that only depends on $\mathcal{G},d,r$, which proves the result.
\end{proof}

Now we can prove \textbf{Step 2}.

\begin{proof}[Proof of Proposition~\ref{Prop:CorrelationDecayStep2} (Step 2)]
Let us check it in the first case \eqref{equa:approxStep2_1_1}, the idea for the rest is identical. First, we apply Lemma~\ref{Lemm:splitTwo} to approximate
\begin{equation}\label{equa:Step3_Aux1}
\Tr_{BC}\left( \widetilde{\sigma}^{BC}  E_{ABCD}^{(AC)}\right) \approx \Tr_{BC}\left( \widetilde{\sigma}^{BC}  E_{D_{1}AB_{1}}^{(A)}  E_{B_{3}CD_{3}}^{(C)}\right)
\end{equation}
with an error
\begin{equation*}
\left\| \Tr_{BC}\left( \widetilde{\sigma}^{BC}  E_{ABCD}^{(AC)}\right) - \Tr_{BC}\left( \widetilde{\sigma}^{BC}  E_{D_{1}AB_{1}}^{(A)}  E_{B_{3}CD_{3}}^{(C)}\right) \right\| \leq \mathcal{G}' \frac{\mathcal{G}^{\ell}}{(\lfloor \ell/r\rfloor)!}\,.
\end{equation*}
Next, we want to approximate 
\begin{equation}\label{equa:Step3_Aux2}
 \Tr_{BC}\left( \widetilde{\sigma}^{BC} E_{D_{1},A,B_{1}}^{(A)} \, E_{B_{3},C,D_{3}}^{(C)}\right)  \approx \Tr_{B}\left( \widetilde{\sigma}^{B} E_{B,A,D_{1}}^{(A)}\right) \, \Tr_{BC}\left( \widetilde{\sigma}^{BC}  E_{BC D_{3}}^{(C)}\right)\,.
\end{equation}
For that, we use Lemma~\ref{Lemm:TensorProductFact} and the fact that $\widetilde{\sigma}^{BC} = \sigma^{B} \otimes \sigma^{C}$. Inserting these expressions in the left hand-side of  \eqref{equa:Step3_Aux2},  we get
\[
\begin{split}
& \Tr_{BC}(\widetilde{\sigma}^{BC} \, E_{D_{1},A, B_{1}}^{(A)} \, E_{B_{3},C, D_{3}}^{(C)})  = \\[2mm]
& \hspace{1cm} = \sum_{k, j} Q_{A}^{(j)}  \Tr_{C}\left(\sigma^{C} Q_{C}^{(k)}\right) \Tr_{B}\left(\sigma^{B} Q^{(j)}_{B_{1}} Q^{(k)}_{B_{3}} \right) \, Q^{(j)}_{D_{1}} Q^{(k)}_{D_{3}} \,,
\end{split}
\]
and inserting them also on right hand-side of \eqref{equa:Step3_Aux2}, we get
\[
\begin{split}
& \Tr_{B}\left( \widetilde{\sigma}^{B} E_{B_{1}AD_{1}}^{(A)} \right) \, \Tr_{BC}\left( \widetilde{\sigma}^{BC}  E_{B_{3}, C, D_{3}}^{(C)}\right) = \\[2mm]
& \hspace{1cm}  = \sum_{k, j} Q_{A}^{(j,m)}  \Tr_{C}\left(\sigma^{C} Q_{C}^{(k)}\right) \Tr_{B}\left(\sigma^{B} Q^{(j)}_{B_{1}}\right) \Tr_{B}\left( \sigma^{B} Q^{(k)}_{B_{3}} \right) \, Q^{(j,m)}_{D_{1}} Q^{(k)}_{D_{3}} \,.
\end{split}
\]
Comparing both expressions, it is clear that we can estimate the additive error in \eqref{equa:Step3_Aux1} using the correlation decay property of $\sigma_{B}$ on the 1D chain $B$:
\begin{align*}
& \left\| \Tr_{BC}(\widetilde{\sigma}_{BC} \, E_{D_{1}A B_{1}}^{(A)} \, E_{B_{3}C D_{3}}^{(C)}) - \Tr_{B}\left( \widetilde{\sigma}_{B} E_{B_{1}AD_{1}}^{(A)} \right) \, \Tr_{BC}\left( \widetilde{\sigma}_{BC} E_{B_{3} C D_{3}}^{(C)}\right) \right\| \leq \\[2mm]
& \hspace{1cm} \leq  \sum_{n,m,k,j} \| Q_{A}^{(j)}\| \, \| Q_{C}^{(k)}\| \,  \| \Tr_{B}\left(\sigma^{B} Q^{(j)}_{B_{1}} Q^{(k)}_{B_{3}} \right) - \Tr_{B}\left(\sigma^{B} Q^{(j)}_{B_{1}}\right) \Tr_{B}\left(\sigma^B Q^{(k)}_{B_{3}} \right)  \|  \, \| Q_{D_{1}}^{(j)}\| \, \| Q_{D_{3}}^{(k)}\|\\[2mm]
&  \hspace{1cm} \leq \xi(\ell) \, \sum_{k,j} \sum_{n,m,k,j} \| Q_{A}^{(j)}\| \, \| Q_{C}^{(k)}\|  \,  \| Q^{(j)}_{B_{1}}\| \, \| Q^{(j)}_{B_{1}}\| \, \| Q^{(k)}_{B_{3}}  \| \, \| Q_{D_{1}}^{(j)}\| \, \| Q_{D_{3}}^{(k)}\|\\[2mm]
& \hspace{1cm} \leq \mathcal{G}'' \xi( \ell )\,.
\end{align*}
Finally, we can approximate
\begin{equation}\label{equa:Step3_Aux3}
 \Tr_{B}\left( \widetilde{\sigma}^{B} E_{B_{1}AD_{1}}^{(A)} \right) \, \Tr_{BC}\left( \widetilde{\sigma}^{BC} E_{B_{3}, C, D_{3}}^{(C)}\right) \approx \Tr_{B}\left( \widetilde{\sigma}^{B} E_{ABCD}^{(A)} \right) \, \Tr_{BC}\left( \widetilde{\sigma}^{BC} E_{ABCD}^{(C)}\right)\,.
\end{equation}
The additive error in this case, can be estimated using Lemma~\ref{Lemm:splitTwo} 
\begin{align*}
& \left\|\Tr_{B}\left( \widetilde{\sigma}^{B} E_{B_{1}AD_{1}}^{(A)} \right) \, \Tr_{BC}\left( \widetilde{\sigma}^{BC} E_{B_{3}, C, D_{3}}^{(C)}\right) - \Tr_{B}\left( \widetilde{\sigma}^{B} E_{ABCD}^{(A)} \right) \, \Tr_{BC}\left( \widetilde{\sigma}^{BC} E_{ABCD}^{(C)}\right)\right\| \\[2mm]
& \hspace{1cm} \leq \| E_{B_{1}AD_{1}}^{(A)} - E_{ABCD}^{(A)}\| \cdot \| E_{B_{3} CD_{3}}^{(C)}\| + \|E_{ABCD}^{(A)} \| \cdot \|E_{B_{3}CD_{3}}^{(C)} - E_{ABCD}^{(C)} \| \\[2mm]
& \hspace{1cm}  \leq 4 \mathcal{G}' \exp{\mathcal{G}^{12r}12rJ} \frac{\mathcal{G}^{\ell}}{(\lfloor \ell/r \rfloor)!}\,.
\end{align*}
Finally, combining \eqref{equa:Step3_Aux1}, \eqref{equa:Step3_Aux2} and \eqref{equa:Step3_Aux3} with their respective additive error estimations, we conclude the result.
\end{proof}

%\begin{tcolorbox}
%\textbf{Step 3:}
\begin{Prop}[\bfseries Step 3] \label{Prop:CorrelationDecayStep3} Assume that $\ell \geq r$. The approximations in equations \eqref{equa:approxStep2_1_2} and \eqref{equa:approxStep2_1_4} can be rewritten in terms of their inverses 
\begin{align} 
\label{equa:approxStep3_1_2} \Tr_{ABC}\left( \widetilde{\sigma}_{ABC} E_{ABCD}^{(AC)}\right)^{-1} & \approx \Tr_{AB}\left( \widetilde{\sigma}_{B} E_{ABCD}^{(C)}\right)^{-1}  \, \Tr_{BC}\left( \widetilde{\sigma}_{BC} E_{ABCD}^{(A)}\right)^{-1}\\[2mm] 
\label{equa:approxStep3_1_4} \Tr_{B}\left( \widetilde{\sigma}^{B} E_{ABCD}^{(AC)}\right)^{-1} & \approx \Tr_{B}\left( \widetilde{\sigma}^{B} E_{ABCD}^{(C)}\right)^{-1} \, \Tr_{B}\left( \widetilde{\sigma}^{B} E_{ABCD}^{(A)}\right)^{-1}\,.
\end{align}
Similarly, we have
\begin{align} 
\label{equa:approxStep3_2_2}  
\Tr_{ABCD}\left(\widetilde{\sigma}^{ABCD} E^{(AC)}_{ABCD}\right)^{-1}  
& \approx \Tr_{BCD}\left(\widetilde{\sigma}^{BCD} \, E_{ABCD}^{(C)}\right)^{-1} \, \Tr_{ABD}\left(\widetilde{\sigma}^{ABD} \, E_{ABCD}^{(A)}\right)^{-1} \\[2mm]
\label{equa:approxStep3_2_4} 
\Tr_{BD}\left(\widetilde{\sigma}^{BD} E^{(AC)}_{ABCD}\right)^{-1}  
& \approx \Tr_{BD}\left(\widetilde{\sigma}^{BD}  E_{ABCD}^{(C)}\right)^{-1} \, \Tr_{BD}\left(\widetilde{\sigma}^{BD}  E_{ABCD}^{(A)}\right)^{-1}  
\end{align}
Indeed, the additive error in all these approximations $x \approx y$ can be estimated by 
\[ \| x-y\| \leq \mathcal{G}^{(iv)} \mathcal{G}''' \,\left(\xi(\ell) + \frac{\mathcal{G}^{\ell}}{(\lfloor \ell/r \rfloor)!}\right) \,,\]
for some positive constant $\mathcal{G}^{(iv)} = \mathcal{G}^{(iv)}(r,J)$.
\end{Prop}
%\end{tcolorbox}

For the proof we need a uniform bound.

\begin{Lemm}\label{Lemm:boundsPartialtrace}
With the previous notation, if $X \subset \Lambda$ is a union of intervals of the set $\{ A,B,C,D\}$ and $Y \in \{ A,C,AC\}$, then
\[ Q_{X,Y}:=\Tr_{X}(\widetilde{\sigma}_{X} E_{ABCD}^{(Y)})  \,\,  \mbox{ is invertible with } \,\, \| Q_{X,Y}\|_{\infty}, \| (Q_{X,Y})^{-1}\|_{\infty} \leq \mathcal{G}^{(iv)}\,.\]
for some positive constant $\mathcal{G}^{(iv)} = \mathcal{G}^{(iv)}(r,J)$.
\end{Lemm}

\begin{proof}
Let us define for every $n \geq 0$ the expansional
\[ E_{n} := E(-\tfrac{1}{2}W^{(Y)}, \tfrac{1}{2} \widetilde{H}_{\alpha_{n}\gamma_{n}}) = e^{-(\frac{1}{2} W^{(Y)} + \frac{1}{2} \widetilde{H}_{\alpha_{n} \gamma_{n}})} e^{\frac{1}{2}\widetilde{H}_{\alpha_{n} \gamma_{n}}}\,. \]
Observe that for large $m$, we have $\alpha_{m}\gamma_{m} = ABCD$. Moreover, since $X$ is a union of intervals of $\{ A,B,C,D\}$, then $\widetilde{H}_{XX^c} = \widetilde{H}_{X} + \widetilde{H}_{X^{c}}$. Thus, we can rewrite
\begin{align*}  
Q_{X,Y} = \frac{1}{\widetilde{Z}_{X}} \Tr_{X}\left(e^{-(W^{(Y)}+\widetilde{H}_{XX^c})} e^{\widetilde{H}_{X^{c}}}\right)
& = 
\frac{1}{\widetilde{Z}_{X}}e^{-\frac{1}{2}\widetilde{H}_{X^{c}}} \Tr_{X}\left(e^{\frac{1}{2}\widetilde{H}_{X^{c}}}e^{-(W^{(Y)}+\widetilde{H}_{XX^c})} e^{\frac{1}{2}\widetilde{H}_{X^{c}}}\right)e^{\frac{1}{2}\widetilde{H}_{X^{c}}}\\[2mm] 
& = 
\frac{1}{\widetilde{Z}_{X}} e^{-\frac{1}{2} \widetilde{H}_{X^c}} \, \Tr_{X}\left( e^{-H_{X}} E_{m}^{\dagger} E_{m} \right)  \, e^{\frac{1}{2} \widetilde{H}_{X^c}}\\[2mm]
& = e^{-\frac{1}{2} \widetilde{H}_{X^c}} \, \Tr_{X}\left( (\widetilde{\sigma}^{X})^{1/2} E_{m}^{\dagger} E_{m} (\widetilde{\sigma}^{X})^{1/2}\right) \, e^{\frac{1}{2} \widetilde{H}_{X^c}}\,.
\end{align*}
If we define for each $n \geq 0$
\[ O_{n}:= \Tr_{X}\left( (\widetilde{\sigma}^{X})^{1/2} E_{n}^{\dagger} E_{n} (\widetilde{\sigma}^{X})^{1/2}\right)\,,  \]
then this is a positive operator satisfying the following bounds for every $0\leq n \leq m$
\[\mathcal{G}'^{-1} \mathbbm{1} \leq O_{n} \leq \mathcal{G}' \mathbbm{1} \quad , \quad \| O_{n} - O_{m}\| \leq \mathcal{G}' \frac{\mathcal{G}^{n}}{(\lfloor n/r \rfloor + 1)!}\,.  \]
As a consequence, we also have
\[ \| O_{n}^{-1} - O_{m}^{-1}\| = \| O_{m}^{-1}(O_{n} - O_{m}) O_{n}^{-1}\| \leq \mathcal{G}'^{3} \frac{\mathcal{G}^n}{(\lfloor n/r \rfloor + 1)!}\,. \]
Applying now Theorem~\ref{Theo:localityProp}, we can estimate
\begin{align*}
\| Q_{X,Y}^{-1} \| 
& = \| e^{\frac{1}{2}H_{X^c}} (O_{0}^{-1} + \sum_{n=1}^{m} (O_{n}^{-1} - O_{n-1}^{-1})) \, e^{\frac{1}{2}H_{X^c}} \|\\[2mm]
& \leq \| e^{\frac{1}{2}H_{X^c}} O_{0}^{-1} e^{\frac{1}{2}H_{X^c}}\| + \sum_{n=1}^{m} \| e^{\frac{1}{2}H_{X^c}} (O_{n}^{-1} - O_{n-1}^{-1}) e^{\frac{1}{2}H_{X^c}}  \|\\
& \leq \| O_{0}^{-1}\| \mathcal{G}^{6r} + \sum_{n=1}^{m} \| O_{n}^{-1} - O_{n-1}^{-1}\| \, \mathcal{G}^{6(r+n)}\\
& \leq \mathcal{G}'^2 \mathcal{G}^{6r} + \sum_{n=1}^{m} \mathcal{G}'^{3} \mathcal{G}^{6(r+n)} \frac{\mathcal{G}^{n-1}}{(\lfloor n/r \rfloor)!}\,.
\end{align*}
This last series is absolutely convergent, which establishes the desired conclusion.
\end{proof}

Now we can prove \textbf{Step 3}.

\begin{proof}[Proof of Proposition~\ref{Prop:CorrelationDecayStep3}(Step 3)]
We only need to use that 
\[\| O_{1}^{-1} - O_{2}^{-1}\| \leq \| O_{1}^{-1}\| \, \| O_{2}^{-1}\| \, \| O_{1} - O_{2}\|\]
and apply the estimates from Proposition~\ref{Prop:CorrelationDecayStep2}(Step 2) on $\| O_{1} - O_{2}\|$, together with the bound from Lemma~\ref{Lemm:boundsPartialtrace} on $\| O_{1}^{-1}\|$ and $\| O_{2}^{-1}\|$.
\end{proof}

%\begin{tcolorbox}
%\textbf{Step 4:} 
\begin{Prop}[\bfseries Step 4]\label{Prop:CorrelationDecayStep4}
Assume that $\ell \geq r$. Then, we can estimate
\begin{equation*} 
\|\sigma_{AD} \sigma_{D}^{-1} \sigma_{DC} \sigma_{ADC}^{-1}  
  - \mathbbm{1} \| \,\,,\,\,
   \| \sigma_{A}  \sigma_{C} \sigma_{AC}^{-1}  
   - \mathbbm{1} \|
   \,\, \leq \,\,
   4 \left(\mathcal{G}^{(iv)}\right)^{7} \mathcal{G'''} \left(\xi(\ell) + \frac{\mathcal{G}^{\ell}}{(\lfloor \ell/r \rfloor)!}\right)\,.
\end{equation*}
\end{Prop}
%\end{tcolorbox}

\begin{proof}
Let us focus on the first inequality. Applying to equation \eqref{equa:Step1_1} from Proposition~\ref{Prop:CorrelationDecayStep1} (Step 1)   approximations \eqref{equa:approxStep2_1_1}, \eqref{equa:approxStep2_1_3}, \eqref{equa:approxStep3_1_2} and \eqref{equa:approxStep3_1_4} from Propositions~\ref{Prop:CorrelationDecayStep2} (Step 2) and~\ref{Prop:CorrelationDecayStep3} (Step 3), we get
\begin{align*} 
 \sigma_{AD} \sigma_{D}^{-1} \sigma_{DC} \sigma_{ADC}^{-1} & \approx \Tr_{B}\left( \widetilde{\sigma}^{B} E_{ABCD}^{(A)}\right) \, \Tr_{BC}\left( \widetilde{\sigma}^{BC}  E_{ABCD}^{(C)}\right)\\[2mm]
& \hspace{1cm} \cdot \Tr_{AB}\left( \widetilde{\sigma}_{B} E_{ABCD}^{(C)}\right)^{-1}  \, \Tr_{BC}\left( \widetilde{\sigma}_{BC} E_{ABCD}^{(A)}\right)^{-1}\\[2mm] 
& \hspace{1cm} \cdot \Tr_{AB}\left( \widetilde{\sigma}^{AB}  E_{ABCD}^{(A)}\right) \, \Tr_{B}\left( \widetilde{\sigma}^{B} E_{ABCD}^{(C)}\right)\\[2mm]
& \hspace{1cm} \cdot \Tr_{B}\left( \widetilde{\sigma}^{B} E_{ABCD}^{(C)}\right)^{-1} \, \Tr_{B}\left( \widetilde{\sigma}^{B} E_{ABCD}^{(A)}\right)^{-1}\\[2mm]
 & = \mathbbm{1}\,.
\end{align*} 
To bound the additive error, we use that if $\| a_{i} - b_{i}\| \leq \varepsilon$ and $\| a_{i}\|, \|b_{i} \| \leq \kappa$ for $i=1,2,3,4$, then $\| a_{1}a_{2}a_{3}a_{4} - b_{1} b_{2} b_{3} b_{4}\| \leq 4 \kappa^{3} \varepsilon$. In our case, we can take $\kappa = \left(\mathcal{G}^{(iv)}\right)^2$ by Lemma~\ref{Lemm:boundsPartialtrace}, and $\varepsilon = \mathcal{G}^{(iv)} \mathcal{G}''' \,\left(\xi(\ell) + \frac{\mathcal{G}^{\ell}}{(\lfloor \ell/r \rfloor)!}\right)$ by Step 2 and 3. This yields the first inequality. The argument for the second inequality is analogous.
\end{proof}

We can now finish this section with the proof of the main result.

\begin{proof}[Proof of Theorem~\ref{Theo:1DCorrelationDecayAux}]
We will focus on \eqref{equa:1DCorrelationDecayAux1} and \eqref{equa:1DCorrelationDecayAux2}, since \eqref{equa:1DCorrelationDecayAux3} is completely analogous to the latter. If $\ell \geq r$,  the desired bounds follow directly from Proposition~\ref{Prop:CorrelationDecayStep4}. On the other hand, if $\ell < r$ we can simply estimate
\begin{align*} 
\|\sigma_{AD} \sigma_{D}^{-1} \sigma_{DC} \sigma_{ADC}^{-1}  
  - \mathbbm{1} \| & \leq 1+ \| \sigma_{AD} \sigma_{D}^{-1} \sigma_{DC} \sigma_{ADC}^{-1}\|\,,\\[2mm]
   \| \sigma_{A}  \sigma_{C} \sigma_{AC}^{-1}  
   - \mathbbm{1} \|
   & \leq 1+ \| \sigma_{A}  \sigma_{C} \sigma_{AC}^{-1}\|\,.
\end{align*}
We then apply the expressions given in Proposition~\ref{Prop:CorrelationDecayStep1} (Step 1) to rewrite each expression $\sigma_{AD} \sigma_{D}^{-1} \sigma_{DC} \sigma_{ADC}^{-1}$ and $\sigma_{A}  \sigma_{C} \sigma_{AC}^{-1}$, as a product of four factors. Each of these terms is  bounded by a constant depending only on $r$ and $J$ by Lemma~\ref{Lemm:boundsPartialtrace}.
\end{proof}

\section{Proof of Theorem~\ref{Theo:abelianModelsMarginals}} \label{Appendix:Marginals_AbelianQDM}

Recall that $A_{s}$ and $B_{p}$ are orthogonal projections, and that for every projection $P$, using that $P^{n} = P$ for each $n \geq 1$, 
\begin{equation}\label{equa:exponentialProjection} 
e^{\beta P} = \sum_{n=0}^{\infty} \frac{\beta^{n}}{n!} P^{n} = \mathbbm{1} + \sum_{n=1}^{\infty} \frac{\beta^{n}}{n!} P = \mathbbm{1} + (e^{\beta} - 1)P = (1+\gamma_{\beta}) \, \mathbbm{1} \,\, + \,\, \gamma_{\beta} \,\widetilde{P}  \,.
\end{equation}
where we are denoting $\widetilde{P}:= |G| \, P - \mathbbm{1}$.
\noindent Hence, we can rewrite
\begin{equation*}\label{equa:GibbsStarOp}
    e^{\beta A_{s}}   \,\, = \,\,  (1 + \gamma_{\beta}) \mathbbm{1}  + \gamma_{\beta} \sum_{g \neq 1}  A_{s}(g) \,\, =  \,\, \sum_{g \in G} (\delta_{g,1} + \gamma_{\beta}) A_{s}(g) \,,
\end{equation*} 
where $\delta_{1,g}$ stands for the Kronecker delta. Inserting this expression in the product of the exponentials of star operators, and expanding the product, we get
\begin{equation}\label{equa:ProductGibbsStarOp}
\prod_{s \in \cS_{\mathcal{R}}} e^{\beta_{1}A_{s}} 
=  \prod_{s \in \cS_{\mathcal{R}}} \, \sum_{g \in G} (\delta_{g,1} + \gamma_{\beta}) A_{s}(g) 
= \sum_{(g_{s}) \in G^{\cS_{\mathcal{R}}}} \, \prod_{s \in \cS_{\mathcal{R}}} (\delta_{g_s,1} + \gamma_{\beta}) \prod_{s \in \cS_{\mathcal{R}}} A_{s}(g_s)\,. 
\end{equation}
Analogously, we can derive the formulas
\begin{equation*}\label{equa:GibbsPlaqOp}
    e^{\beta B_{p}}   \,\, = \,\,  (1 + \gamma_{\beta}) \mathbbm{1}  + \gamma_{\beta} \sum_{\chi \neq \mathbf{1}}  B_{p}(\chi) \,\, = \,\,  \sum_{\chi \in \widehat{G}} (\delta_{\chi,\mathbf{1}} + \gamma_{\beta})B_{p}(\chi)  \,,
\end{equation*} 
and
\begin{equation}\label{equa:ProductGibbsPlaqOp}
\prod_{p \in \mathcal{P}_{\cR}} e^{\beta B_{p}} = \prod_{s \in \mathcal{P}_{\cR}} \, \sum_{\chi \in \widehat{G}}\left( \delta_{\chi,\mathbf{1}}+\gamma_{\beta}\right) B_{p}(\chi) = \sum_{(\chi_{p}) \in \widehat{G}^{\mathcal{P}_\cR}} \, \prod_{p \in \mathcal{P}_\cR} (\delta_{\chi_{p}, \mathbf{1}} + \gamma_{\beta}) \,  \prod_{p \in \mathcal{P}_\cR} B_{p}(\chi_{p})
\end{equation}
Combining \eqref{equa:ProductGibbsStarOp} and \eqref{equa:ProductGibbsPlaqOp}
\[ \prod_{s \in \cS_{\mathcal{R}}} e^{\beta_{1}A_{s}} \prod_{p \in \mathcal{P}_{\cR}} e^{\beta B_{p}} = \sum_{(g_{s}) \in G^{\cS_{\mathcal{R}}}} \sum_{(\chi_{p}) \in \widehat{G}^{\mathcal{P}_\cR}} \, \prod_{s \in \cS_{\mathcal{R}}} (\delta_{g_s,1} + \gamma_{\beta}) \, \prod_{p \in \mathcal{P}_\cR} (\delta_{\chi_{p}, \mathbf{1}} + \gamma_{\beta})   \prod_{s \in \cS_{\mathcal{R}}} A_{s}(g_s) \,  \prod_{p \in \mathcal{P}_\cR} B_{p}(\chi_{p})\,.\]
Next, we have to take partial trace over $\cR$ on the previous expression. By linearity, we will obtain a linear combination of summands of the form
\begin{equation}\label{equa:partialTraceSummand}
\Tr_{\mathcal{R}}\left( \prod_{s \in \cS_{\mathcal{R}}} A_{s}(g_s) \,  \prod_{p \in \mathcal{P}_\cR} B_{p}(\chi_{p}) \right)
\end{equation}
We claim now that this last term will be nonzero only if both families $(g_{s})_{s}$ and $(\chi_{p})_{p}$ are constant. To check this claim, let us start by observing what happens if we trace only one edge $e_{0} \in \mathcal{R}$. This edge belongs to exactly two stars, say $s_{0}$ and $s_{0}'$ and two plaquettes, say $p_{0}$ and $p'_{0}$, as in the following picture:
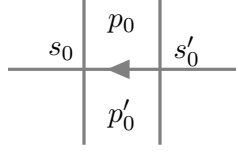
\[
\begin{tikzpicture}

%Grid
\draw[very thick, gray]  (0,0) -- (3,0) ;
\draw[very thick, gray]  (1,-1) -- (1,1) ;
\draw[very thick, gray]  (2,-1) -- (2,1) ;
\fill[gray, very thick] (1.6,-0.13) -- (1.6,0.13) -- (1.3,0);

%Labels
\draw (1.5,0.6) node {$p_{0}$}; 
\draw (1.5,-0.6) node {$p_{0}'$}; 
\draw (0.7,0.24) node {$s_{0}$}; 
\draw (2.35,0.24) node {$s_{0}'$}; 

\end{tikzpicture} 
\]
Thus, we can take out from the partial trace the rest of factors that act as the identity on $e$:
\begin{equation}\label{equa:partialTraceSummand2}
\begin{split}
\Tr_{e_{0}}\left( \prod_{s \in \cS_{\cR}} A_{s}(g_s) \,  \prod_{p \in \mathcal{P}_\cR} B_{p}(\chi_{p}) \right) & = \prod_{s \in \cS_{\mathcal{R}}} A_{s}(g_s) \,  \prod_{p \in \mathcal{P}_\cR} B_{p}(\chi_{p}) \, \cdot \\[1mm]
& \hspace{1cm} \cdot \Tr_{e_{0}}\left(A_{s_{0}}(g_{s_{0}}) \cdot A_{s_{0}'}(g_{s_{0}'}) \cdot B_{p_{0}}(\chi_{p_{0}}) \cdot  B_{p_{0}'}(\chi_{p_{0}'}) \right)
\end{split}
\end{equation}
Observe that in the product of the four last operators, the resulting operator acting on the edge $e_{0}$ is
\begin{equation}\label{equa:localTraceEdge0}
L^{g_{s_{0}}} L^{g_{s_{0}'}^{-1}} \Pi_{\chi_{p_{0}}} \Pi_{\overline{\chi}_{p_{0}'}} = \sum_{h \in G} \overline{\chi}_{p_{0}}(h) \chi_{p_{0}}(h) \dyad*{g_{s_{0}} g_{s_{0}'}^{-1}h}{h}\,. 
\end{equation}
Thus, tracing out $e_{0}$ we get
\begin{equation}\label{equa:localTraceEdge}
\Tr_{e}(L^{g_{s_{0}}} L^{g_{s_{0}'}^{-1}} \Pi_{\chi_{p_{0}}} \Pi_{\overline{\chi}_{p_{0}'}}) = \sum_{h \in G} \overline{\chi}_{p_{0}}(h) \chi_{p_{0}}(h) \braket{g_{s_{0}} g_{s_{0}'}^{-1}h}{h}\,. 
\end{equation}
For this element to be nonzero, we need $g_{s_{0}} g_{s_{0}'}^{-1} h = h$ for some $h \in G$, which is equivalent to $g_{s_{0}} = g_{s_{0}'}$. In this case, all the brakets in \eqref{equa:localTraceEdge} are equal to one, and so the trace reduces to
\[  \Tr_{e}(L^{g_{s_{0}}} L^{g_{s_{0}'}^{-1}} \Pi_{\chi_{p_{0}}} \Pi_{\overline{\chi}_{p_{0}'}}) = \sum_{h \in G}\overline{\chi}_{p_{0}}(h) \chi_{p_{0}}(h) = \langle \chi_{p_{0}'}, \chi_{p_{0}} \rangle\,. \]
Using the orthogonality relations of the characters in $\widehat{G}$, we get that the product is zero unless $\chi_{p_{0}} = \chi_{p_{0}'}$. In summary, \eqref{equa:localTraceEdge} and thus \eqref{equa:partialTraceSummand2}, are nonzero if and only if $g_{s_{0}} = g_{s_{0}'}$ and $\chi_{p_{0}} = \chi_{p_{0}'}$. Observe that, in this case, the operator in \eqref{equa:localTraceEdge0} becomes the identity, and thus its trace is equal to $|G|$. In summary, we have that
\begin{multline}\label{eq:local-trace}
    \Tr_{e_{0}}\left(A_{s_{0}}(g_{s_{0}}) \cdot A_{s_{0}'}(g_{s_{0}'}) \cdot B_{p_{0}}(\chi_{p_{0}}) \cdot  B_{p_{0}'}(\chi_{p_{0}'}) \right) \\= \delta_{g_{s_0}, g_{s_0'}} \delta_{\chi_{p_0}, \chi_{p_0'}}
    \abs{G} A_{s_{0}}(g_{s_{0}}) \cdot A_{s_{0}'}(g_{s_{0}}) \cdot B_{p_{0}}(\chi_{p_{0}}) \cdot  B_{p_{0}'}(\chi_{p_{0}}).
\end{multline}

We are going to extend now the previous observation to any connected (by stars and operators) region $\mathcal{R}$, proving that \eqref{equa:partialTraceSummand} is nonzero only if $(g_{s})$ and $(\chi_{p})$ are constant. Indeed, given that $\cR$ is connected by stars, for any two arbitrary stars $s$ and $s'$ in $\cS_{\cR}$, there exists a sequence of stars $s=s_{1}, \ldots, s_{n}=s'$ such that $\emptyset \neq s_{j} \cap s_{j+1} \subset \cR$ for each $j$.  By expressing $\Tr_{\cR} = \Tr_{\cR \setminus \{ e_{j}\}} \Tr_{e_{j}}$, we see that \eqref{equa:partialTraceSummand} is nonzero only if $g_{s_{j}} = g_{s_{j+1}}$ basing on the previous argument for a single edge. But this yields that $g_{s} = g_{s_{1}} = \ldots g_{s_{n}} = g_{s'}$. Since $s$ and $s'$ were arbitrary, we conclude that $(g_{s})$ has to be constant. Similarly, one can verify that \eqref{equa:partialTraceSummand} is nonzero only if  $(\chi_{p})$ is also constant. 

Using the previous fact, we can erase most of summands and only consider those corresponding to constant families in
\[ 
\Tr_{\cR}(\prod_{s \in \cS_{\mathcal{R}}} e^{\beta_{1}A_{s}} \prod_{p \in \mathcal{P}_{\cR}} e^{\beta B_{p}}) = \sum_{g \in G} \sum_{\chi \in \widehat{G}} \,  (\delta_{g,1} + \gamma_{\beta})^{|\cS_{\mathcal{R}}|} \,  (\delta_{\chi, \mathbf{1}} + \gamma_{\beta})^{|\mathcal{P}_\cR|} \,  \Tr_{\cR}( \prod_{s \in \cS_{\mathcal{R}}} A_{s}(g) \,  \prod_{p \in \mathcal{P}_\cR} B_{p}(\chi))\,.\]
Observe that the weights maximize in the case in which $g=1$ and $\chi = \mathbf{1}$. Thus we can split the previous expression into four summands
\begin{align}
\Tr_{\cR}(\prod_{s \in \cS_{\mathcal{R}}} e^{\beta_{1}A_{s}} \prod_{p \in \mathcal{P}_{\cR}} e^{\beta B_{p}}) 
& \label{equa:firstSummand} = (1+\gamma_\beta)^{|\cS_{\mathcal{R}}|} \,  (1 + \gamma_{\beta})^{|\mathcal{P}_\cR|} \,  \Tr_{\cR}( \prod_{s \in \cS_{\mathcal{R}}} A_{s}(1) \,  \prod_{p \in \mathcal{P}_\cR} B_{p}(\mathbf{1}))\\[2mm]
& \label{equa:secondSummand} \quad  + (1+\gamma_{\beta})^{|\cS_{\mathcal{R}}|} \,   \gamma_{\beta}^{|\mathcal{P}_{\cR}|} \,  \sum_{\substack{\chi \in \widehat{G}\\ \chi \neq \mathbf{1}}}\Tr_{\cR}( \prod_{s \in \cS_{\mathcal{R}}} A_{s}(1) \,  \prod_{p \in \mathcal{P}_\cR} B_{p}(\chi)) \\[2mm]
& \label{equa:thirdSummand} \quad + \gamma_{\beta}^{|\cS_{\mathcal{R}}|} \,   (1+\gamma_{\beta})^{|\mathcal{P}_{\cR}|} \,  \sum_{\substack{g \in G\\ g \neq 1}}\Tr_{\cR}( \prod_{s \in \cS_{\mathcal{R}}} A_{s}(g) \,  \prod_{p \in \mathcal{P}_\cR} B_{p}(\mathbf{1})) \\[2mm] 
& \label{equa:fourthSummand} \quad + \gamma_{\beta}^{|\cS_{\cR}|} \,   \gamma_{\beta}^{|\mathcal{P}_{\cR}|} \,  \sum_{\substack{g \in G\\ g \neq 1}} \sum_{\substack{\chi \in \widehat{G}\\ \chi \neq \mathbf{1}}}\Tr_{\cR}( \prod_{s \in \cS_{\cR}} A_{s}(g) \,  \prod_{p \in \mathcal{P}_\cR} B_{p}(\chi)) \,.
\end{align}
In the first summand \eqref{equa:firstSummand}, observe that $A_{s}(1)$ and $B_{p}(\mathbf{1})$ are equal to the identity. Tracing out $\cR$ results in a factor $|G|^{\cR}$ accompanying the identity on the complemenet of $\cR$. If we extract as a common factor $|G|^{\cR}(1+\gamma_{\beta})^{|\cS_\cR|} (1+\gamma_{\beta})^{|\mathcal{P}_\cR|}$ from the four factors, we then obtain
\begin{equation}\label{equa:AbelianMarginalAlmostIdentity}
\Tr_{\cR}(\prod_{s \in \cS_{\mathcal{R}}} e^{\beta_{1}A_{s}} \prod_{p \in \mathcal{P}_{\cR}} e^{\beta B_{p}}) = (1+\gamma_{\beta})^{|\cS_\cR|} (1+\gamma_{\beta})^{|\mathcal{P}_\cR|} \left( \mathbbm{1} + \cS_{2} + \cS_{3} + \cS_{4}\right) 
\end{equation}
where
\[ \cS_{2} = \left(\frac{\gamma_{\beta}}{1+\gamma_{\beta}}\right)^{|\mathcal{P}_{\cR}|} \sum_{\substack{\chi \in \widehat{G}\\ \chi \neq \mathbf{1}}} \frac{1}{|G|^{\cR}}\Tr_{\cR}( \prod_{s \in \cS_{\mathcal{R}}} A_{s}(1) \,  \prod_{p \in \mathcal{P}_\cR} B_{p}(\chi)) 
= \left(\frac{\gamma_{\beta}}{1+\gamma_{\beta}}\right)^{|\mathcal{P}_{\cR}|} \sum_{\substack{\chi \in \widehat{G}\\ \chi \neq \mathbf{1}}} \prod_{p \in \mathcal{P}_\cR} B_{p}(\chi)
\,, \]
\[ \cS_{3} = \left(\frac{\gamma_{\beta}}{1+\gamma_{\beta}}\right)^{|\cS_{\cR}|} \sum_{\substack{g \in G\\ g \neq 1}} \frac{1}{|G|^{\cR}}\Tr_{\cR}( \prod_{s \in \cS_{\mathcal{R}}} A_{s}(g) \,  \prod_{p \in \mathcal{P}_\cR} B_{p}(\mathbf{1}))
= \left(\frac{\gamma_{\beta}}{1+\gamma_{\beta}}\right)^{|\cS_{\cR}|} \sum_{\substack{g \in G\\ g \neq 1}} \prod_{s \in \cS_{\mathcal{R}}} A_{s}(g)
\,, \]
\begin{multline*}
\cS_{4} = \left(\frac{\gamma_{\beta}}{1+\gamma_{\beta}}\right)^{|\cS_{\cR}|+|\mathcal{P}_{\cR}|}  \sum_{\substack{g \in G\\ g \neq 1}} \sum_{\substack{\chi \in \widehat{G}\\ \chi \neq \mathbf{1}}} \frac{1}{|G|^{\cR}}\Tr_{\cR}( \prod_{s \in \cS_{\cR}} A_{s}(g) \,  \prod_{p \in \mathcal{P}_\cR} B_{p}(\chi))\\
=  \left(\frac{\gamma_{\beta}}{1+\gamma_{\beta}}\right)^{|\cS_{\cR}|+|\mathcal{P}_{\cR}|}  \sum_{\substack{g \in G\\ g \neq 1}} \sum_{\substack{\chi \in \widehat{G}\\ \chi \neq \mathbf{1}}} \prod_{s \in \cS_{\cR}} A_{s}(g) \,  \prod_{p \in \mathcal{P}_\cR} B_{p}(\chi)\,.
\end{multline*} 
Recall that operators $A_{s}(g)$ and $B_{p}(\chi)$ commute with each other for every $g \in G$ and $\chi \in \widehat{G}$. Then, we can use the above expressions for $\cS_1$, $\cS_2$ and $\cS_3$ to rewrite \eqref{equa:AbelianMarginalAlmostIdentity} as
\begin{equation}\label{equa:AbelianMarginalAlmostIdentity2}
 \kappa_{\cR} \left( \mathbbm{1} + \left(\frac{\gamma_{\beta}}{1+\gamma_{\beta}}\right)^{|\cS_{\cR}|}\sum_{\substack{g \in G\\ g \neq 1}} \prod_{s \in \cS_{\mathcal{R}}} A_{s}(g)\right)  
\left( \mathbbm{1} + \left( \frac{\gamma_{\beta}}{1+\gamma_{\beta}}\right)^{|\mathcal{P}_{\cR}|}  \sum_{\substack{\chi \in \widehat{G}\\ \chi \neq \mathbf{1}}} \prod_{p \in \mathcal{P}_\cR} B_{p}(\chi) \right) 
\end{equation}
where $\kappa_{\cR}:= |G|^{|\cR|}(1+\gamma_{\beta})^{|\cS_\cR|} (1+\gamma_{\beta})^{|\mathcal{P}_\cR|}$. Finally, let us define
\[ A_{\cS_\cR} := \frac{1}{|G|} \sum_{g \in G} \prod_{s \in \cS_\cR} A_{s}(g) 
\quad , \quad 
B_{\mathcal{P}_\cR} := \frac{1}{|G|} \sum_{\chi \in \hat{G}} \prod_{p \in \mathcal{P}_\cR} B_{p}(\chi)\,. \]
It is easy to see that these are orthogonal projections, as they are self-adjoint and idempotent. We can finally rewrite \eqref{equa:AbelianMarginalAlmostIdentity} in terms of these projections:
\[\textstyle 
\kappa_{\cR} 
\left( \left(1- \left(\frac{\gamma_{\beta}}{1+\gamma_{\beta}}\right)^{|\cS_\cR|}\right) \mathbbm{1} + |G|\left(\frac{\gamma_{\beta}}{1+\gamma_{\beta}}\right)^{|\cS_\cR|} A_{\cS_\cR} \right)
\left( \left(1-\left(\frac{\gamma_{\beta}}{1+\gamma_{\beta}}\right)^{|\cP_\cR|}\right) \mathbbm{1} + |G|\left(\frac{\gamma_{\beta}}{1+\gamma_{\beta}}\right)^{|\cP_\cR|} B_{\cP_\cR} \right)\,.
\]
This finishes the proof.

%Finally, observe that since $B_{p}(\chi)$ and $A_{s}(g)$ are unitaries, their operator norm is equal to one. This allows us to estimate
%\[ \| \cS_{2}\|_{\infty} \leq |G|\left(\frac{\gamma_{\beta}}{1+\gamma_{\beta}}\right)^{|\mathcal{P}_{\cR}|} \quad , \quad \|\cS_{3}\|_{\infty} \leq |G| \left(\frac{\gamma_{\beta}}{1+\gamma_{\beta}}\right)^{|\cS_{\cR}|} \quad , \quad \| \cS_{4}\|_{\infty} \leq |G|^{2} \left(\frac{\gamma_{\beta}}{1+\gamma_{\beta}}\right)^{|\mathcal{P}_{\cR}| + |\cS_\cR|} \,. \]
%This concludes the result.

\section{Proof of Theorem~\ref{thm:marginals-decayQDM}}\label{Appendix:Marginals_GeneralQDM}

For every $\widehat{g}:\cR \to G$ let us denote $\ket{\widehat{g}} := \otimes_{e \in \cR} \ket{\widehat{g}(e)}$. The set of these vector form an orthonormal basis of $\cH_\cR$, and the corresponding orthogonal projections
\begin{equation}\label{equa:ProjectionGroupElements}   
P_{\widehat{g}} := \mathbbm{1}_{\cR^{c}} \otimes \ketbra{\widehat{g}}{\widehat{g}} \,,
\end{equation}
which notice are supported on region $\cR$, satisfy the relation
\begin{equation}\label{equa:completenessProjections} 
\mathbbm{1} = \sum_{\widehat{g} \in G^{\cR}} P_{\widehat{g}}\,. 
\end{equation}
Moreover, notice that these projection are diagonal in the \emph{group} basis as the plaquette operators, so that $B_{p} P_{\widehat{g}} = P_{\widehat{g}} B_{p}$ for every plaquette $p$, and $\Tr_{\cR}(QP_{\widehat{g}}) = \Tr_{\cR}(P_{\widehat{g}}Q) = \Tr_{\cR}(P_{\widehat{g}}QP_{\widehat{g}})$ for every operator $Q$. This allows us to rewrite 
\begin{equation}\label{equa:marginalsFirstDecomp} 
\Tr_{\cR}\left(  \prod_{p \in \cP_\cR} e^{\beta B_{p}} \prod_{s \in \cS_\cR} e^{\beta A_{s}} \right) = \sum_{\widehat{g} \in G^\cR} \Tr_{\cR}\left(  \prod_{p \in \cP_\cR} e^{\beta B_{p}} P_{\widehat{g}} \right) \cdot \Tr_{\cR}\left(  \prod_{s \in \cS_\cR} e^{\beta A_{s}} P_{\widehat{g}} \right) \,. 
\end{equation} 
Observe that the two the factors that appear in each summand are positive semidefinite and commute with each other. We first prove the following alternative expression for them.

\begin{Lemm}\label{Lemm:starOperatorsTraceRewritten}
For every region $\cR$ being connected by stars, and every $\widehat{g}: \cR \to G$
\begin{equation}\label{equa:starOperatorsTraceRewritten} 
\Tr_{\cR}\left(  \prod_{s \in \cS_\cR} e^{\beta A_{s}} P_{\widehat{g}} \right)  = \left( (1+\gamma_{\beta})^{|\cS_\cR|} - \gamma_{\beta}^{|\cS_\cR|} \right) \, \Tr_{\cR}(P_{\widehat{g}}) + (e^{\beta} - 1)^{|\cS_\cR|} \Tr_{\cR}\left( \prod_{s \in \cS_\cR} A_{s} \, P_{\widehat{g}} \right)\,. 
\end{equation}
\end{Lemm}

\begin{proof}
Since each $A_{p}$ is a projection, we can rewrite  $e^{\beta A_{s}} = \mathbbm{1} + (e^{\beta} - 1) A_{s}$,  and expand
\[ \prod_{s \in \cS_\cR} e^{\beta A_{s}} = \sum_{W \subset \cS_\cR} (e^{\beta} - 1)^{|W|} \prod_{s \in W} A_{s}\,, \]
where, for $W=\emptyset$, the last product is understood to be the identity. Then,
\begin{equation}\label{equa:starOperatorsTraceRewrittenAux1} 
\Tr_{\cR}\left( \prod_{s \in \cS_\cR} e^{\beta A_{s}} \, P_{\widehat{g}} \right) = \sum_{W \subset \cS_\cR}(e^{\beta} - 1)^{|W|} \Tr_{\cR}\left( \prod_{s \in W} A_{s} \, P_{\widehat{g}} \right)\,. 
\end{equation} 
We claim that for every $W \subset \cS_\cR$ with $W \neq \cS_\cR$,
\begin{equation}\label{equa:starOperatorsTraceRewrittenAux2} 
\Tr_{\cR}\left( \prod_{s \in W} A_{s} \, P_{\widehat{g}} \right) = \frac{1}{|G|^{|W|}} \Tr_{\cR}(P_{\widehat{g}}) \,. 
\end{equation}
The claim is obvious if $W=\emptyset$, so let us assume that $\emptyset \subsetneq W \subsetneq \cS_{\cR}$, and fix two stars, one in $W$ and another one not in $W$. By the hypothesis on $\cR$ we can connect them with a sequence of neighbouring stars, that is, we can construct a sequence $s_{0}, \ldots , s_{n}$ such that $s_{0} \in W$, $s_{n} \notin W$ and $ \emptyset \neq s_{j} \cap s_{j+1} \subset \cR$ for every $j=0,1, \ldots n-1$.  Then, there must be  two neighbouring stars $s_{k}, s_{k+1}$ along the sequence such that $s_{k} \in W$, $s_{k+1} \notin W$ and its common edge $e_{k}$ belongs to $\cR$. Actually $s_{k}$  will necessarily be the only star in $W$ containing $e_{k}$. Thus, 
\[ \Tr_{\{e_{k}\}}\left( \prod_{s \in W} A_{s} P_{\widehat{g}} \right)  =  \left(\prod_{s \in W, s \neq s_{k}} A_{s}\right) \, \Tr_{\{e_{k}\}}\left(  A_{s_{k}} P_{\widehat{g}} \right) \,. \]
Finally, observe that if we expand the star $A_{s_{k}} = \frac{1}{|G|} \sum_{g \in G} A_{s_{k}}(g)$, then  since $\Tr(L^{g} \dyad{h}) = \Tr(R^{g} \dyad{h}) = \delta_{g,1}$, we have that
\[ \Tr_{\{e_{k}\}}\left( \prod_{s \in W} A_{s} P_{\widehat{g}} \right)  = \frac{1}{|G|} \Big(\prod_{\substack{s \in W\\ s \neq s_{k}}} A_{s}\Big) \Tr_{\{e_{k}\}}(P_{\widehat{g}}) =  \frac{1}{|G|}  \Tr_{\{e_{k}\}}\left( \prod_{s \in W , s \neq s_{k}} A_{s} P_{\widehat{g}} \right)  \,. \]
Applying $\Tr_{\cR \setminus \{ e_{k}\}}$ to the above expression, we deduce that
\[ \Tr_{\cR}\left( \prod_{s \in W} A_{s} P_{\widehat{g}} \right)  = \frac{1}{|G|}  \Tr_{\cR}\left( \prod_{s \in W , s \neq s_{k}} A_{s} P_{\widehat{g}} \right)  \,. \]
We have erased star $s_{k}$ from the productory at the expenses of adding a factor $1/|G|$. Iterating this process, we can erase every star in $W$ to get \eqref{equa:starOperatorsTraceRewrittenAux2}, proving the claim. 

Next, by applying this identity to all summands of  \eqref{equa:starOperatorsTraceRewrittenAux1}, we get that
\[ \Tr_{\cR}\left( \prod_{s \in \cS_\cR} e^{\beta A_{s}} P_{\widehat{g}} \right) = \sum_{\substack{W \subset \cS_\cR \\ W \neq \cS_\cR}} \gamma_{\beta}^{|W|} \Tr_{\cR}(
P_{\widehat{g}}) + (e^{\beta} - 1)^{|\cS_\cR|} \Tr_{\cR}\left( \prod_{s \in \cS_\cR} A_{s} P_{\widehat{g}} \right) \,. \]
Finally, applying the binomial formula to the first summatory leads to \eqref{equa:starOperatorsTraceRewritten}. 
\end{proof}

\begin{Lemm}\label{Lemm:starOperatorsTraceRewritten2}
For every region $\cR$ connected by stars and plaquettes and every $\widehat{g}: \cR \to G$, we have
\begin{equation}\label{equa:starOperatorsTraceRewritten2} 
0\leq \Tr_{\cR}\left( \prod_{s \in \cS_\cR} A_{s} \, P_{\widehat{g}} \right) \leq \frac{|G|}{|G|^{|\cS_\cR|}} \, \mathbbm{1}\,.
\end{equation}
\end{Lemm}

\begin{proof}
 The first inequality obvious, so we focus on the second one. If we expand each factor $A_{s} = \frac{1}{|G|} \sum_{a \in G} A_{s}(a)$, we can rewrite 
\[ \textstyle \Tr_{\cR}\left( P_{\widehat{g}} \prod_{s \in \cS_\cR} A_{s}  \right) = \sum_{\widehat{a}: \cS_\cR \to G}\Tr_{\cR}\left( P_{\widehat{g}} \prod_{s \in \cS_\cR}A_{s}(\widehat{a}(s)) \right)\,. \] 
In each summand of the previous expression, since $P_{\widehat{g}}$ is a one-dimensional orthogonal projection when considered on region $\cR$, the trace over $\cR$ will be either zero or one. Actually, since $ \prod_{s \in \cS_\cR}A_{s}(\widehat{a}(s))$ is a tensor product of operators acting on each edge, we can calculate the trace over $\cR$ by examining the trace over each edge. Observe that if we take $e \in \cR$, and denote by $s_{1}$ and $s_{2}$ the two stars of $\cS_\cR$ containing $e$, then the trace of the corresponding tensor factor is
\[
\Tr_{e}\left( L^{\widehat{a}(s_{1})} R^{\widehat{a}(s_{2})^{-1}} \dyad{\widehat{g}(e)} \right)  = \braket{\widehat{g}(e)|\widehat{a}(s_1) \, \widehat{g}(e) \, \widehat{a}(s_2)^{-1}}\,.
\]
Therefore, the trace on site $e$ will be one (i.e. nonzero) if and only if the group elements corresponding to stars $s_{1}$ and $s_{2}$ are related via conjugation with $\widehat{g}(e)$:
\begin{equation}\label{equa:conjugationCondition} 
\widehat{a}(s_{2}) =  \widehat{g}(e)^{-1} \widehat{a}(s_{1}) \widehat{g}(e)\,.  
\end{equation} 
As a consequence, the trace over $\cR$ will be one (i.e. nonzero) if, and only if, the elements of $\widehat{a}$ and $\widehat{g}$ satisfy the compatibility condition \eqref{equa:conjugationCondition} for every pair of neighbouring stars $s_{1}$ and $s_{2}$ with common edge $e$. Fixed any $\widehat{g}: \cR \to G$, let us denote by $S_{\widehat{g}}$ the set of elements $\widehat{a}$ that satisfy this property. We can then prove the following properties:
\begin{enumerate}
\item[$(i)$] \emph{Considering $G^{\cS_\cR}$ as a group with the pointwise product, we have that $S_{\widehat{g}}$ is a subgroup}.  Indeed, for any two elements $\widehat{a}, \widehat{b}$ in $S_{\widehat{g}}$, we have that $\widehat{c}:=\widehat{a} \cdot \widehat{b}$ also belongs to $S_{\widehat{g}}$, since
\begin{align*} 
\widehat{g}(e)^{-1} \widehat{c}(s_{1}) \widehat{g}(e) 
& = \widehat{g}(e)^{-1} \widehat{a}(s_{1}) \widehat{b}(s_{1}) \widehat{g}(e)\\ 
& = \widehat{g}(e)^{-1} \widehat{a}(s_{1})  \widehat{g}(e) \widehat{g}(e)^{-1} \widehat{b}(s_{1}) \widehat{g}(e)\\ 
& = \widehat{a}(s_{2}) \widehat{b}(s_{2})\\ 
& = \widehat{c}(s_{2})\,, 
\end{align*}
for every pair of neighbouring stars $s_{1}, s_{2} \in \cS_\cR$ with common edge $e \in \cR$. It is also clear that the identity  $\widehat{1} \in G^{\cS_\cR}$ also satisfies the compatibility condition and so belongs to $S_{\widehat{g}}$. Finally, if $\widehat{a} \in S_{\widehat{g}}$, then its inverse $\widehat{a}^{-1}$, given by $\widehat{a}^{-1}(h) = \widehat{a}(h)^{-1}$ for every $h \in G$, also belongs to $S_{\widehat{g}}$, since
\[ \widehat{g}(e)^{-1} \widehat{a}^{-1}(s_{1}) \widehat{g}(e) = \widehat{g}(e)^{-1} \widehat{a}(s_{1})^{-1} \widehat{g}(e) = ( \widehat{g}(e)^{-1} \widehat{a}(s_{1}) \widehat{g}(e) )^{-1} = \widehat{a}(s_{2})^{-1} = \widehat{a}^{-1}(s_{2}) \,. \]
\item[$(ii)$] \emph{The subgroup $S_{\widehat g}$ contains at most $|G|$ elements}. It is enough to prove that fixed any star $s_{0} \in \cS_\cR$, if two elements $\widehat{a}$ and $\widehat{b}$ of $S_{\widehat g}$  coincide on $s_{0}$, then they must coincide on every star $s \in \cS_\cR$. To check this, let us fix an arbitrary star $s'$. Since $\cR$ is connected by stars, we can find a path of stars $s_{0}, s_{1}, \ldots, s_{n-1}, s_{n} = s'$ such that $s_{j} \cap s_{j+1} = \{ e_{j}\} \subset \cR$. Then, using the compatibility condition on each pair of consecutive stars, we get 
\[ \widehat{a}(s_{n}) = \widehat{g}(e_{n})^{-1} \widehat{a}(s_{n-1}) \widehat{g}(e_{n}) = \ldots  =  (\widehat{g}(e_{n})^{-1} \cdot \ldots \cdot \widehat{g}(e_{1})^{-1}) \, \widehat{a}(s_{0}) \, (\widehat{g}(e_{1}) \cdot \ldots \cdot \widehat{g}(e_{n}))\,.  \]
A similar formula holds for $\widehat{b}$, and since $\widehat{a}(s_{0}) = \widehat{b}(s_{0})$, we conclude that $\widehat{a}(s_{n}) = \widehat{b}(s_{n})$.
\item[$(iii)$] \emph{The operator given by
\[ A_{\cR}(\widehat{g}) := \frac{1}{|S_{\widehat{g}}|} \sum_{\widehat{a} \in S_{\widehat{g}}} \prod_{s \in \cS_\cR} A_{s}(\widehat{a}(s))  \]
defines an orthogonal projection}. To verify that it is Hermitian, notice that star operators satisfy $A_{s}(h)^{\dagger} = A_{s}(h^{-1})$, and also that $\widehat{a} \in S_{\widehat{g}}$ if, and only if, $\widehat{a}^{-1} \in S_{\widehat{g}}$. To show that it is idempotent, note that this follows from the fact that star operators satisfy $A_{s}(h) A_{s}(h') = A_{s}(hh')$.
\end{enumerate}
Using the above properties, we can rewrite and upper estimate 
\[ \textstyle \Tr_{\cR}\left(  P_{\widehat{g}} \prod_{s \in \cS_\cR} A_{s}\right) = \frac{|S_{\widehat{g}}|}{|G|^{|\cS_\cR|}}  \Tr_{\cR}\left( P_{\widehat{g}} A_{\cR}(\widehat{g})\right) \leq \frac{|S_{\widehat{g}}|}{|G|^{|\cS_\cR|}} \, \Tr_{\cR}(P_{\widehat{g}}) \leq \frac{|G|}{|G|^{|\cS_\cR|}} \, \Tr_{\cR}(P_{\widehat{g}}) \,. \]
Applying this upper estimate to \eqref{equa:starOperatorsTraceRewritten} from Lemma~\ref{Lemm:starOperatorsTraceRewritten}, we conclude the second inequality.
\end{proof}

Finally, we need the following auxiliary result. 

\begin{Lemm}\label{Lemm:plaquetteOperatorsTraceRewritten}
For every region $\cR$ connected by plaquettes, it holds that
\begin{equation}\label{equa:plaquetteOperatorsTraceRewritten}
 \Tr_{\cR}\left( \prod_{p \in \cP_\cR} e^{\beta B_{p}} \right) = \left( (1+\gamma_{\beta})^{|\cP_\cR|} - \gamma_{\beta}^{|\cP_\cR|} \right) \Tr_{\cR}(\mathbbm{1}) + (e^{\beta} - 1)^{|\cP_\cR|} \Tr_{\cR}\left( \prod_{p \in \cP_\cR} B_{p}\right)\,.
\end{equation}
\end{Lemm}

\begin{proof}
We expand each factor $e^{\beta B_{p}} = \mathbbm{1} + (e^{\beta} - 1) B_{p}$, as in the proof of Lemma~\ref{Lemm:starOperatorsTraceRewritten}, to get
\begin{equation}\label{equa:plaquetteOperatorsTraceRewrittenAux1}  
\Tr_{\cR}\left( \prod_{p \in \cP_\cR} e^{\beta B_{p}} \right) = \sum_{W \subset \cP_\cR} (e^{\beta} - 1)^{|W|}\Tr_{\cR}\left( \prod_{p \in W} B_{p} \right)\,, 
\end{equation}
where, for $W=\emptyset$, the last product is understood to be the identity. We claim that if $W \neq \cP_\cR$, then
\begin{equation}\label{equa:plaquetteOperatorsTraceRewrittenAux2} 
\Tr_{\cR}\left( \prod_{p \in W} B_{p} \right) = \frac{1}{|G|^{|W|}}\Tr_{\cR}\left[ \mathbbm{1} \right] \,.
\end{equation}
Since $\cR$ is connected by plaquettes, there exists an edge $e_{0} \in \cR$ such that exactly one of the two plaquettes containing $e_{0}$ belongs to $W$, say $p_{0}$, so that
\[ \Tr_{\cR}\left( \prod_{p \in W} B_{p}\right) = \Tr_{\cR \setminus \{ e_{0}\}}\left( \Tr_{\{e_{0}\}}\left( \prod_{p \in W} B_{p} \right)\right) = \Tr_{\cR \setminus \{ e_{0}\}}\left( \prod_{\substack{p \in W \\ p \neq p_{0}}} B_{p} \, \Tr_{e_{0}}\left( B_{p_{0}}\right) \right) \,. \]
At this point, we use that $\Tr_{e_{0}}(B_{p_{0}}) = \frac{1}{|G|} \Tr_{e_{0}}(\mathbbm{1})$, which can be easily verified using the explicit description of $B_{p_{0}}$. This leads to
\[ \Tr_{\cR}\left( \prod_{p \in W} B_{p}\right)  = \frac{1}{|G|}\Tr_{\cR}\left( \prod_{p \in W, p \neq p_{0}} B_{p}\right)\,.  \]
Iterating this process, we can \emph{erase} every element of $W$ at the expenses of adding a scalar factor $1/|G|$, and conclude that the claim \eqref{equa:plaquetteOperatorsTraceRewrittenAux2} holds. Replacing this identity in \eqref{equa:plaquetteOperatorsTraceRewrittenAux1}, we can rewrite
\begin{align*}
\Tr_{\cR}\left( \prod_{p \in \cP_\cR} e^{\beta B_{p}} \right)& = \sum_{W \subset \cP_\cR, W \neq \cP_\cR} \gamma_{\beta}^{|W|}\Tr_{\cR}\left( \mathbbm{1} \right)  \, +\, (e^{\beta} - 1)^{|\cP_\cR|}\Tr_{\cR}\left( \prod_{p \in \cP_\cR} B_{p} \right)\\[2mm]
 & = \left( (1+\gamma_{\beta})^{|\cP_\cR|}  - \gamma_{\beta}^{|\cP_\cR|}\right) \, \Tr_{\cR}(\mathbbm{1}) \, +\, (e^{\beta} - 1)^{|\cP_\cR|} \Tr_{\cR}\left( \prod_{p \in \cP_\cR} B_{p} \right)\,,
\end{align*}
where in the last identity we have used the binomial formula.
\end{proof}

\begin{Lemm}\label{Lemm:plaquetteOperatorsTraceEstimates}
For every region $\cR$ being connected by plaquettes, it holds that
\begin{equation}\label{equa:plaquetteOperatorsTraceEstimates}
0 \leq \Tr_{\cR}\left( \prod_{p \in \cP_\cR} B_{p}\right) \leq \frac{|G|^{|\cR|+1}}{|G|^{|\cP_\cR|}} \mathbbm{1}
\end{equation}
\end{Lemm}

\begin{proof}
The first inequality is obvious, so let us focus on the upper bound. Let $\overline{\cR}$ denote the union of all plaquettes in $\cP_\cR$. We can use  the completeness relation \eqref{equa:completenessProjections} on $\overline{\cR}$ to rewrite
\begin{equation*} 
 \prod_{p \in \cP_\cR} B_{p} = \left(  \prod_{p \in \cP_\cR} B_{p} \right) \sum_{\widehat{h} \in  G^{\overline{\cR}}} P_{\widehat{h}}   = \sum_{\widehat{h} \in  G^{\overline{\cR}}} \,\, \left(\prod_{p \in \cP_\cR} B_{p}\right) P_{\widehat{h}} \,. 
\end{equation*}
Recalling the definition of the plaquette operator from \eqref{equa:plaquetteOperatorDefinition}, where $\delta_{1}: G \to \{ 0,1\}$ is the characteristic function taking the value one at $1$ and zero elsewhere, we know that $B_{p}$ acts diagonally on the computational basis. Therefore, each $B_{p}$ satisfies
\begin{align*} 
\begin{tikzpicture}[equation]   
  \draw[step=1.0,gray,thin] (-0.5,-0.5) grid (1.5,1.5);
  \draw[ultra thick] (0,0) -- (1, 0) -- (1,1) -- (0,1) -- (0,0);
\draw (0.5, 0.5) node {$p$};
\draw (0.5, 1.4) node {$e_{1}$};
\draw (-0.3, 0.5) node {$e_{2}$};
\draw (0.5, -0.4) node {$e_{3}$};
\draw (1.3, 0.5) node {$e_{4}$};
\end{tikzpicture}
\hspace{1cm}
B_{p} P_{\widehat{h}} = \delta_{1}(\, \widehat{h}_{p}) P_{\, \widehat{h}} \quad \mbox{where} \quad \widehat{h}_{p}:=\widehat{h}(e_{1}) \, \widehat{h}(e_{2}) \, \widehat{h}(e_{3})^{-1} \, \widehat{h}(e_4)^{-1} \,.
\end{align*}
with $e_{1}, e_{2}, e_{3}, e_{4}$ being the edges forming the plaquette $p$ (ordered counterclockwise starting from the top, as in the diagram above). Therefore, we can simplify
\begin{equation*} 
 \prod_{p \in \cP_\cR} B_{p}  = \sum_{\widehat{h} \in G^{\overline{\cR}}} \,\, \left(\prod_{p \in \cP_\cR} \delta_{1}( \widehat{h}_{p})\right)  P_{\, \widehat{h}}\,. 
\end{equation*}
Next, we apply the partial trace $\Tr_{\cR}(\cdot)$ on the previous expression. Notice that $\Tr_{\cR}(P_{\widehat{h}})$ is a projection supported on $\partial \cR:=\overline{\cR} \setminus \cR$. Specifically,
\[ \textstyle \Tr_{\cR}\left( P_{\,\widehat{h}}\right) = \Tr_{\cR}\left( \bigotimes_{e \in \overline{\cR}} \ket{\widehat{h}(e)}\bra{\widehat{h}(e)} \right) = \bigotimes_{e \in \partial \cR} \ket{\widehat{h}(e)}\bra{\widehat{h}(e)} = P_{\, \widehat{h}|_{\partial \cR}}\,. \]
By grouping together summands $\widehat{h}$ that coincide on $\partial{\cR}$, we then get
\begin{equation}\label{equa:plaquetteOperatorsTraceRewrittenAux4} 
\Tr_{\cR}\left( \prod_{p \in \cP_\cR} B_{p}\right) =  \sum_{\widehat{f} \in G^{\partial \cR}}\,\, \left(\, \sum_{\widehat{h} \in G^{\overline{\cR}} \colon \widehat{h}|_{\partial \cR} = \widehat{f}} \,\,\, \prod_{p \in \cP_\cR} \delta_{1}(\, \widehat{h}_{p}) \right) P_{\widehat{f}}\,.
\end{equation}
% \begin{equation}\label{equa:plaquetteOperatorsTraceRewrittenAux4} 
% \Tr_{\cR}\left( \prod_{p \in \cP_\cR} B_{p}\right) =  \sum_{\widehat{f}: \partial{\cR} \to G}\,\, \left({\textstyle\sum_{\substack{\widehat{h}:\overline{\cR} \to G \\[1mm] \widehat{h}|_{\partial{\cR}} = \widehat{f}}}} \,\,\, \prod_{p \in \cP_\cR} \delta_{1}\left( \widehat{h}_{e_{1}^p} \widehat{h}_{e_{2}^p}  \widehat{h}_{e_{3}^p}^{-1}  \widehat{h}_{e_{4}^p}^{-1}\right) \right) P_{\widehat{f}}\,.
% \end{equation}
Therefore, we can estimate its supremum norm by
\[  \left\| \Tr_{\cR}\left( \prod_{p \in \cP_\cR} B_{p}\right)  \right\|_{\infty} \leq \sup_{\widehat{f} \in G^{\partial \cR}} \,  \left( \sum_{\widehat{h} \in G^{\overline{\cR}} \colon \widehat{h}|_{\partial \cR} = \widehat{f}} \,\,\, \prod_{p \in \cP_\cR} \delta_{1}(\, \widehat{h}_{p})  \right)\,. \]
% \[  \left\| \Tr_{\cR}\left( \prod_{p \in \cP_\cR} B_{p}\right)  \right\|_{\infty} \leq \sup_{\widehat{f}: \partial \cR \to G} \,  \left( {\textstyle \sum_{\substack{\widehat{h}:\overline{\cR} \to G \\[1mm] \widehat{h}|_{\partial{\cR}} = \widehat{f}}}} \,\,\, \prod_{p \in \cP_\cR} \delta_{1}\left( \widehat{h}_{e_{1}^p} \widehat{h}_{e_{2}^p}  \widehat{h}_{e_{3}^p}^{-1}  \widehat{h}_{e_{4}^p}^{-1}\right)  \right)\,. \]
Notice that in each sum, we are fixing some boundary conditions $\widehat{h}(e)$ for $e \in \partial \cR$, so the sum amounts to assigning an element $\widehat{h}(e) \in G$ to every edge $e \in \cR$. This can be roughly estimated by
\[ \sum_{\widehat{h} \in G^{\cR}} \, \prod_{p \in \cP_\cR} \delta_{1}(\, \widehat{h}_{p}) \leq \sum_{\widehat{h} \in G^{\cR}} 1 = |G|^{|\cR|}\,. \]
However, this counting does not take into account that for each $p \in \cP_\cR$ there is a constraint $\widehat{h}_{p}=1$. Naively speaking, each constraint reduces the degrees of freedom by a factor of $|G|$ if they are independent, but since the region is connected by plaquettes, there is one redundant constrain. Therefore, we expect that
\[ \sum_{\widehat{h} \in G^{\cR}} \,\, \prod_{p \in \cP_\cR} \delta_{1}(\, \widehat{h}_{p})    \leq \frac{|G|^{|\cR|}}{|G|^{|\cP_\cR|-1}}. \]
To prove this, we just need the following easy observation on the delta function $\delta_{1}$:
\begin{equation}\label{equa:plaquetteOperatorsTraceRewrittenAux5}
\sum_{h \in G} \delta_{1}(uhu') \delta_{1}(vh^{-1}v') = \delta_{1}(u'uv'v) \quad \quad \mbox{for every $u,u',v,v' \in G$}\,.
\end{equation}
Let us consider first the case in which $\cR$ consists of only one edge $\cR=\{e_1\}$, and $\cP_\cR =\{ p_0,p_1\}$. If we fix the values of $\widehat{h}(e)$ for $e \in \overline{\cR} \setminus \{ e_{1}\}$,  we can then apply \eqref{equa:plaquetteOperatorsTraceRewrittenAux5} to simplify
\begin{align*}
\sum_{\widehat{h} \in G^{\cR}} \,\, \prod_{p \in \cP_\cR} \delta_{1}(\,\widehat{h}_{p}) = 
 \sum_{\widehat{h}(e_1) \in G} \delta_{1}(\,\widehat{h}_{p_0}) \delta_{1}(\,\widehat{h}_{p_1}) 
 & = \sum_{\widehat{h}(e_1) \in G} \delta_{1}\left(\ldots \widehat{h}(e_1) \ldots \right) \, \delta_{1}\left(\ldots \widehat{h}(e_1)^{-1} \ldots \right)\\[2mm]
 & = \delta_{1}(\ldots) \leq 1\,.
 \end{align*}
This idea can be applied iteratively on larger regions. Since $\cR$ is connected by plaquettes, we can order all plaquettes of $\cP_\cR$ as a finite sequence  $p_0, p_1, \dots, p_n$ such that plaquette $p_k$  shares (at least) one edge, say $e_{k} \in \cR$, with $\bigcup_{j=1}^{k-1} p_j$ for every $k \geq 1$. If we fix the values $\widehat{h}(e)$ for $e \in \overline{\cR} \setminus \{ e_{1}, \ldots e_{n}\}$, then
\begin{align*}
 \sum_{\widehat{h}(e_1), \ldots, \widehat{h}(e_n) \in G} \,\,   \prod_{j=0}^{n} \delta_{1}(\, \widehat{h}_{p_j}) &  =
\sum_{\widehat{h}(e_2), \ldots, \widehat{h}(e_n)} \left(\sum_{\widehat{h}(e_1) \in G} \delta_{1}(\ldots \widehat{h}(e_1) \ldots ) \delta_{1}(\ldots \widehat{h}(e_1)^{-1}\ldots)\right) \prod_{j=2}^{n} \delta_{1}(\widehat{h}_{p_j})\\
& =
\sum_{\widehat{h}(e_2), \ldots, \widehat{h}(e_n)} \delta_{1}(\ldots \widehat{h}(e_2)\ldots) \prod_{j=2}^{n} \delta_{1}(\widehat{h}_{p_j})\\
& =
\sum_{\widehat{h}(e_3), \ldots, \widehat{h}(e_n)} \left(\sum_{\widehat{h}(e_2) \in G} \delta_{1}(\ldots \widehat{h}(e_2) \ldots ) \delta_{1}(\ldots \widehat{h}(e_2)^{-1}\ldots)\right) \prod_{j=3}^{n} \delta_{1}(\widehat{h}_{p_j})\\[2mm]
& =
\sum_{\widehat{h}(e_3), \ldots, \widehat{h}(e_n)} \delta_{1}(\ldots \widehat{h}(e_{3})\ldots) \prod_{j=3}^{n} \delta_{1}(\widehat{h}_{p_j})\\
& = \ldots\\[2mm]
& = \delta_{1}(\ldots)\,.
\end{align*}
Therefore, if we fix the boundary values $\widehat{h}(e) \in G$ for $e \in \partial \cR$, we can estimate
\begin{align*}
\sum_{\widehat{h} \in G^{\cR}} \,\, \prod_{j=0}^{n} \delta_{1}(\widehat{h}_{p_j}) 
 = \sum_{\widehat{h}  \in G^{\cR \setminus \{ e_{1}, \ldots, e_{n}\}}} \,\, \sum_{\widehat{h}(e_{1}), \ldots \widehat{h}(e_n)} \,\, \prod_{j=0}^{n} \delta_{1}(\widehat{h}_{p_j}) & \leq \sum_{\widehat{h}  \in G^{\cR \setminus \{ e_{1}, \ldots, e_{n}\}}} \delta_{1}(\ldots)\\[2mm] 
& \leq |G^{\cR \setminus \{ e_{1}, \ldots, e_{n}\}}| =  |G|^{|\cR|-|\cP_\cR|+1}\,. 
\end{align*}
This finishes the proof.
\end{proof}

We can now finish the section with the proof of the result.
\begin{proof}[Proof of Theorem~\ref{thm:marginals-decayQDM}]
 Combining Lemmas~\ref{Lemm:starOperatorsTraceRewritten} and~\ref{Lemm:starOperatorsTraceRewritten2}, we can estimate
for every $\widehat{g} \in G^{\cR}$
\[
\left( (1+\gamma_{\beta})^{|\cS_\cR|} - \gamma_{\beta}^{|\cS_\cR|}   \right) \,  \mathbbm{1} \leq 
\Tr_{\cR}\left(  \prod_{s \in \cS_\cR} e^{\beta A_{s}} P_{\widehat{g}} \right)  \leq  \left( (1+\gamma_{\beta})^{|\cS_\cR|} - \gamma_{\beta}^{|\cS_\cR|}  + |G| \gamma_{\beta}^{|\cS_\cR|}  \right) \,  \mathbbm{1}\,.   \]
Applying this to the decomposition given in \eqref{equa:marginalsFirstDecomp}, we deduce that
\begin{multline}\label{equa:marginals-decayQDMAux1} 
 \left( (1+\gamma_{\beta})^{|\cS_\cR|} - \gamma_{\beta}^{|\cS_\cR|}   \right) \Tr_{\cR}\left( \prod_{p \in \cP_\cR} e^{\beta B_{p}} \right)\\ \leq \Tr_{\cR}\left(  \prod_{p \in \cP_\cR} e^{\beta B_{p}} \prod_{s \in \cS_\cR} e^{\beta A_{s}} \right) \leq\\  \left( (1+\gamma_{\beta})^{|\cS_\cR|} - \gamma_{\beta}^{|\cS_\cR|}  + |G| \gamma_{\beta}^{|\cS_\cR|}  \right) \, \Tr_{\cR}\left( \prod_{p \in \cP_\cR} e^{\beta B_{p}} \right)\,.
\end{multline}
But, from the combination of Lemmas~\ref{Lemm:plaquetteOperatorsTraceRewritten} and
\ref{Lemm:plaquetteOperatorsTraceEstimates}, we know that
\[ |G|^{|\cR|} \left( (1+\gamma_{\beta})^{|\cP_\cR|} - \gamma_{\beta}^{|\cP_\cR|} \right)  \mathbbm{1} \leq \Tr_{\cR}\left( \prod_{p \in \cP_\cR} e^{\beta B_{p}} \right) \leq |G|^{|\cR|} \left( (1+\gamma_{\beta})^{|\cP_\cR|} - \gamma_{\beta}^{|\cP_\cR|} +  |G|\gamma_{\beta}^{|\cP_\cR|} \right)\mathbbm{1}\,. \]
Thus, applying these inequalities to \eqref{equa:marginals-decayQDMAux1}, we conclude that
\begin{multline}\label{equa:marginals-decayQDMAux2} 
 |G|^{|\cR|} 
 \left( (1+\gamma_{\beta})^{|\cP_\cR|} - \gamma_{\beta}^{|\cP_\cR|} \right) \left( (1+\gamma_{\beta})^{|\cS_\cR|} - \gamma_{\beta}^{|\cS_\cR|}   \right) \, \mathbbm{1} \\ 
 \leq 
 \Tr_{\cR}\left(  \prod_{p \in \cP_\cR} e^{\beta B_{p}} \prod_{s \in \cS_\cR} e^{\beta A_{s}} \right) 
 \leq\\  
  |G|^{|\cR|} 
 \left( (1+\gamma_{\beta})^{|\cS_\cR|} - \gamma_{\beta}^{|\cS_\cR|}  + |G| \gamma_{\beta}^{|\cS_\cR|}  \right) \, 
 \left( (1+\gamma_{\beta})^{|\cP_\cR|} - \gamma_{\beta}^{|\cP_\cR|} +  |G|\gamma_{\beta}^{|\cP_\cR|} \right) \, \mathbbm{1}\,.
\end{multline}
Taking $\kappa_{\cR} = |G|^{|\cR|} (1+\gamma_{\beta})^{|\cP_\cR|} (1+\gamma_{\beta})^{|\cS_\cR|}$ as a common multiplicative factor in both the lower and upper bounds, we conclude the result.
\end{proof}

\end{document}